%% file: QKD_report.tex
\documentclass[a4paper,onecolumn,11pt,accepted=2022-07-25,allowtoday]{quantumarticle}
\pdfoutput=1

\usepackage{multirow}
\usepackage{pdflscape}
\usepackage{cite}
\usepackage{lineno}
\usepackage{ifthen}
\usepackage{makeidx}
\usepackage{amsmath,amssymb,amstext} 
\usepackage[pdftex]{graphicx} 
\usepackage{fullpage}
\usepackage{datetime}
\usepackage{pdfpages}
\numberwithin{equation}{section}
\usepackage[numbers]{natbib}
\usepackage{mathtools}

\DeclarePairedDelimiter\ket{\lvert}{\rangle}
\DeclarePairedDelimiterX\braket[2]{\langle}{\rangle}{#1 \delimsize\vert #2}

\usepackage[pagebackref=true]{hyperref} 
\hypersetup{
    plainpages=false,       
    pdfpagelabels=true,     
    bookmarks=true,         
    unicode=false,          
    pdftoolbar=true,        
    pdfmenubar=true,        
    pdffitwindow=false,     
    pdfstartview={FitH},    
    pdftitle={
Robust Interior Point Method for  Quantum Key Distribution Rate Computation
},    
    pdfnewwindow=true,      
    colorlinks=true,        
    linkcolor=blue,         
    citecolor=green,        
    filecolor=magenta,      
    urlcolor=cyan           
}

\usepackage{adjustbox}  

\usepackage{comment}

\usepackage{float}
\usepackage{bbm}
\usepackage{exscale}
\usepackage{tabularx}
\usepackage{syntonly}

\usepackage{algorithm}
\usepackage{algorithmic}
\usepackage{amssymb}
\usepackage{amsmath}
\usepackage{amsthm}
\usepackage{latexsym}
\usepackage{makeidx}
\usepackage{longtable}
\usepackage{placeins}
\usepackage[noabbrev]{cleveref}
\numberwithin{equation}{section}

\makeindex

\counterwithin{table}{section}
\numberwithin{equation}{section}
\numberwithin{figure}{section}

\usepackage{chngcntr}

\def\G{\Gamma}

\def\R{\mathbb{R}}
\def\Rp{\R_+}
\def\one{\mathbbm{1}}
\def\Sc{\mathbb{S}}
\def\Sn{\Sc^n}

\def\Snp{\Sc_+^n}

\def\Snpp{\Sc_{++}^n}

\def\H{\mathbb{H}}
\def\Hp{\mathbb{H}_+}
\def\Hpp{\mathbb{H}_{++}}

\def\Rn{\mathbb{R}^n}

\def\Rnt{\mathbb{R}^{n^2}}

\def\C{\mathbb{C}}

\def\Cnn{\mathbb{C}^{n\times n}}

\def\GN{{\bf GN\,}}
\def\GNp{{\bf GN}}

\def\eqref#1{{\normalfont(\ref{#1})}}

\newtheorem{theorem}{Theorem}[section]

\newtheorem{definition}[theorem]{Definition}

\newtheorem{assump}[theorem]{Assumption}
\newtheorem{prop}[theorem]{Proposition}

\newtheorem{cor}[theorem]{Corollary}
\newtheorem{corollary}[theorem]{Corollary}

\newtheorem{remark}[theorem]{Remark}

\newtheorem{lemma}[theorem]{Lemma}

\crefname{thm}{Theorem}{Theorems}
\Crefname{thm}{Theorem}{Theorems}
\crefname{assump}{Assumption}{Theorems}
\Crefname{assump}{Assumption}{Theorems}
\crefname{problem}{Problem}{Theorems}
\Crefname{problem}{Problem}{Theorems}
\crefname{conjecture}{Conjecture}{Theorems}
\Crefname{conjecture}{Conjecture}{Theorems}
\crefname{proposition}{Proposition}{Propositions}
\Crefname{proposition}{Proposition}{Propositions}
\crefname{prop}{Proposition}{Propositions}
\Crefname{prop}{Proposition}{Propositions}
\crefname{cor}{Corollary}{Corollaries}
\Crefname{cor}{Corollary}{Corollaries}
\crefname{lem}{Lemma}{Lemmas}
\Crefname{lem}{Lemma}{Lemmas}
\theoremstyle{definition}
\crefname{definition}{definition}{definitions}
\Crefname{definition}{Definition}{Definitions}
\crefname{defn}{definition}{definitions}
\Crefname{defn}{Definition}{Definitions}
\crefname{remark}{Remark}{Remarks}
\Crefname{remark}{Remark}{Remarks}
\crefname{rmk}{Remark}{Remarks}
\Crefname{rmk}{Remark}{Remarks}
\crefname{example}{Example}{Examples}
\Crefname{example}{Example}{Examples}
\crefname{align}{}{}
\Crefname{align}{}{}
\crefname{equation}{}{}
\Crefname{equation}{}{}

\makeatletter
\newcommand*\bigcdot{\mathpalette\bigcdot@{.5}}
\newcommand*\bigcdot@[2]{\mathbin{\vcenter{\hbox{\scalebox{#2}{$\m@th#1\bullet$}}}}}
\makeatother

\newcommand{\crr}{\color{red}}
\newcommand{\crb}{\color{blue}}
\newcommand{\crg}{\color{green}}

\newcommand{\textdef}[1]{\textit{#1}\index{#1}}

\newcommand{\<}{\langle}
\renewcommand{\>}{{\rangle}}

\newcommand{\cZ}{{\mathcal Z} }
\newcommand{\cR}{{\mathcal R} }

\newcommand{\cL}{{\mathcal L} }
\newcommand{\cA}{{\mathcal A} }
\newcommand{\cV}{{\mathcal V} }
\newcommand{\cG}{{\mathcal G} }
\newcommand{\mV}{{m_{\scriptscriptstyle V}}}

\newcommand{\dGN}{{d_{\scriptscriptstyle GN}}}

\newcommand{\MGV}{{\cM_{\scriptscriptstyle\!{\GV}}}}

\newcommand{\cM}{{\mathcal M} }
\newcommand{\cGUV}{{\cG_{\scriptscriptstyle UV}}}
\newcommand{\cZUV}{{\cZ_{\scriptscriptstyle UV}}}

\newcommand{\cMZ}{{\cM_{\scriptscriptstyle Z}}}
\newcommand{\WcZ}{\widehat \cZ}
\newcommand{\WcG}{\widehat \cG}
\newcommand{\cMrho}{{\cM_{\scriptscriptstyle \rho}}}
\newcommand{\mv}{{m_{\scriptscriptstyle V}}}
\newcommand{\gammaV}{{\gamma_{\scriptscriptstyle V}}}
\newcommand{\gV}{{\gamma_{\scriptscriptstyle V}}}
\newcommand{\GV}{{\Gamma_{\scriptscriptstyle V}}}
\newcommand{\cF}{{\mathcal F} }

\newcommand{\cH}{{\mathcal H} }

\newcommand{\CC}{{\mathcal C} }
\newcommand{\GG}{{\mathcal G} }

\newcommand{\cN}{{\mathcal N} }
\newcommand{\cS}{{\mathcal S} }

\newcommand{\cW}{{\mathcal W} }

\newcommand{\pdipp}{\textbf{p-d i-p}}

\newcommand{\QKDp}{\textbf{QKD}}
\newcommand{\QKD}{\textbf{QKD}\,}

\newcommand{\SDP}{\textbf{SDP}\,}

\newcommand{\SDPp}{\textbf{SDP}}
\newcommand{\FR}{\textbf{FR}\,}
\newcommand{\FRp}{\textbf{FR}}

\newcommand{\Rm}{{\R^m\,}}

\newcommand{\A}{{\mathcal A}}

\newcommand{\bbm}{\begin{bmatrix}}
\newcommand{\ebm}{\end{bmatrix}}
\newcommand{\bem}{\begin{pmatrix}}
\newcommand{\eem}{\end{pmatrix}}
\newcommand{\beq}{\begin{equation}}
\newcommand{\beqs}{\begin{equation*}}
\newcommand{\bet}{\begin{table}}
\newcommand{\eeq}{\end{equation}}
\newcommand{\eeqs}{\end{equation*}}
\newcommand{\beqr}{\begin{eqnarray}}

\newcommand{\Vd}{V_\delta}
\newcommand{\Vs}{V_\sigma}

\DeclareMathOperator{\face}{face}

\DeclareMathOperator{\nul}{null}
\DeclareMathOperator{\range}{range}

\DeclareMathOperator{\sym}{{\cS}}
\DeclareMathOperator{\sksym}{{\mathcal{SK}}}
\DeclareMathOperator{\trace}{{Tr}}

\DeclareMathOperator{\BlkDiag}{{BlkDiag}}
\DeclareMathOperator{\tr}{{Tr}}
\DeclareMathOperator{\diag}{{diag}}
\DeclareMathOperator{\Diag}{{Diag}}

\DeclareMathOperator{\svec}{{svec}}
\DeclareMathOperator{\Hvec}{{Hvec}}

\DeclareMathOperator{\sMat}{{sMat}}
\DeclareMathOperator{\Cvec}{{Cvec}}

\DeclareMathOperator{\HMat}{{HMat}}

\DeclareMathOperator{\relint}{{relint}}

\DeclareMathOperator{\rank}{{rank}}

\DeclareMathOperator{\real}{{\Re}}
\DeclareMathOperator{\imag}{{\Im}}

\DeclareMathOperator{\Tr}{{Tr}}

\newcommand{\nc}{\newcommand}
\nc{\arrow}{{\rm arrow\,}}
\nc{\Arrow}{{\rm Arrow\,}}
\nc{\BoDiag}{{\rm B^0Diag\,}}
\nc{\bodiag}{{\rm b^0diag\,}}

\nc{\Mm}{{\mathcal M}^{m} }
\nc{\Mmn}{{\mathcal M}^{mn} }
\nc{\Mnr}{{\mathcal M}_{nr} }
\nc{\Mnmr}{{\mathcal M}_{(n-1)r} }
\nc{\kwqqp}{Q{$^2$}P\,}
\nc{\kwqqps}{Q{$^2$}Ps}
\def\argmin{\mathop{\rm argmin}}

\nc{\notinaho}{(X,S)\in \overline{AHO}(\A)}
\nc{\inaho}{(X,S)\in AHO(\A)}

\newcommand{\bea}{\begin{eqnarray}}%
\newcommand{\eea}{\end{eqnarray}}%
\newcommand{\beas}{\begin{eqnarray*}}%
\newcommand{\eeas}{\end{eqnarray*}}%
\newcommand{\Int}{{\rm int\,}}

\newcommand{\im}{{\rm Im}\,}%
{}

\newcommand{\Hnp}[1][]{\,\mathbb{H}_+^{\ifthenelse{\equal{#1}{}}{n}{#1}}}
\newcommand{\Hkp}[1][]{\,\mathbb{H}_+^{\ifthenelse{\equal{#1}{}}{k}{#1}}}
\newcommand{\Hnpp}[1][]{\,\mathbb{H}_{++}^{\ifthenelse{\equal{#1}{}}{n}{#1}}}
\newcommand{\Hn}[1][]{\,\mathbb{H}^{\ifthenelse{\equal{#1}{}}{n}{#1}}}
\newcommand{\Hk}[1][]{\,\mathbb{H}^{\ifthenelse{\equal{#1}{}}{k}{#1}}}
\newcommand{\Dn}[1][]{\,\mathbb{D}^{\ifthenelse{\equal{#1}{}}{n}{#1}}}

\begin{document}

\title{Robust Interior Point Method for Quantum Key Distribution Rate Computation}

\author{\href{https://huhao.org}{\color{black}{Hao Hu}}}
\affiliation{Department of Combinatorics and Optimization,   Faculty of Mathematics, University of Waterloo, Waterloo, Ontario, Canada N2L 3G1}
\affiliation{Department of Mathematical Sciences, Clemson University, Clemson, SC, United States 29634}
\author{\href{https://uwaterloo.ca/combinatorics-and-optimization/about/people/j5im}{\color{black}{Jiyoung Im}}}
\affiliation{Department of Combinatorics and Optimization,   Faculty of Mathematics, University of Waterloo, Waterloo, Ontario, Canada N2L 3G1}
\author{Jie Lin}
\affiliation{Institute for Quantum Computing and Department of Physics and Astronomy, University of Waterloo, Waterloo, Ontario, Canada N2L 3G1}
\affiliation{Department of Electrical \& Computer Engineering, University of Toronto, Toronto, Ontario, Canada  M5S 3G4}
\orcid{0000-0003-1750-2659}
\author{Norbert L\"utkenhaus}
\affiliation{Institute for Quantum Computing and Department of Physics and Astronomy, University of Waterloo, Waterloo, Ontario, Canada N2L 3G1}
\orcid{0000-0002-4897-3376}
\author{\href{http://www.math.uwaterloo.ca/~hwolkowi/}{\color{black}{Henry Wolkowicz}}}
\affiliation{Department of Combinatorics and Optimization,   Faculty of Mathematics, University of Waterloo, Waterloo, Ontario, Canada N2L 3G1}

\date{\today}

  \maketitle


\begin{abstract} 
Security proof methods for quantum key distribution, \QKDp,
that are based on the numerical key rate calculation problem, are
powerful in principle. However, 
the practicality of the methods are limited by computational resources and the
efficiency and accuracy of the underlying algorithms for convex optimization. 
We derive a stable reformulation of the convex nonlinear semidefinite
programming, \SDPp, model for the
key rate calculation problems. We use this to develop an efficient,
accurate algorithm. The stable reformulation is based on novel forms of facial reduction, \FRp, for both
the linear constraints and nonlinear quantum relative entropy objective function.
This allows for a Gauss-Newton type interior-point approach that avoids
the need for perturbations to obtain strict feasibility, a technique
currently used in the literature. The result is high
accuracy solutions with theoretically proven lower bounds for the
original \QKD from the \FR stable reformulation. This provides novel
contributions for \FR for general \SDPp.

We report on empirical results that dramatically improve on speed and
accuracy, as well as solving previously intractable problems.
\end{abstract}



\tableofcontents
\listoftables
\listoffigures

\newpage

\section{Introduction}\label{sec:intro}

\index{\QKDp, quantum key distribution}

\textdef{Quantum key distribution, \QKDp} (see e.g.,~\cite{Scarani2009,
	Xu2020} for reviews), is 
a secure communication method that distributes a secret key between 
two honest parties (traditionally known as Alice and Bob),
even in the presence of an eavesdropper (traditionally called Eve). Those keys can be used for secure communication, authentication, secure multi-party computation and other cryptographic applications.
Moreover, \QKD is a quantum-resistant (quantum-safe) key establishment
protocol. The need for quantum-resistant cryptography is widely 
recognized given that the threat of quantum computers to current
cryptosystems can become a reality in the future. In contrast to
the area of post-quantum cryptography that is also 
believed to be quantum-resistant,
one can actually prove the information-theoretic security of \QKD protocols, an 
attractive and important feature. 
Moreover, in these security proofs, one need only assume 
that Eve follows the laws of quantum mechanics. Such an all-powerful 
Eve can be assumed to have access to unlimited computational powers,
including not-yet-available quantum computers. The core of a 
proof of security
of any \QKD protocol is to calculate the secret key rate, which is the
number of secret key bits obtained per exchange of a quantum signal.  As
one needs to take into account all possible eavesdropping attacks from
an all-powerful Eve, an analytical calculation of the key rate is
extremely challenging. Such analytical calculations are 
generally limited to protocols 
with high symmetry like the BB84 protocol \cite{Bennett1984}. Moreover, one often has to invoke inequalities to obtain an analytical lower bound of the true key rate. Those inequalities can significantly loose the key rate in many practical parameter regimes of a \QKD protocol.

Fortunately, the key rate problem  
can be formulated as the optimal value of a convex minimization problem.  In general, the size of the problem makes it intractable to find analytical solutions. Therefore, we resort to numerical approaches.
The numerical calculation of the key rate can be done by
finding a \emph{provable tight lower bound} of this convex 
optimization problem.
This is true
for both asymptotic~\cite{coles2016numerical,winick2017reliable} and
finite-size regimes~\cite{george2020finite}.
In principle, any (device-dependent) \QKD protocol
can be analyzed in this numerical framework, including measurement-device-independent, and both
discrete-variable and continuous-variable protocols.
 If we are able to solve the problem numerically without loosing tightness, we can potentially provide tighter key rates, even in situations where some valid analytical bounds are known. This is highly relevant for practical implementations of \QKD protocols as it may enable us to gain key rates without modifying hardwares. This paper continues with the convex optimization approach,
and contributes novel robust strategies for numerical calculation of
key rates that exhibit strong practical performance.

\subsection{Convex Optimization}\label{sec:intro_convex}
The secret key rate convex optimization calculations can be typically performed
after introducing suitable tools to reduce the dimension of the problem, 
e.g.,~the squashing
models~\cite{zhang2021security}, or the dimension reduction 
method~\cite{upadhyaya2021dimension}. These tools are necessary as we are interested in \QKD protocols that have quantum optical implementations, and thus work with infinite-dimensional Hilbert spaces corresponding to optical modes.  In reality, the success of
this security proof method is often limited by computational resources,
as well as the efficiency and accuracy of underlying algorithms.  For a specific protocol, it is also possible to  further reduce the dimension of the key rate calculation problem by using properties of the protocol. For example, as done in \cite{zhang2021security,Li2020}, one can without loss of generality consider only block-diagonal density matrices in the feasible set due to the block-diagonal structure of measurement operators. Then by utilizing properties of the objective function, we can split the key rate calculation problem into several convex optimization problems for individual blocks. Our goal here is to develop a tool that works for general protocols. In applying our method, one can always choose the problem formulation after applying all applicable aforementioned methods to reduce the dimension as the starting point. Furthermore, it should be noted that the finite-size
key rate problem involves a variation of the asymptotic key rate
formulation. For the simplicity of our discussion, we focus on the asymptotic formulation in this paper.

The work in~\cite{winick2017reliable} provides a reliable framework to compute
the key rate using a two-step routine. In the first step, one tries to 
efficiently find a near optimal, feasible point, of the optimization 
problem; see~\Cref{eq:keyrateopt} for the explicit problem formulation. 
In the second step, one then obtains a reliable lower bound from this
feasible point by a linearization and duality argument. In terms of numerical
computation, the bottleneck of this approach
for large-size \QKD problems comes from the
first step, as it involves semidefinite optimization with a nonlinear
objective function. In particular, the work in~\cite{winick2017reliable}
proposes an algorithm  based on the Frank-Wolfe method to solve the
first step. However, this method can converge slowly in practice.
\label{pg:cancnv}We note that in Faybusovich and
Zhou~\cite{Faybusovich2020longstep}, they
also work on providing a more efficient algorithm based on a
long-step path-following interior-point method for the \QKD key rate
calculation problem. However, their discussions are restricted to real
symmetric matrices, while for \QKD key rate calculations, it is important
to handle Hermitian matrices. Although it might be possible to extend
the algorithm in~\cite{Faybusovich2020longstep} to deal with Hermitian
matrices, currently the extension is not done and thus, we cannot
directly compare our algorithms with theirs. 
In addition, the problem formulations used
in~\cite{winick2017reliable,Faybusovich2020longstep}
do not guarantee positive definiteness of the matrices involved in
the objective function. Therefore, they perturb current feasible
solutions by adding a
small multiple of the identity matrix. This perturbation is not required in our
new method in this paper due to our regularization using facial reduction, \FRp.
\footnote{\label{pg:footFZ}
	Following the original submission of this paper, we have
	obtained access to the code in~\cite{Faybusovich2020longstep}. 
	We have observed that
	the code requires a Slater point, as it uses the analytic center as the
	starting point of the algorithm. Therefore, we could not use it for many
	of our instances where strict feasibility fails.
	Moreover, the random problems solved in
	the paper have positive definite (strictly feasible) optimal solutions
	that are relatively close to the analytic center.
	We hope to fully test the code in \cite{Faybusovich2020longstep} 
	on our facially reduced
	problems in a follow-up paper. Hopefully, we can combine the information
	on efficient evaluations of the Hessian in \cite{Faybusovich2020longstep}  
	with our \FR techniques and improve both algorithmic approaches.}

In this paper, we derive a stable reformulation of the convex semidefinite programming, \SDPp, model for the key rate calculation problem of \QKDp. 
We use this to derive efficient,
accurate, algorithms for the problem, in particular, for finding
provable lower bounds for the problem. The stable reformulation is based on
a novel facial reduction, \FRp, approach. We exploit the Kronecker structure
and do \FR first for the linear constraints to guarantee a
positive definite, strictly feasible solution. Second we exploit the
properties of the completely positive maps and do \FR
on the nonlinear, quantum relative entropy objective function,
to guarantee positive definiteness of its arguments.
This allows for a Gauss-Newton type interior-point approach that avoids
the need for perturbations to obtain positive definiteness, 
a technique currently used in the literature. The result is high
accuracy solutions with provable lower (and upper) bounds for the convex
minimization problem. We note that the convex minimization problem
is designed to provide a \emph{lower bound} for the key rate.

\subsection{Outline and Main Results}

\label{sec:mainres}
In~\Cref{sect:prel} we present the preliminary notations and convex
analysis tools that we need.
In particular, we include details about the linear maps
and adjoints  and basic \emph{facial reduction, \FRp}, needed for our 
algorithms; see~\Cref{sect:lintrsadj,sect:confacered}.

The details and need for facial reduction, \FRp, is discussed in~\Cref{sect:FR}. We perform \FR due to the loss of strict feasibility for the linear constraints for some classes of instances, as well as the rank deficiency of the images of the linear maps that consist of the nonlinear objective function. The \FR guarantees that the reformulated model satisfies the strict feasibility of the linear constraints, and the objective function is evaluated at positive definite density matrices.  A partial \FR that stems from singularity encountered from the \emph{reduced density operator constraint} is discussed in \Cref{sect:partialFR}. A second type of \FR performed on the completely positive mappings of the objective function is discussed in \Cref{sec:FRontheobj}. These \FR steps result in a much simplified reformulated model \cref{eq:finalredinrho} for which strict feasibility holds and the objective function arguments preserve positive definiteness. This allows for efficient accurate evaluation of the objective function that uses a spectral decomposition and avoids the less accurate matrix log function evaluation. In addition, we discuss the differentiability of the objective function, both first and second order, in \Cref{sect:derivsquant}.

In~\Cref{sect:intpt} we begin with the optimality conditions and
a projected Gauss-Newton, \GNp, interior point method. 
This uses the modified objective function
that is well defined for positive definite density matrices, $\rho\succ 0$.
We use the \emph{stable} \GN search direction for the primal-dual
interior-point, \pdipp, method. This avoids unstable backsolve steps for
the search direction. We also use a sparsity preserving
nullspace representation for the primal feasibility in ~\Cref{sect:sparsenullrep}.
This provides for exact primal feasibility steps during the algorithm.
Optimal diagonal precondition for the linear system is presented 
in~\Cref{sect:precond}.

\label{pg:lowerbnd}
Our upper and lower bounding techniques are given 
in~\Cref{sect:dualbnds}.
In particular, we provide novel theoretical based lower bounding
techniques for the \FR and original problem
in~\Cref{cor:lowerbound,cor:lowerboundpert}, respectively.\footnote{This
appears to be a novel contribution for general nonlinear convex \SDP
optimization.}

Applications to the security analysis of some selected \QKD protocols
are given in~\Cref{sect:tests}. This includes comparisons
with other codes as well as solutions of problems that could not be
solved previously. We include the lower bounds and
illustrate its strength by including the relative
gaps between lower and upper bounds;
and we compare with the analytical optimal values when it is possible to do so.

We provide concluding remarks in~\Cref{sect:concl}.
Technical proofs, further references and results, appear in
\Cref{app:pfstech,app:compasp}. The details for six protocol examples
used in our tests are given in~\Cref{sect:testexamples}.

A reader, whose main interest is to use the code released from this work to perform \QKD security analysis, can safely skip~\Cref{sect:innerprod,sect:lintrsadj,sect:confacered} after reading \Cref{sect:QKDbackground}. One may refer to \Cref{sect:notations} for notations if necessary. The main results of~\Cref{sect:FR,sect:intpt} are summarized at the beginning of those sections. The final model of the key rate problem after regularization via \FR is given in~\cref{eq:finalredinrho}. The projected Gauss-Newton interior-point algorithm for this model is presented in \Cref{algo:GNalgo} with a less technical explanation given in~\Cref{sect:projGNipalgor}. In terms of reliable lower bound for the original key rate problem, one may be interested in~\Cref{remark:lowerbound}  and \Cref{cor:lowerboundpert}. One can then safely skip details presented in the rest of~\Cref{sect:FR,sect:intpt} and proceed with~\Cref{sect:tests} to compare numerical performance of our method with other existing approaches.

\section{Preliminaries}
\label{sect:prel}
We now continue with the terminology and preliminary background for the
paper as well as presenting the notations. \Cref{sect:innerprod,sect:lintrsadj,sect:confacered} contain 
the convex optimization background material to understand details of our problem reformulation and our projected Gauss-Newton interior-point algorithm. We also include pointers to additional convex
optimization background that appears in~\Cref{app:pfstech}.

\subsection{QKD key rate calculation background}\label{sect:QKDbackground}
The asymptotic key rate $R^{\infty}$ is given by the Devetak-Winter
formula~\cite{Devetak2005} that can be written in the following form \cite{winick2017reliable}:
\begin{equation}\label{eq:dwkeyrate}
	R^{\infty} =  \min_\rho D(\cG(\rho)\|\cZ(\cG(\rho))) - p_{\text{pass}}\delta_{\text{EC}},
\end{equation}
where $D(\delta\|\sigma)=:$\textdef{$f(\delta,\sigma)= 
	\tr \left( \delta (\log \delta - \log \sigma)\right)$}
is the quantum relative entropy,
\textdef{$p_{\text{pass}}$} is the probability that a given signal is
used for the key generation rounds, and \textdef{$\delta_{\text{EC}}$} is the cost of error correction per round. The last two parameters are directly determined by observed data and thus are not a part of the optimization. Thus, the essential task of the quantum key distribution rate computation is to solve the following nonlinear convex semidefinite program:
\begin{equation}\label{eq:keyrateopt}
	\begin{array}{rcllll}
		&& \min_\rho &
		D(\cG(\rho)\|\cZ(\cG(\rho))) \\
		&&\text{s.t.} & \Gamma(\rho) = \gamma,
		\\  &&&  \rho \succeq 0,
	\end{array}
\end{equation}where
 \index{$m$, number of linear constraints}
$\Gamma: \H^{n} \rightarrow \R^{m}$ is a linear map 
defined by $\Gamma(\rho) = \left(\tr (\G_i \rho)\right)$;
$\Hn$ is the linear space of Hermitian matrices over the reals;
and $\gamma \in \Rm$. In this problem, $\{\Gamma_i\}$ is a set of
Hermitian matrices corresponding to physical observables. The
data pairs  $\Gamma_i, \gamma_i$ are known observation statistics that
include the $\trace (\rho) =1$ constraint.  The maps
$\cG$ and $\cZ$ are linear, completely positive maps that are specified
according to the description of a \QKD protocol. In general, $\cG$ is
\label{pgno:trace} 
trace-non-increasing, while $\cZ$ is trace-preserving and its Kraus
operators are a resolution of identity. The maps are usually represented via
the so-called operator-sum (or Kraus operator) representation.
(More details on these representation are given below as
needed; see also~\Cref{def:operGandZ}.)

In other words, the optimization problem is of the form in~\cref{eq:absQKD}:
\begin{equation}
\label{eq:absQKD}
 \qquad \min \{f(\rho) : \Gamma(\rho)  =\gamma, \ \rho\succeq 0 \},
\end{equation}
where the objective function $f$ is the quantum relative entropy 
function as shown in \cref{eq:dwkeyrate}, and the constraint set is a spectrahedron, i.e.,~the 
intersection of an affine manifold and the positive semidefinite cone. 
The affine manifold is defined using the linear map
for the linear equality constraints in~\cref{eq:absQKD}:
\[
\Gamma(\rho) = \left(\tr (\G_i \rho)\right), i=1,\ldots, m, \quad
\Gamma : \Hn \to \Rm.
\]
These are divided into two sets: the observational and reduced density operator
constraint sets, i.e.,~$S_O\cap S_R$.

\label{pg:states}
The set of states $\rho$ satisfying the \textdef{observational constraints} is given by
\index{$S_O$, observational constraints}
\index{constraint sets}
\index{constraint sets!observational, $S_O$}
\index{constraint sets!reduced density, $S_R$}
\index{$\rho$, state} 
\begin{equation}
	\label{eq:constrobser}
	S_O = \left\{\rho \succeq 0 \,:\, \langle P^A_s\otimes P^B_t,\rho\rangle =
	p_{st},\, \forall st\right\},
\end{equation}
where we let \textdef{$n_A,n_B$} be the sizes of
$P^A_s\in \H^{n_A},P^B_t\in \H^{n_B}$, respectively; and
we denote the \textdef{Kronecker product, $\otimes$}.
We set \textdef{$n=n_An_B$ which is the size of $\rho$}. 
\index{size of $\rho$, $n=n_An_B$}
\index{$\otimes$, Kronecker product}

The set of states $\rho$ satisfying the constraints with respect to the
\emph{reduced density operator, $\rho_A$}, is
\index{reduced density operator constraint, $S_R$}
\index{$S_R$, reduced density operator constraint} 
\begin{equation}	\label{eq:constrreduced}
\begin{array}{rrl}
S_R &=& \left\{\rho\succeq 0 \,:\, \tr_{B}(\rho) = \rho_{A} \right\} \\
	&=&\left\{\rho\succeq 0 \,:\, 		\langle \Theta_j\otimes \one_B,\rho\rangle = \theta_j,\, \forall j = 1,\ldots,m_R\right\},
\end{array}
\end{equation}
where  $\theta_j = \langle \Theta_j,\rho_A\rangle$ and $\{\Theta_j\}$ 
forms an orthonormal basis for the real vector space of Hermitian matrices 
on system A. This implicitly defines the linear map and
constraint in \textdef{$\tr_B(\rho) = \rho_A$}. Here we denote the identity matrix \textdef{$\one_B \in \H^{n_B}$}.

Here, we may assume that $\Gamma_1 = I$ and $\gamma_1 =1$ to guarantee that 
we restrict our variables to \textdef{density matrices}, i.e.,~semidefinite 
and unit trace. (See~\cite[Theorem 2.5]{MR1796805}.)

\subsection{Notations}\label{sect:notations}

\index{$\Sn$, set of real symmetric $n$-by-$n$ matrices}
\index{set of real symmetric $n$-by-$n$ matrices, $\Sn$}
\index{$\real(X)$, real part of $X$}
\index{real part of $X$, $\real(X)$}
\index{$\imag(X)$, imaginary part of $X$}
\index{imaginary part of $X$, $\imag(X)$}
\index{$\Snp$, positive semidefinite cone of $n$-by-$n$  real symmetric matrices}
\index{positive semidefinite cone of $n$-by-$n$  real symmetric matrices, $\Snp$}
\index{$\Snpp$, positive definite cone of $n$-by-$n$  real symmetric matrices}
\index{positive definite cone of $n$-by-$n$  real symmetric matrices, $\Snpp$}
\index{$\Hnp$, positive semidefinite cone of $n$-by-$n$  Hermitian matrices}
\index{positive semidefinite cone of $n$-by-$n$  Hermitian matrices, $\Hnp$}
\index{$\Hnpp$, positive definite cone of $n$-by-$n$ Hermitian matrices}
\index{positive definite cone of $n$-by-$n$ Hermitian matrices, $\Hnpp$}
\index{$X\succeq 0$}
\index{$X\succ 0$}

We use $\Cnn$ to denote the space of $n$-by-$n$ complex matrices, and
$\Hn$ to denote the \emph{subset} of $n$-by-$n$ Hermitian matrices; we use
$\H$ when the dimension is clear. We use $\Sn,\Sc$ for the subspaces of $\Hn$
of real symmetric matrices. Given a matrix $X\in \Cnn$, we use $\real(X)$ and $\imag(X)$ to denote the real and the imaginary parts of $X$, respectively. 
We use $\Hnp,\Snp$ ($\Hnpp, \Snpp$, resp) to denote the positive
semidefinite cone (the positive definite cone, resp); and again
we leave out the dimension when it is clear.
We use the partial order notations $X\succeq 0,
X\succ 0$ for semidefinite and definite, respectively.
We let $\Rn$ denote the usual vector space of real $n$-coordinates; 
$\mathcal{P}_C(X)$ denotes the projection of $X$ onto the 
closed convex set $C$.
For a matrix $X$, we use $\range(X)$ and $\nul(X)$ to denote the
\textdef{range} and the \textdef{nullspace} of $X$, respectively. 
We let \textdef{$\BlkDiag$}$(A_1,A_2,\ldots,A_k)$ denote the block  diagonal matrix with diagonal blocks $A_i$.
\index{$\Rn$, vector space of real $n$-coordinates} 
\index{vector space of real $n$-coordinates, $\Rn$} 
\index{$\mathcal{P}_C(X)$, projection of $X$ onto $C$}
\index{$\range(X)$, range of $X$}
\index{$\nul(X)$, nullspace of $X$}
\index{$\BlkDiag(A_1,A_2)$, block diagonal matrix with diagonal blocks $A_1,A_2$}

\subsection{Real Inner Product Space $\Cnn$ }
\label{sect:innerprod}

In general, $\Hn$ is not a subspace of $\Cnn$ unless we treat both as
vector spaces over $\R$. To do this we define a 
\textdef{real inner product in $\Cnn$} that takes the standard inner
products of the real and imaginary parts:
\begin{equation}
\label{def:complex_innerprod}
\begin{array}{rcl}
\langle Y,X\rangle 
&=&
\langle\Re(Y), \Re(X)\rangle+\langle\imag(Y), \imag(X)\rangle 
\\&=&
\Tr \left(\Re(Y)^\dagger \Re(X)\right)+\Tr \left(\imag(Y)^\dagger \imag(X)\right)
\\&=&
\Re\left(\Tr (Y^\dagger  X)\right).
\end{array}
\end{equation}
We note that
\[
\Re(\<Y,X\> ) = \<\Re(Y), \Re(X) \> + \< \imag(Y),   \imag(X) \>,\quad
\imag(\<Y,X\> ) = -  \<\Re(Y),  \imag(X) \> +  \< \imag(Y), \Re(X)  \> .
\]
Over the reals, $\dim (\Hn) = n^2, \dim (\Cnn) = 2n^2$. The induced 
norm is the Frobenius norm $\|X\|^2_F = \<X,X\>=\trace \left(X^\dagger X\right)$, where
we denote the \textdef{conjugate transpose, $\cdot^\dagger $}.
\index{$\cdot^\dagger $, conjugate transpose}
\index{$\|X\|_F$, Frobenius norm}
\index{Frobenius norm, $\|X\|_F$}

\subsection{Linear Transformations and Adjoints}
\label{sect:lintrsadj}

Given a linear map $\cL: \mathcal{D} \to \mathcal{R}$, we call the
unique linear map $\cL^\dagger : \mathcal{R} \to \mathcal{D}$  the \textdef{adjoint} of $\cL$, if it satisfies
\[
\left\langle \cL (X),Y\right\rangle = \left\langle X,\cL^\dagger  (Y)\right\rangle, \, \forall X \in \mathcal{D},Y\in \mathcal{R}.
\]
Often in our study, we use vectorized computations instead of using complex matrices directly. In order to relieve the computational burden, we use isomorphic and isometric realizations of matrices by ignoring the redundant entries.
We consider $\H^n$ as a vector space of dimension $n^2$ over the reals. We
define $\Hvec(H)\in \Rnt$ by stacking $\diag(H)$ followed by
$\sqrt 2$ times the strict upper triangular
parts of $\real(H)$ and $\imag(H)$, both columnwise:
\[
\Hvec(H) = \begin{pmatrix}
\diag(H)\cr
\sqrt 2 \real(upper(H))\cr
\sqrt 2  \imag(upper(H))
\end{pmatrix}\in \Rnt, \quad \HMat = \Hvec^{-1} = \Hvec^\dagger  .
\]
We note that for the real symmetric matrices $\Sn$, we can use the first
\textdef{triangular number, $t(n)=n(n+1)/2$} of elements in $\Hvec$,
and we denote this
by \textdef{$\svec$}$(S)\in \R^{t(n)}$, with adjoint \textdef{$\sMat$}.
\index{$t(n)=n(n+1)/2$, triangular number}

\index{$\cL^\dagger $, adjoint of $\cL$}
\index{adjoint of $\cL$, $\cL^\dagger $}

We use various linear maps in an \SDP framework. For given $\Gamma_i\in \Hn, i=1,\ldots,m$, define 
\[
\Gamma : \Hn \to \Rm \text{ by }
\Gamma (H) = (\langle \Gamma_i,H\rangle)_i \in \Rm.
\]
The adjoint satisfies
\[
\langle \Gamma (H),y\rangle = \sum_i y_i \Tr (\Gamma_iH) = 
\Tr \left( H \left(\sum_i y_i \Gamma_i\right) \right)= 
\left\langle  H, \Gamma^\dagger  (y) \right\rangle.
\]
The matrix representation $A$ of $\Gamma$ is found from
\[
(A \Hvec(H))_i = (\Gamma (H))_i = \<\Gamma_i,H\> = \left\langle\Hvec (\Gamma_i),\Hvec (H)\right\rangle,
\]
i.e., for $g_i=\Hvec (\Gamma_i),\ \forall i$ and $h=\Hvec (H)$, 
\[
\Gamma(H) \equiv A(h), \ \text{ where } A=\begin{bmatrix} g_1^\dagger\cr \vdots\cr g_m^\dagger \end{bmatrix}.
\]
Specialized adjoints for matrix multiplication are given 
in~\Cref{sect:adjMM}.

\subsection{Cones, Faces, and Facial Reduction, \FRp}
\label{sect:confacered}

The facial structure of the semidefinite cone is well understood. We
outline some of the concepts we need for facial reduction
and exposing vectors (see e.g.,~\cite{DrusWolk:16}).
We recall that a \textdef{convex cone} $K$ is defined by: $\lambda K\subseteq K,
\forall \lambda \geq 0, \, K+K\subset K$, i.e.,~it is a cone and so
contains all rays, and it is a convex set.  \index{dual cone, $S^\dagger $}
For a set $S\subseteq \H$ we denote the \textdef{dual cone},
$S^\dagger  = \{\phi \in \H : \langle \phi,s\rangle \geq 0, \ \forall s \in
S\}$.
\index{$S^\dagger $, dual cone}
\index{$\unlhd$}
\begin{definition}[{\textdef{face}}]
A convex cone $F$ is a face of a convex cone $K$, denoted $F\unlhd K$,
if 
\[
x,y\in K, x+y\in F \implies x,y \in F.
\]
Equivalently, for a general convex set $K$ and convex subset $F\subseteq
K$, we have $F\unlhd K$, if 
\[
[x,y] \subset K, z\in \relint [x,y], z\in F \implies [x,y]\subset F,
\]
where $[x,y]$ denotes the line segment joining $x,y$.
\end{definition}
\index{line segment, $[x,y]$}
\index{$[x,y]$, line segment}

Faces of the positive semidefinite cone are characterized by the range
or nullspace of any element in the relative interior of the faces. 
In fact, the following characterizations in \Cref{lem:propfaces}
hold.

\index{$\unlhd$}
\begin{lemma}
\label{lem:propfaces}
Let $F$ be a convex subset of $\Hp^n$ with $X\in \relint F$. Let
\[
X = \begin{bmatrix} P & Q\end{bmatrix}
 \begin{bmatrix} D & 0 \cr 0 & 0\end{bmatrix}
 \begin{bmatrix} P & Q\end{bmatrix}^\dagger 
\]
be the orthogonal spectral decomposition with $D\in \Hpp^r$.
Then the following are equivalent:
\begin{enumerate}
\item
$F\unlhd \Hp^n$;
\item
\label{item:facesranges}
$F= \{ Y \in \Hp^n \,:\, \range(Y) \subset \range(X) \}
= \{ Y \in \Hp^n \,:\, \nul(Y) \supset \nul(X) \}$;
\item
\label{item:facialvec}
$F= P \Hp^r P^\dagger $;
\item
\label{item:expvect}
$F= \Hp^n \cap (Q Q^\dagger )^\perp$.
\end{enumerate}
\end{lemma}

\label{pg:compactspectral}
The matrix $P$, in \Cref{item:facialvec} of \Cref{lem:propfaces},  allows us to represent any matrix $Y\in F\unlhd \Hnp$ 
as a \textdef{compact spectral decomposition}, i.e., $Y = P D P^\dagger$ 
with diagonal $D \in \H_{+}^r$.
This compact representation leads to a reduction in
the variable dimension. The matrix $QQ^\dagger $, in \Cref{item:expvect} of \Cref{lem:propfaces}, is called an
\textdef{exposing vector} for the face $F$.
Exposing vectors come into play throughout~\Cref{sect:FR}.

\begin{definition}[minimal face]
Let $K$ be a closed convex cone and let $X \subseteq K$. Then $\face(X)\unlhd
K$ is the \emph{minimal face}, the intersection of all faces of $K$ that
contain $X$.
\end{definition}
\index{$\face(X)$, minimal face}
\index{minimal face, $\face(X)$}

Facial reduction is a process of identifying the minimal face of $\Hp^n$
containing the spectrahedron $\{\rho:\Gamma(\rho) = \gamma\}\cap \Hnp$.
\Cref{lem:FRfarkas} plays an important role in the heart of the facial
reduction process. Essentially, either there exists a strictly
feasible $\rho$,  or the alternative holds that there exists
a linear combination of the $\Gamma_i$ that is positive semidefinite but
has a zero expectation.

\begin{lemma}[theorem of the alternative, {\cite[Theorem 3.1.3]{DrusWolk:16}}]  
\label{lem:FRfarkas}
For the feasible constraint system in \eqref{eq:absQKD}, exactly one of the following statements holds:
\begin{enumerate}
\item there exists $\rho\succ 0$ such that $\Gamma(\rho) = \gamma$;
\item  there exists $y$ such that 
\begin{equation}
\label{eq:little_auxsystem}
0 \ne \Gamma^\dagger (y) \succeq 0 \ , \ \<\gamma,y\> = 0.
\end{equation}
\end{enumerate}

\end{lemma}

In \Cref{lem:FRfarkas}, the matrix $\Gamma^\dagger (y)$ is an exposing vector 
for the face containing the constraint set in \eqref{eq:absQKD}.

\section{Problem Formulations and Facial Reduction}
\label{sect:FR}

\index{$\Hn$, set of $n$-by-$n$ Hermitian matrices}
\index{set of $n$-by-$n$ Hermitian matrices, $\Hn$}

The original problem formulation is given in ~\cref{eq:keyrateopt}.
Without loss of generality we can assume that the
feasible set, a \textdef{spectrahedron}, is nonempty. This is because
our problem is related to a physical scenario, and 
we can trivially set the
key rate to be zero when the feasible set is empty. 
Note that the Hermitian (positive semidefinite, density) matrix
$\rho$ is the only variable in the above~\cref{eq:keyrateopt}
optimization problem. Motivated by the fact that the mappings $\cG,
\cZ\circ \cG$ are positive semidefinite preserving but possibly
not positive definite preserving, we rewrite
\cref{eq:keyrateopt} as follows:\footnote{This allows us to regularize 
	below using facial reduction, \FRp.}
\begin{equation}
	\label{eq:keyrateoptequivsigdel} \quad
	\begin{array}{rcllll}
		&& \min_{\rho,\sigma,\delta} &
		\trace(\delta(\log \delta -\log \sigma)) \\
		&&\text{s.t.} & \Gamma(\rho) =  \gamma
		\\ &&&  \sigma = \cZ(\delta) 
		\\ &&&  \delta = \cG(\rho) 
		\\  &&& \rho, \sigma, \delta \succeq 0.
	\end{array}
\end{equation}
 Due to the structure of the linear mapping $\cG$, the matrix $\delta$
is often singular in~\cref{eq:keyrateoptequivsigdel}. Therefore,
strict feasibility fails in~\cref{eq:keyrateoptequivsigdel}.
This indicates that the objective function, the
\textdef{quantum relative entropy function} is evaluated on
singular matrices in both~\cref{eq:keyrateopt,eq:keyrateoptequivsigdel},
creating theoretical
and numerical difficulties. In fact, the domain of the problem
that guarantees finiteness for the objective function, requires
restrictions on the ranges of the linear mappings.
By moving back and forth between equivalent formulations of the types
in~\cref{eq:keyrateopt,eq:keyrateoptequivsigdel}, we derive a
regularized model that simplifies type~\cref{eq:keyrateopt}, and where positive
definiteness is preserved. In particular, the regularization allows for an
efficient interior point method even though the objective
function is not differentiable on the boundary of the
semidefinite cone. This allows for efficient algorithmic developments.
In addition, this enables us to accurately solve previously
intractable problems.

We now present the details on various formulations of \QKD
from~\cref{eq:keyrateopt,eq:keyrateoptequivsigdel}. We show
that facial reduction allows for regularization of both the constraints 
and the objective function, i.e.,~this means that we have positive
\emph{definite} feasible points, and a proper domain for the objective
function with positive definite matrices. This obviates the need for
adding perturbations of the identity.
We include results about
\FR for positive transformations and show that highly accurate \FR can
be done in these cases.

 In particular, we provide a regularized
reformulation of our problem, see~\cref{eq:finalredinrho}, and the
regularization statement that guarantees positive definiteness, 
see~\cref{eq:rho_welldef}.

\subsection{Properties of Objective Function and Mappings $\cG,\cZ$}
\label{sect:objfn}

The \textdef{quantum relative entropy function} 
$D : \Hnp \times \Hnp \to \Rp\cup \{+\infty\}$ is denoted by
$ D(\delta || \sigma)$, and is defined as
\begin{equation}
\label{eq:objective}
D(\delta || \sigma ) = 
\left\{
\begin{array}{ll}
\trace(\delta \log \delta)-\trace (\delta \log \sigma) & \text{if } \range(\delta)\cap \nul(\sigma) = \emptyset \\
\infty & \text{otherwise.}
\\
\end{array}
\right.  
\end{equation}
That the quantum relative entropy $D$ is finite if $\range(\delta)\subseteq \range(\sigma)$ is shown
by extending the matrix log function to be $0$ on the nullspaces of
$\delta, \sigma$. (See~\cite[Definition 5.18]{watrous_2018}.)
It is known that $D$ is nonnegative, equal to $0$ if, and only if,
$\delta =\sigma$, and is jointly convex in both $\delta$ and
$\sigma$, see ~\cite[Section 11.3]{MR1796805}.
\label{pg:delta}

\begin{definition}
\label{def:operGandZ}
The linear map \textdef{$\GG: \Hn \to \Hk$} is defined as a
sum of matrix products (\textdef{Kraus representation})
\begin{equation}
\label{eq:Gmap}
\displaystyle\cG(\rho) := \sum_{j=1}^{\ell} K_j \rho K_j^\dagger ,
\end{equation}
where $K_j \in \mathbb{C}^{k \times n}$ and $\sum_{j=1}^{\ell} K_j^\dagger  K_j \preceq I$. The adjoint is $\cG^\dagger (\delta) := \sum_{j=1}^{\ell} K_j^\dagger  \delta K_j$.
\end{definition}
Typically we have $k>n$ with $k$ being a multiple of $n$; and thus we can have $\cG(\rho)$ rank deficient for
all $\rho\succ 0$. 

\begin{definition}
	\label{def:operGandG}
	The self-adjoint (projection) linear map \textdef{$\mathcal{Z}:
\Hk \to \Hk$} is defined as the sum
	\index{$\cZ : \Hk \to \Hk$}
	\begin{equation}
	\label{eq:Zmap}
	\displaystyle\cZ(\delta) := \sum_{j=1}^{N} Z_j \delta Z_j,
	\end{equation}
	where $Z_j = Z_j^2 = Z_j^\dagger  \in \Hkp$ and $\sum_{j=1}^N Z_j = I_k$.
\end{definition}

Since $\sum_{j=1}^N Z_j = I_k$, the set $\{Z_i\}_{j=1}^N$ is a 
\textdef{spectral resolution of $I$}. \Cref{prop:Zopprops} below states
some interesting properties of the operator $\cZ$; see
also~\cite[Appendix C, (C1)]{Coles2012}.

\begin{prop}
\label{prop:Zopprops}
The linear map $\cZ$  in \Cref{def:operGandG}
is an orthogonal projection on $\Hk$. Moreover, for $\delta \succeq 0$,
\begin{equation}
\label{eq:deltaZdeltalog}
\trace \left(\delta \log \cZ (\delta)\right) = \trace \left(\cZ (\delta) \log \cZ (\delta)\right) .
\end{equation}

\end{prop}
\begin{proof}
First we show that the matrices of $\cZ$ satisfy 
\begin{equation}
\label{eq:ZiZjzero}
Z_iZ_j=0,~\forall~i\neq j.
\end{equation}
For $i,j \in \{1,\ldots, N\}$, we have by \Cref{def:operGandG} that
	\begin{equation}
		\label{eq:orthogZj}
		\begin{array}{rcl}
			Z_i \left( \sum_{s=1}^N Z_s \right) Z_i = Z_i I_kZ_i = Z_i
			& \implies &
			0=\sum_{s\ne i} Z_i Z_s Z_i = \sum_{s\ne i} (Z_sZ_i)^\dagger (Z_sZ_i)
			\\& \implies &
			Z_jZ_i=0,\ \forall j\neq i.
		\end{array}
	\end{equation}
We now have $\cZ=\cZ^2=\cZ^{1/2}=\cZ^\dagger $.  Thus, $\cZ$ is an orthogonal projection. 
The equality \cref{eq:deltaZdeltalog} holds by the properties of the map $\cZ$ that it removes the off-diagonal blocks in its image.
 \end{proof}

Using \eqref{eq:objective}, \Cref{lemma:rangeZrelation} below shows that 
the objective value of the model \eqref{eq:keyrateopt} is finite on the
feasible set. This also provides insight on the usefulness of
\FR on the variable $\sigma$ done below.
\begin{lemma}
\label{lemma:rangeZrelation}
Let $X\succeq 0$. Then $\range(X) \subseteq \range(\cZ(X))$.
\end{lemma}
\begin{proof}
See \Cref{sec:proof_rangelem}. 
\end{proof}

\begin{remark}
\label{rem:notposdef}
In general, the mapping $\cG$ in~\cref{eq:Gmap}
\underline{does not preserve positive definiteness}. Therefore the 
objective function $f(\rho)$, see~\cref{eq:f_in_trace} below, may  need
to evaluate $\trace (\delta \log \delta)$ and $\trace (\delta \log \sigma)$
with \emph{both} $\delta = \cG(\rho)$ and $\sigma = \cZ\circ \cG(\rho)$ 
always singular.
Although the objective function $f$ is well-defined at singular points 
$\delta, \sigma$, the gradient of $f$ at singular points $\delta, 
\sigma$ is not well-defined. Our approach using \FR within an 
interior point method avoids these numerical 
difficulties.
\end{remark}

\subsection{Reformulation via Facial Reduction (\FRp)}

Using \Cref{prop:Zopprops}, we can now reformulate the objective function 
in the key rate optimization problem \cref{eq:keyrateoptequivsigdel} to 
obtain the following equivalent model:
\begin{equation}
\label{eq:optprob}
\begin{array}{rcllll}
&& \min_{\rho,\sigma,\delta} &
\trace(\delta\log \delta) - \trace( \sigma \log \sigma) \\
&&\text{s.t.} & \Gamma(\rho) = \gamma
\\ &&&  \sigma - \cZ(\delta)  = 0
\\ &&&  \delta - \cG(\rho) = 0
\\  &&&   \textdef{$\rho\in \Hnp$}, \,\textdef{$\sigma \in \Hkp$},\,\textdef{$\delta \in \Hkp$}.
\end{array}
\end{equation}
The new objective function is the key in our analysis, as it simplifies the expressions for gradient and Hessian. Next, we derive facial reduction based on the constraints in \cref{eq:optprob}.

\subsubsection{Partial \FR on the Reduced Density Operator Constraint}
\label{sect:partialFR}

\index{facially reduced reduced density operator constraint}

\index{$S_R$, reduced density operator constraint} 
\index{reduced density operator constraint, $S_R$} 

Consider the spectrahedron $S_R$ defined by the reduced density operator 
constraint in \cref{eq:constrreduced}.
We now simplify the problem via \FR by using only~\cref{eq:constrreduced}
in the case that $\rho_A \in \H^{n_{A}}$ is singular.
We now see in~\Cref{thm:FRonObserv} that we
can do this explicitly using the spectral decomposition of $\rho_A$;
see also~\cite[Sec. II]{Ferenczi2012}.
Therefore, this step is extremely accurate. 
Using the structure arising from the reduced density operator constraint, 
we obtain partial \FR on the constraint set in~\Cref{thm:FRonObserv}.
\begin{theorem}
\label{thm:FRonObserv}
Let $\range (P) = \range (\rho_A)\subsetneq \H^{n_A}, P^\dagger P=\one_r$ for $r < n_A$, and let $V=
P\otimes \one_B$. Then the spectrahedron $S_R$ in~\cref{eq:constrreduced} has the property that
\begin{equation}\label{eq:SO_face}
\rho \in S_R \implies \rho = VRV^\dagger , \text{  for some  }
R\in \H^{r\cdot n_B}_{+}.
\end{equation}
\end{theorem}
\begin{proof}
Let $\begin{bmatrix} P&Q \end{bmatrix}$ be a unitary matrix such that 
$\range (P) = \range (\rho_{A})$ and $\range (Q) = \nul (\rho_{A})$. Let $W = QQ^{\dagger} \succeq 0$. 
Recall that the adjoint $\Tr_B^\dagger(W) = W\otimes \one_B$.
Then $\rho \in S_R$ implies that
\begin{equation} \label{eq:TRBW}
\langle W \otimes \one_B, \rho \rangle = \langle W, \tr_{B}(\rho) \rangle = \langle W, \rho_{A} \rangle = 0,
\end{equation}
where \textdef{$\one_B \in \H^{n_B}$} is the identity matrix of size
$n_{B}$, and we use~\cref{eq:constrreduced} to guarantee that
$\Tr_B(\rho)= \rho_A$.
Therefore, $W \otimes \one_B \succeq 0$ is an exposing vector
for the spectrahedron $S_R$ in \cref{eq:constrreduced}. And we can write 
$\rho = VRV^\dagger $ with $V=P\otimes \one_B$ for any $\rho \in S_R$.
This yields an equivalent representation \cref{eq:SO_face} with a
smaller positive semidefinite constraint.\footnote{We provide a
self-contained alternate proof in~\Cref{sec:proofFRonObserv}.}
\end{proof}

We emphasize that facial reduction is not only powerful in reducing 
the variable dimension, but also in reducing the number of constraints.
Indeed, if $\rho_{A}$ is not full-rank, then at least one of the
constraints in \cref{eq:constrreduced} becomes redundant and can be
discarded; see~\cite{bw3,SWW:17}. 
In this case, it is equivalent to the matrix $\rho_A$ becoming smaller
in dimension.
(Our empirical observations show that many of the other observational
constraints $\Gamma_i(\rho)=\gamma_i$ 
also become redundant and can be discarded.)

\subsubsection{\FR on the Constraints Originating from $\cG,\cZ$}
\label{sec:FRontheobj}

Our motivation is that the domain of the objective function may be
restricted to the boundary of the semidefinite cone, i.e.,~the
matrices $\cG(\rho), \cZ(\cG(\rho))$ are singular by the definition
of $\cG$.  We would like to guarantee 
that we have a well-formulated problem with
strictly feasible points in the domain of the objective function so that the derivatives are well-defined. This guarantees basic numerical stability.
This is done by considering the constraints in the equivalent
formulation in~\eqref{eq:keyrateoptequivsigdel}.

\index{compact spectral decomposition}
We first note the useful equivalent form for the entropy function. 
\begin{lemma}
\label{lem:VdeltaV}
Let $Y = VRV^\dagger  \in \Hp,\, R\succ 0$ be the compact spectral decomposition of $Y$ with $V^\dagger V=I$.
Then
\[
\trace  (Y \log Y)
= \trace  (R\log R) .
\]
\end{lemma}
\begin{proof}
We obtain a unitary matrix $U = \begin{bmatrix}
V & P
\end{bmatrix}$ by completing the basis.
Then  $Y = UDU^\dagger $, where $D = \BlkDiag(R,0)$. We conclude, with $0\cdot \log0  =0 $, that $\trace Y \log Y = \trace  D \log D = \trace R \log R$.
\end{proof}

We use the following simple result to obtain the exposing vectors of the minimal face in the problem analytically.

\begin{lemma}
\label{lem:tranform}
Let $\CC \subseteq \Hn_+$ be a given convex set
with nonempty interior. Let $Q_{i} \in \mathbb{H}^{k \times n},
i=1,\dots,t$, be given matrices. Define the linear map 
	$\A : \Hn \rightarrow \Hk$ and matrix $V\in \C^{k\times r}$ by 
$$
\A (X) = \sum_{i=1}^t Q_{i}XQ_{i}^\dagger , \quad
	\range(V) = \range\left( \sum_{i=1}^tQ_{i}Q_{i}^\dagger  \right).
$$ 
Then the minimal face,  
\[
\face(\cA(\CC)) = V \mathbb{H}_{+}^{r} V^\dagger.
\]
\end{lemma}

\begin{proof}
First, note that properties of the mapping
implies that $\cA(\CC)\subset \Hkp$.
Nontrivial exposing vectors $0\neq W\in \H^n_+$ of $\cA(\CC)$ can be
characterized by the null space of the adjoint operator $\A^\dagger $:
$$
\begin{array}{rcl}
0\neq W\in \H^n_+, \langle W, \cA(\CC)\rangle = 0
&\iff &
0\neq W \succeq 0, \,
\langle W, Y \rangle = 0 ,  \, \forall \ Y \in \cA(\CC) \quad 
\\&\iff &
0\neq W \succeq 0, \,
\langle \A^\dagger  (W), X \rangle = 0 ,  \, \forall \ X \in \CC \quad 
\\&\iff &
0\neq W \succeq 0, \,
  W \in \nul(\A^\dagger )
\\&\iff &
0\neq W \succeq 0, \,
 Q_i^\dagger WQ_i = 0, \forall i,
\\&\iff &
0\neq \range(W) \subseteq \nul \left( \sum_i Q_iQ_i^\dagger  \right),
\end{array}
$$
where the third equivalence follows from $\Int (\CC) \neq \emptyset$, and the
fourth equivalence follows from the properties of the sum
of mappings of a semidefinite matrix.

 The choice of $V$ follows from choosing a maximal rank exposing
vector and constructing $V$ using \Cref{lem:propfaces}:
\[
\range(V) = \nul (W) =  \range \left( \sum_i Q_iQ_i^\dagger  \right).
\]
\end{proof}

We emphasize that the minimal face in \Cref{lem:tranform} means that 
$V$ has a minimum number of columns, as without loss of 
generality, we choose it to be full column
rank. In other words, this is the greatest reduction in the dimension of the image. 
In addition, the exposing vectors of $\cA(\CC)$ are characterized by the positive semidefinite matrices in the null space of $\A^\dagger $. 
This implies that \FR can be done in one step.

\index{exposing vector} 
\index{$V_{\rho}$}
\index{$V_{\delta}$}
\index{$V_{\sigma}$}

We describe how to apply \Cref{lem:tranform} to obtain $V_{\rho},V_{\delta},V_{\sigma}$ of the minimal face of $(\Hnp,\Hkp,\Hkp)$ containing the feasible region of \eqref{eq:optprob}. By \Cref{lem:propfaces}, we may write 
$$\begin{array}{rrll}
	\rho &=& V_{\rho}R_{\rho}V_{\rho}^\dagger  \in \H_+^n , & R_\rho \in \H_+^{n_\rho} \\
	\delta &=& V_{\delta}R_{\delta}V_{\delta}^\dagger  \in \H_+^k , & R_\delta \in \H_+^{k_\delta} \\
	\sigma &=& V_{\sigma}R_{\sigma}V_{\sigma}^\dagger  \in \H_+^k,  & R_\sigma \in \H_+^{k_\sigma}. \\
\end{array}$$

Define the linear maps
$$\begin{array}{rrrrrl}
	\Gamma_V:& \mathbb{H}_{+}^{n_{\rho}} \rightarrow \R^{m}&\text{ by }& \Gamma_V(R_{\rho}) &=& \Gamma( V_{\rho}R_{\rho}V_{\rho}^\dagger ), \\
	\cG_V:& \mathbb{H}_{+}^{n_{\rho}} \rightarrow \Hkp &\text{ by }& \cG_V (R_{\rho}) &=& \mathcal{G}( V_{\rho}R_{\rho}V_{\rho}^\dagger ), \\
	\cZ_V:& \mathbb{H}_{+}^{k_{\delta}} \rightarrow \Hkp &\text{ by }& \cZ_V(R_{\delta}) &=& \mathcal{Z}( V_{\delta}R_{\delta}V_{\delta}^\dagger ).
\end{array}$$
The matrices $V_{\rho},V_{\delta},V_{\sigma}$ are obtained as follows.
\begin{enumerate}
	\item 
\label{item:minfaceVrho}
We apply \FR to $\cF_\rho :=  \{ \rho \in \Hnp : \Gamma(\rho) = \gamma\}$ to find $V_{\rho}$ for the minimal face, $\face(\cF_\rho) \unlhd \Hnp $.
	\item Define \index{$\cR_{\rho}$}
	\[
	\mathcal{R}_{\rho}:= \{ R_\rho \in \H_+^{n_\rho} : \Gamma_V(R_\rho) = \gamma\}.
	\] 
	
	Note that $\text{int}( \mathcal{R}_{\rho}) \neq \emptyset$.
Applying \Cref{lem:tranform} to $\cF_\delta :=   \{ \cG_V (R_{\rho}) \in \Hkp : R_{\rho} \in \mathcal{R}_{\rho}  \}$, the matrix $V_{\delta}$ yields the minimal face, $\face(\cF_\delta) \unlhd \Hkp$ if we choose
\begin{equation}
	\label{eq:analytic_Vd}
	\range ( \Vd ) = \range \left( \cG_V(I) \right) .
\end{equation}

\item 	 \index{$\cR_{\delta}$}
Define $$\mathcal{R}_{\delta}:=\{ R_{\delta} \in
\mathbb{H}_+^{k_{\delta}} : V_{\delta}R_{\delta}V_{\delta}^\dagger  =
\cG_V (R_{\rho}), \ R_{\rho} \in \mathcal{R}_{\rho} \}.$$
We again note that $\text{int}( \mathcal{R}_{\delta}) \neq \emptyset$. Applying
\Cref{lem:tranform} to $\cF_\sigma :=  \{ \cZ_V (R_{\delta}) \in \Hkp :
R_{\delta} \in \mathcal{R}_{\delta} \}$, we find the matrix
$V_{\sigma}$ representing the minimal face, $\face(\cF_\sigma) \unlhd \Hkp$.
Thus, we choose $\Vs$ satisfying
\begin{equation}
	\label{eq:analytic_Vs}
	\range( \Vs ) = \range \left( \cZ_V(I) \right).
\end{equation}

\end{enumerate}
As above, this also can be seen by looking at the image of $I$ and the relative interior of the range of $\cZ_V$.
We note, by \Cref{lemma:rangeZrelation}, that $\range(\Vs) \supseteq \range(\Vd)$. Note that we have assumed the exposing vector of maximal rank for the original constraint set on $\rho$ in the first step is obtained. Without loss of generality, we can assume that the columns in  $V_{\rho},V_{\delta},V_{\sigma}$ are orthonormal. This makes the subsequent computation easier.

\begin{assump}
Without loss of generality, we assume $V_M^\dagger V_M=I$ for $M=\rho,\delta,\sigma$.
\end{assump}
 Define $\cV_\delta (R_\delta) := V_\delta R_\delta V_\delta^\dagger  $ and $\cV_\sigma (R_\sigma) := V_\sigma R_\sigma V_\sigma^\dagger $. Applying \Cref{lem:VdeltaV} and substituting for $\rho,\delta,\sigma$ to \eqref{eq:optprob}, we obtain the equivalent formulation \cref{eq:optprobobjconstrthree}.
\begin{equation}
\label{eq:optprobobjconstrthree}
	\begin{array}{rclll}
	&\min  &   \Tr (R_\delta \log (R_\delta) ) - \Tr\big( 
                        R_\sigma \log(R_\sigma)\big) \\
		&\text{s.t.} &   \GV ( R_\rho   ) = \gamma 
\\ && \cV_\sigma (R_\sigma)  - \cZ_V(R_\delta) = 0
\\ &&  
 \cV_\delta (R_\delta) - \cG_V(R_\rho) = 0
		   \\&&   R_\rho,  R_\sigma,  R_\delta  \succeq 0.
\end{array}
\end{equation}
After facial reduction, many of the linear equality constraints in \cref{eq:optprobobjconstrthree} end up being
redundant. We may delete redundant constraints and keep a well-conditioned equality constraints. 
In the next section, we show that the removal of the redundant constraints can be performed by \emph{rotating} the constraints.

\subsubsection{Reduction on the Constraints}

Recall that our primal problem after \FR is given in
\cref{eq:optprobobjconstrthree}.
Moreover, by the work above we can assume that  $\GV$ is surjective. In \Cref{thm:rotate_const} and \Cref{thm:rotate_const2} below, we show that we can simplify the last two equality constraints in \cref{eq:optprobobjconstrthree} by an appropriate rotation.

\begin{theorem}	\label{thm:rotate_const}
	Let $R_\rho \in  \H_+^{n_\rho}$ and $R_\delta \in   \H_+^{k_\delta}$. It holds that
	\begin{equation}\label{eq:rotate_const}
		\cV_\delta (R_\delta) = \cG_V(R_\rho)  \quad \Longleftrightarrow \quad R_\delta = \cGUV (R_\rho),
	\end{equation}
	where  $\cGUV ( \cdot) :=  V_\delta^\dagger  \cG_{V}(\cdot) V_{\delta}$.
\end{theorem}

\begin{proof}
	Let $P$ be such that $U = \begin{bmatrix} V_\delta & P \end{bmatrix}$ is unitary. Rotating the first equality in \cref{eq:rotate_const} using the unitary matrix $U$ yields an equivalent equality $U^\dagger  \cV_\delta(R_\delta) U = U^\dagger  \cG_V(R_\rho) U$. Applying the orthogonality of $V_{\delta}$, the left-hand side above becomes
	\begin{equation}
		\label{eq:UdeltaU}
		U ^\dagger  \cV_\delta (R_\delta) U =
		\begin{bmatrix}
			R_\delta & 0\cr 0 & 0
		\end{bmatrix}.
	\end{equation}
From facial reduction, it holds that $\range(\Vd) = \range (\cG_V)$ and thus $P^\dagger \cG_V = 0$. Therefore, the right hand-side becomes
\begin{equation}
		\label{eq:UGrhoU}
	U^\dagger  \cG_V(R_\rho) U 
= \begin{bmatrix} \Vd^\dagger  \\ P^\dagger  \end{bmatrix}  \cG_V(R_\rho) \begin{bmatrix} V_\delta & P \end{bmatrix}  = \begin{bmatrix}
	V_{\delta}^\dagger \cG_V(R_\rho) V_{\delta} & 0\\
	0 & 0
\end{bmatrix}.
\end{equation}

\end{proof}

\begin{theorem}
	\label{thm:rotate_const2}
	Let $R_\sigma \in  \H_+^{k_\sigma}$ and $R_\delta \in   \H_+^{k_\delta}$. It holds that
	\begin{equation}
		\cV_\sigma (R_\sigma ) = \cZ_V(R_\delta) \quad \Longleftrightarrow \quad R_\sigma =  \cZUV (R_\delta),
	\end{equation}
	where  $\cZUV ( \cdot) :=  V_\sigma^\dagger  \cZ_{V}(\cdot) V_{\sigma}$.
\end{theorem}
\begin{proof}
Using the unitary matrix $U = \begin{bmatrix} \Vs & P \end{bmatrix} $ in the proof of \Cref{thm:rotate_const}, we obtain the statement. 
\end{proof}

With \Cref{thm:rotate_const,thm:rotate_const2}, we reduce the number of linear constraints in \eqref{eq:optprobobjconstrthree} as below.

\begin{equation}
\label{eq:optprobobjintptFRreduced}
\begin{array}{rclll}
&\min  &   \Tr (R_\delta \log (R_\delta) ) - \Tr\big( 
R_\sigma \log(R_\sigma)\big) \\
&\text{s.t.} &   \GV ( R_\rho   ) = \gamma 
\\ && R_\sigma  - \cZUV(R_\delta) = 0
\\ &&  
R_\delta - \cGUV(R_\rho) = 0
\\&&   R_\rho \in \Hp^{n_\rho},  R_\sigma \in
\Hp^{k_\sigma},  R_\delta \in \Hp^{k_\delta}.
\end{array}
\end{equation}

We emphasize that the images of $\cZ_V$ and $\cG_V$ in \eqref{eq:optprobobjconstrthree} are both in $\H^k$  but the images of $\cZUV$ and $\cGUV$ in \eqref{eq:optprobobjintptFRreduced} are in $\H^{k_\sigma}$ and $\H^{k_\delta}$, respectively, and $k_\delta\le k_\sigma \le k$.
The facial reduction performed on the variables $\delta, \sigma$ may yield $k_\delta<k_\sigma$. Hence, the two trace terms in the objective function in \eqref{eq:optprobobjintptFRreduced} cannot be consolidated into one trace term in general.

\index{$k_\delta$, $k_\sigma$}

\begin{remark}
\label{lem:GUV_ZUV_preserve}
The mapping $\cGUV$ satisfies
the properties for $\cG$ in~\eqref{eq:Gmap}.
However, the properties in~\cref{eq:Zmap} do not hold
for the mapping $\cZUV$.
\end{remark}

\subsection{Final Model for QKD key rate calculation}
\label{sect:finalmodel}

In this section we have a main result, i.e.,~the main model that we work
on and the derivatives. 
We eliminate some of variables in the model \cref{eq:optprobobjintptFRreduced} to obtain a simplified formulation.
Define $\textdef{$\WcZ$} := \cZUV \circ \cGUV$ and $\textdef{$\WcG$} : = \cGUV$. We substitute $R_\sigma = \WcZ(R_\rho)$ and $R_\delta = \WcG(R_\rho)$ back in the objective function in \cref{eq:optprobobjintptFRreduced}. For simplification, and by abuse of notation, we set
\[
\fbox{$\rho \leftarrow R_\rho,\, \sigma \leftarrow R_\sigma,\,
	\delta \leftarrow R_\delta$.} 
\]
We obtain the final model for \QKD key rate calculation problem:
\begin{equation}
\label{eq:finalredinrho}
\begin{array}{rclll}
p^* =&\min  &  f(\rho) = 
\Tr\big( \WcG (\rho) (\log \WcG (\rho) ) \big) - \Tr \Big( \WcZ(\rho)  \log \WcZ (\rho)  \Big) 
\\
&\text{s.t.} &   \GV ( \rho   ) = \gV
\\&&   \textdef{$\rho \in \Hp^{n_\rho}$},
\end{array}
\end{equation} 
where $\gV \in \R^{\mv}$ for some $\mv \le m$. The final model is essentially in the same form as the original model \cref{eq:keyrateopt}; see also \Cref{prop:Zopprops}.
\index{$\mv$, number of linear constraints after \FR}

Note that the final model now has smaller number of variables   compared to the original problem~\cref{eq:keyrateopt}.
Moreover, the objective function $f$, with the modified linear maps $\WcG,\WcZ$, is well-defined and analytic
on $\rho\in \Hpp^{n_\rho}$, i.e.,~we have
\begin{equation}
\label{eq:rho_welldef}
\rho \succ 0 \implies \WcG(\rho) \succ 0 \implies
 \WcZ(\rho) 
\succ 0.\footnote{This follows from~\cite[Theorem 6.6]{con:70}, i.e.,~from  
$\relint (A C) = A \relint (C)$, where $C$ is a convex set and 
$A : \mathbb{E}^n \to \mathbb{E}^m$ is a linear map.}
\end{equation}
\label{pg:conex}

Some derivative background is given in~\Cref{sect:derivsquant}.
We conclude this section by presenting the derivative formulae for gradient and Hessian. The simple formulae in \Cref{thm:explgradHess} are a direct application of 
\Cref{lem:gradfromchainrule}. Throughout \Cref{sect:intpt} we work with these derivatives. 
\begin{theorem}[derivatives of regularized objective]
\label{thm:explgradHess}
Let $\rho \succ 0$. The gradient of $f$ in~\eqref{eq:finalredinrho} is
\[
\nabla f(\rho) =   \fbox{$\WcG^\dagger  (\log[\WcG(\rho)]) +
\WcG^\dagger  (I)$} -
   \fbox{$\WcZ^\dagger  (\log[\WcZ(\rho)]) +
\WcZ^\dagger  (I)$}.
\]
The Hessian in the direction $\Delta \rho$ is then
\[
\begin{array}{rcl}
\nabla^2 f(\rho)(\Delta \rho) 
&=&
     \fbox{$\WcG^\dagger 
(\log^\prime[\WcG(\rho)](\WcG(\Delta \rho)) $}
-
\fbox{$\WcZ^\dagger 
(\log^\prime[\WcZ(\rho)](\WcZ(\Delta \rho)) $} .
\end{array}
\] 
\end{theorem}

\index{subdifferential, $\partial f$}

\label{pg:subgradient}

Given a real-valued convex function $f$, a \textdef{subgradient} $\phi$
of $f$ at $x$ is a vector satisfying  $f(y) \ge f(x) + \<\phi,y-x\>$,
$\forall y$. The set of gradients of $f$ at $x$ is called the \emph{subdifferential}, denoted by $\partial f(x)$.
For a differentiable function $f$, the subdifferential is a singleton,
$\partial f(x) = \{ \nabla f(x)\}$; see e.g.,~\cite{con:70}.
\begin{theorem}
	\label{cor:fsubdiff}
Let $f$ be as defined in~\eqref{eq:finalredinrho} and 
let $\{\rho_i\}_i \subseteq \mathbb{H}_{++}^{n_\rho}$ 
with $\rho_i \to \bar \rho$.
If  we have the convergence $\lim_i \nabla f(\rho_i) = \phi$, then
\[
\phi \in \partial f(\bar \rho).
\]
\end{theorem}
\begin{proof}
	The result follows from the characterization of the subgradient as containing the convex hull of all limits of gradients, e.g.,~\cite[Theorem 25.6]{con:70}.
\end{proof}

\section{Optimality Conditions, Bounding, \GN Interior Point Method}
\label{sect:intpt}

Arguably, the most popular approach for solving \SDP problems is by
applying a path-following interior point approach to solving perturbed
optimality conditions using Newton's method.  However,
in general, numerical difficulties and instability arise in two ways.
First, the optimality conditions for \SDP problems are overdetermined, 
and a symmetrization is needed to apply a
standard Newton method. Second, block Gaussian elimination is applied to
the linearized Newton system to efficiently solve for the Newton direction.
This is done without regard to partial pivoting to avoid roundoff error
buildup. To avoid these instabilities,  
we apply a projected Gauss-Newton, \GNp, interior point approach,
\Cref{sect:projGNipalgor}, to solve the 
perturbed optimality conditions for our model~\cref{eq:finalredinrho}. 

In this section, we begin by presenting the optimality conditions for the model
\eqref{eq:finalredinrho}; then the \GN search direction is introduced 
in~\Cref{sect:GNsrchdir}; the projected versions are discussed 
in~\Cref{sect:projGN}; we present the algorithm itself
in~\Cref{sect:projGNipalgor}. 
We finish this section with bounding strategies in~\Cref{sect:dualbnds}. 
The important provable lower bound is
presented in~\Cref{sect:lborigprob}.

\subsection{Optimality Conditions and Duality}

We first obtain perturbed optimality conditions for~\cref{eq:finalredinrho}
with positive barrier parameters. This is most often done by 
using a barrier function
and adding terms such as $\mu_\rho \log \det ( \rho )$ to the Lagrangian.
After differentiation we obtain $\mu_\rho \rho^{-1}$ that we equate with
the dual  variable $Z_\rho$. After multiplying through by $\rho$ we
obtain the \textdef{perturbed complementarity equations},
e.g.,~$Z_\rho \rho -\mu_\rho I =0$.

\begin{theorem}
\label{thm:optcondfinalprob}
Let $L$ be the Lagrangian for \cref{eq:finalredinrho}, i.e.,
\[
L(\rho, y ) = f(\rho) +\langle y, \GV  (\rho)  - \gammaV \rangle,
\, y\in \R^{\mv}.
\]
The following holds for problem~\cref{eq:finalredinrho}.
\begin{enumerate}
\item
\[
\begin{array}{rcl}
p^*  
&=& \max\limits_{y}  
        \min\limits_{\rho\succeq 0} L(\rho,y).
\end{array}
\]
\item
\label{item:strduality}
The  \textdef{Lagrangian dual} of \cref{eq:finalredinrho} is
\[
\begin{array}{rcl}
d^*  
&=& \max\limits_{Z\succeq 0,y}  \left(
        \min\limits_{\rho} (L(\rho,y)-\langle Z,\rho\rangle)\right),
\end{array}
\]
and strong duality holds for~\cref{eq:finalredinrho}, i.e., $d^*
=p^* $
and $d^* $ is attained for some $(y,Z) \in \R^\mV\times \H_+^{n_\rho}$.
\item
\label{item:optcharac}
The primal-dual pair $(\rho,(y,Z))$, with $\partial f(\rho) \neq \emptyset$,
is optimal if, and only if,
\begin{equation}
\label{eq:pdoptcond}
\begin{array}{rcll}
0 &\in& \partial f(\rho) + \Gamma_V^\dagger (y) -Z & \text{(dual feasibility)} \\
0 &=& \GV (\rho) - \gammaV & \text{(linear primal feasibility)} \\
 0  &=& \langle \rho,Z\rangle & \text{(complementary slackness)} \\
0 &\preceq &\rho,Z& \text{(semidefiniteness primal feasibility)}.
\end{array}
\end{equation}

\end{enumerate}

\end{theorem}

\begin{proof}

The proof is given in~\Cref{sec:proofoptcondfinalprob}.
\end{proof}

\subsubsection{Perturbed Optimality Conditions }
\label{sec:pertoptconds}

Many interior-point based algorithms try to solve the optimality
conditions  \cref{eq:finalredinrho} by solving a sequence of perturbed
problems while driving the perturbation parameter $\mu \downarrow 0$.
The parameter $\mu$ gives a measure of the duality gap. 
In this section, we present the perturbed optimality conditions for
\QKDp. 
Note that the optimality conditions in~\cref{eq:pdoptcond}
assumed the existence of the subdifferential. This assumption is not
required for the perturbed optimality conditions as we can use existing
gradients.

\begin{theorem}
\label{thm:optcondpd}
The barrier function for~\Cref{eq:finalredinrho} with barrier parameter 
$\mu >0$ is
\[
B_\mu( \rho, y ) = f(\rho)       +\langle y, \GV  (\rho)    - \gammaV \rangle
-\mu \log \det (\rho) .
\]
With $Z = \mu \rho^{-1}$ scaled to $Z\rho -\mu I =0$, we obtain the
perturbed optimality conditions for \cref{eq:finalredinrho} at $\rho,  Z   \succ 0$, $y$:
\begin{equation}
\label{eq:pertoptcond}
\begin{array}{lcccl}
\text{dual feasibility } \ \ \ \quad  ( \nabla B_{\rho} =0)& : &
F_\mu^d  & = &  \nabla_\rho f(\rho) + 
  \Gamma_V^\dagger  (y)   - Z =0  
 \\
\text{primal feasibility } \quad (\nabla B_{y} =0 ) &:&
  F_\mu^p  & = &  \GV  (\rho)    - \gammaV  =0    \\
\text{perturbed complementary slackness }&:&
  F_\mu^c & = & Z \rho - \mu I= 0 . \cr 
\end{array}
\end{equation}
In fact, for each $\mu>0$ there is a unique primal-dual solution
$\rho_\mu,y_\mu,Z_\mu$ 
satisfying~\cref{eq:pertoptcond}. This
defines the central path as $\mu \downarrow 0$.
Moreover, 
\[
(\rho_\mu,y_\mu,Z_\mu) \underset{\mu\downarrow 0}{\to} 
(\rho,y,Z) \text{  satisfying } \cref{eq:pdoptcond}.
\]
\end{theorem}
\begin{proof}	
The optimality condition \eqref{eq:pertoptcond} follows from the
\label{pg:theconvex}
necessary and sufficient optimality conditions of the convex problem
\[
\min_{\rho} \{ f(\rho) - \mu \log\det (\rho) \ : \ \Gamma_V (\rho) = \gammaV \} 
\]
and setting $Z = \mu \rho^{-1}$. Note that $B_\mu$ is the Lagrangian function of this convex problem.
For each $\mu>0$ there exists a unique solution to \eqref{eq:pertoptcond} due to the strict convexity of the barrier term $-\mu \log \det (\rho)$ and boundedness of the level set of the objective.
The standard log barrier argument~\cite{HillierFrederickS2008LaNP,MR2001b:90002}
and \Cref{cor:fsubdiff} together give the last claim.
\end{proof}

\Cref{thm:optcondpd} above provides an interior
point path following method, i.e.,~for each $\mu\downarrow 0$ we solve
the pertubed optimality conditions
\begin{equation}
\label{eq:optcondpert}
F_\mu (\rho,y,Z)=  
\begin{bmatrix} 
  \nabla_\rho f(\rho) + \Gamma_V^\dagger  (y)   - Z  \cr
    \GV  (\rho)    - \gammaV   \cr
   Z \rho - \mu I 
\end{bmatrix}  = 0, \quad \rho,Z \succ 0.
\end{equation}
The question is how to do this efficiently.
The nonlinear system is overdetermined as
\[
F_\mu : \H^{n_\rho} \times \R^{\mV} \times  \H^{n_\rho} \to
\H^{n_\rho} \times \R^{\mV}\times \C^{n_\rho\times n_\rho}.
\]
Therefore we cannot apply Newton's method directly to the
nonlinear system~\cref{eq:optcondpert},
since the linearization does not yield a square system. 
\label{pgapplyGN}

\subsection{Gauss-Newton Search Direction}
\label{sect:GNsrchdir}

To avoid the instability that is introduced by symmetrizations for
applying Newton's method to the
overdetermined optimality conditions, we use a Gauss-Newton approach.
That is, to solve the optimality conditions~\cref{eq:optcondpert}, we consider
the equivalent nonlinear least squares problem
\[
\min_{\rho,Z\succ 0,y} g(\rho,y,Z) : =\frac 12 \|F_\mu (\rho,y,Z)\|^2
 =\frac 12 \|F_\mu^d (\rho,y,Z)\|_F^2
 +\frac 12 \|F_\mu^p (\rho)\|^2
 +\frac 12 \|F_\mu^c (\rho,Z)\|_F^2.
\]
\index{$g(\rho,y,Z)$, nonlinear least square function}
\index{nonlinear least square function, $g(\rho,y,Z)$}
The \textdef{Gauss-Newton direction, $\dGN$}, is the least squares solution of the
linearization 
\[
F_\mu^\prime(\rho,y,Z) \dGN = -F_\mu(\rho,y,Z),
\]
where $F_\mu^\prime$ denotes the Jacobian of $F_\mu$.
\index{$F_\mu^\prime$, Jacobian of $F_\mu$}
\index{Jacobian of $F_\mu$, $F_\mu^\prime$}
\begin{remark}
\label{rem:GNpopular}
The Gauss-Newton method is a popular method for solving nonlinear least
squares problems. It is arguably the method of choice for overdetermined
problems such as the one we have here. It is called Newton's method for
overdetermined systems, e.g.,~\cite{MR1651750}, see also the 
classical book~\cite{DennSch:83}. In particular, it is very successful in cases
where the residual at optimality is small. And, in our case the residual is zero at optimality.
Other symmetrization schemes are discussed in e.g.,~\cite{MR1778232}.
\end{remark}
\begin{lemma}
\label{lem:GNdescent}
 
Under a full rank assumption of $F_\mu^\prime(\rho,y,Z)$, we get
\[
\dGN = -((F_\mu^\prime)(\rho,y,Z)^\dagger  F_\mu^\prime(\rho,y,Z) )^{-1}
          (F_\mu^\prime(\rho,y,Z))^\dagger F_\mu(\rho,y,Z).
\]
Moreover, if $\nabla g(\rho,y,Z)\neq 0$, then
$\dGN$ is a descent direction for $g$.
\end{lemma}
\begin{proof}
The gradient of $g$ is, omitting the variables,
\[
\nabla g = (F_\mu^\prime)^\dagger  (F_\mu );
\]
and the Gauss-Newton direction is the least squares solution of the
linearization $F_\mu^\prime \dGN = -F_\mu$, i.e.,~under a full rank
assumption, we get the solution from the normal equations as
\[
\dGN = -((F_\mu^\prime)^\dagger  F_\mu^\prime )^{-1}
          (F_\mu^\prime)^\dagger F_\mu.
\]
We see that the inner product with the gradient is indeed negative, hence a descent direction.
\end{proof}

We now give an explicit representation of the linearized system for
\eqref{eq:optcondpert}. We define the (right/left matrix multiplication)
linear maps
\[
\cMZ,\ \cMrho  : \H^{n_\rho} \to \C^{n_\rho\times n_\rho}, \quad 
\textdef{$\cMZ(\Delta X) = Z\Delta X$}, \,  
\textdef{$\cMrho(\Delta X) = \Delta X\rho$}.
\]
Then the linearization of \eqref{eq:optcondpert} is 
\index{$\dGN$, \GN-direction}
\begin{equation}
\label{eq:GNorig}
\begin{array}{rcl}
F_\mu^\prime \dGN 
=
\begin{bmatrix}
\nabla^2 f(\rho) \Delta \rho + \Gamma_V^\dagger (\Delta y) -\Delta Z  \cr
\GV (\Delta \rho) \cr
Z \Delta \rho + \Delta Z \rho 
\end{bmatrix}
\vspace{.1in}
=
\begin{bmatrix}
\nabla^2 f(\rho)  & \Gamma_V^\dagger  & -I   \cr
\GV && \cr
\cMZ   &&   \cMrho 
\end{bmatrix}
\begin{pmatrix}
\Delta \rho\cr \Delta y \cr \Delta Z
\end{pmatrix}
\approx&
-F_\mu.
\end{array}
\end{equation}
We emphasize that the last term is in $\C^{n_\rho\times n_\rho}$ and the system is
overdetermined.
The adjoints of $\cMZ,\cMrho$ are discussed in~\Cref{sect:lintrsadj}, 
\Cref{lem:WRadj,lem:Srho}.
Solving the system \eqref{eq:GNorig}, we obtain  the \textdef{\GN-direction, $\dGN$}$\in \H^{n_\rho}\times
\R^{\mv}\times \H^{n_\rho}$.

\subsection{Projected Gauss-Newton Directions}
\label{sect:projGN}

The \GN direction in \eqref{eq:GNorig} solves a relatively large overdetermined
linear least squares system and does not explicitly exploit the zero blocks. We
now proceed to take advantage of the special structure of the linear
system by using projection and block elimination.

\subsubsection{First Projected Gauss-Newton Direction}
\label{sec:1stprojGN}

Given the system \eqref{eq:GNorig}, we can make a substitution for $\Delta Z$ using the first block equation 
\begin{equation}
\label{eq:deltaZ}
\Delta Z = F_\mu^d+\nabla^2 f(\rho) \Delta \rho + \Gamma_V^\dagger (\Delta y) .
\end{equation}
This leaves the two blocks of equations
\index{$\dGN$, \GN-direction}
\index{\GN-direction, $\dGN$}
\[
\begin{array}{rcl}
\left(F_\mu^{pc}\right)^\prime \begin{pmatrix}
\Delta \rho \cr \Delta y
\end{pmatrix}
&=&
\begin{bmatrix}
\GV ( \Delta \rho ) \cr
Z \Delta \rho +
\left(\nabla^2 f(\rho) \Delta \rho + \Gamma_V^\dagger ( \Delta y ) \right) \rho 
\end{bmatrix}
\vspace{.1in}
\\&=&
\begin{bmatrix}
\GV & \cr
\cMZ +\cMrho \nabla^2f(\rho)    &   \cMrho \Gamma_V^\dagger  
\end{bmatrix}
\begin{pmatrix}
\Delta \rho \cr \Delta y
\end{pmatrix}
\\&\approx&
-\begin{bmatrix}
F^p_\mu\cr F^c_\mu +F^d_\mu \rho
\end{bmatrix},
\end{array}
\]
where the superscript in $F_\mu^{pc}$ stands for the primal and complementary slackness constraints.

The adjoint equation follows:
\[
\left[  \left(F_\mu^{pc}\right)^\prime  \right]^\dagger 
\begin{pmatrix}
r_p \cr R_c
\end{pmatrix}
=
\begin{bmatrix}
\Gamma_V^\dagger  & \cMZ^\dagger  +\nabla^2f(\rho)\cMrho^\dagger  \cr
0 &  \GV\cMrho^\dagger 
\end{bmatrix}
\begin{pmatrix}
r_p \cr R_c
\end{pmatrix}.
\]
In addition, we can evaluate the condition number of the system using
$ \left(  \left(F_\mu^{pc}\right)^\prime  \right)^\dagger  \left(F_\mu^{pc}\right)^\prime $.
Note that we include the adjoints as they are needed for matrix free
methods that exploit sparsity.

\subsubsection{Second Projected Gauss-Newton Direction}
\label{sect:2ndprojGNdir}

We can further reduce the size of the linear system by making further 
variable substitutions. 
Recall that in \Cref{sec:1stprojGN} we solve the system with a variable in $\H^{n_\rho} \times \R^{\mV}$, i.e., $n_\rho^2+\mV$ number of unknowns. 
In this section, we make an additional substitution using the second block equation in \eqref{eq:GNorig} and reduce the number of the unknowns to $n_\rho^2$. 

\begin{theorem}
\label{thm:2ndprojGN}
Let $\hat{\rho}\in \H^{n_\rho}$ be a feasible point for $\GV (\cdot)= \gV$.
Let \textdef{$\cN^\dagger  :\R^{n_\rho^2-\mV}  \to \H^{n_\rho}$} be an
injective linear map in adjoint form
so that, again by abuse of notation and redefining the primal residual,
we have the nullspace representation:
\[
F_\mu^{p} = \GV (\rho) - \gV  \, \iff \, F_\mu^{p} = 
    \cN^\dagger (v)+ \hat \rho  -\rho  , \text{  for some }v. 
\]
Then the second projected \GN direction,
$\dGN = \begin{pmatrix} \Delta v\cr\Delta y  \end{pmatrix}$,
is found from the least squares solution of
\begin{equation}
\label{eq:proj2GNeqn}
\begin{array}{l}
\fbox{$
\left[Z  \cN^\dagger (\Delta v) +\nabla^2 f(\rho) 
\cN^\dagger (\Delta v)\rho \right]  + \left[ \Gamma_V^\dagger (\Delta y ) \rho \right]
  =
- F_\mu^c -Z F_\mu^{p} - \left( F_\mu^d +\nabla^2 f(\rho)F_\mu^{p}
\right) \rho.$}
\end{array} 
\end{equation}

\end{theorem}
\begin{proof}
Using the new primal feasibility representation, the perturbed optimality conditions in~\cref{eq:optcondpert} become:
\begin{equation}
\label{eq:modoptcondpertnullsp}
F_\mu (\rho,v,y,Z)=  
\begin{bmatrix} 
  F_\mu^d\cr
 F_\mu^{p}\cr
   F_\mu^c
\end{bmatrix} 
=
\begin{bmatrix} 
  \nabla_\rho f(\rho) + \Gamma_V^\dagger ( y )   - Z  \cr
 \cN^\dagger ( v )  + \hat \rho -\rho \cr
   Z \rho - \mu I 
\end{bmatrix}  
 = 0, \quad \rho,Z \succ 0.
\end{equation}
After linearizing the system \eqref{eq:modoptcondpertnullsp} 
we use the following to find the \GN search direction:
\[
\begin{array}{rcl}
F_\mu^\prime \dGN 
=
\begin{bmatrix}
\nabla^2 f(\rho) \Delta \rho + \Gamma_V^\dagger (\Delta y) -\Delta Z  \cr
\cN^\dagger (\Delta v) -\Delta \rho\cr
Z \Delta \rho + \Delta Z \rho 
\end{bmatrix}
\vspace{.1in}
\approx
-F_\mu.
\end{array}
\]
From the first block equation we have
\[
\begin{array}{rcl}
\Delta Z 
&=&
 F_\mu^d+\nabla^2 f(\rho) \Delta \rho + \Gamma_V^\dagger ( \Delta y )
\\&=&
 F_\mu^d+\nabla^2 f(\rho) (F_\mu^{p}+\cN^\dagger (\Delta v))
             + \Gamma_V^\dagger ( \Delta y ).
\end{array}
\]
From the second block equation, we have
\[
\Delta \rho = F_\mu^{p}+\cN^\dagger (\Delta v).
\]
Substituting $\Delta Z$ and $\Delta \rho$ into $Z \Delta \rho + \Delta Z \rho $ gives
\[
\begin{array}{rclcl}
Z \Delta \rho + \Delta Z \rho 
&=&
Z ( F_\mu^{p}+\cN^\dagger (\Delta v))
 +\left[F_\mu^d+\nabla^2 f(\rho) 
        (F_\mu^{p}+\cN^\dagger (\Delta v))
             + \Gamma_V^\dagger (\Delta y) \right]\rho
\\&=&
\left[Z  \cN^\dagger (\Delta v) +\nabla^2 f(\rho) 
\cN^\dagger (\Delta v)\rho \right]  
+ \left[ \Gamma_V^\dagger (\Delta y ) \rho \right]
+ Z F_\mu^{p} + \left( F_\mu^d +\nabla^2 f(\rho)F_\mu^{p}  \right) \rho .
\end{array} 
\]
Rearranging the terms, the third block equation becomes
\[
\begin{array}{rcl}
F_\mu^{c\prime}\begin{pmatrix}\Delta v\cr \Delta y\end{pmatrix}
 & = &
\left[Z  \cN^\dagger (\Delta v) +\nabla^2 f(\rho) 
\cN^\dagger (\Delta v)\rho \right]  + \left[ \Gamma_V^\dagger ( \Delta y ) \rho \right]
 \\& = &
- F_\mu^c -Z F_\mu^{p} - \left( F_\mu^d +\nabla^2 f(\rho)F_\mu^{p}  \right) \rho.
\end{array} 
\]
\end{proof}

The matrix representation of \eqref{eq:proj2GNeqn} is presented in \Cref{app:matrix_repre_sys}.
It is easy to see that the adjoint satisfying 
$ \langle F_\mu^{c\prime} (\dGN ) , R_c \rangle = 
\langle \dGN, (F_\mu^{c\prime})^\dagger (R_c) \rangle $ 
now follows:
\[
(F_\mu^{c\prime})^\dagger (R_c) =
\begin{bmatrix}
\cN \Hvec \cMZ^\dagger    + \cN \nabla^2 f(\rho)\Hvec \cMrho^\dagger   \cr
\GV \cMrho^\dagger  
\end{bmatrix}(R_c) .
\]
After solving the system \eqref{eq:proj2GNeqn}, we make back substitutions to recover the original variables. In other words, 
once we get  $(\Delta v, \Delta y)$ from solving \eqref{eq:proj2GNeqn}, 
we obtain $(\Delta \rho, \Delta y, \Delta Z)$ using the original system:
\[
\Delta \rho = F_\mu^{p}+\cN^\dagger (\Delta v), \, \quad
\Delta Z =
 F_\mu^d+\nabla^2 f(\rho) (F_\mu^{p}+\cN^\dagger (\Delta v))
             + \Gamma_V^\dagger (\Delta y ).
\]
\Cref{thm:exactPR} below illustrates cases where we maintain the exact primal feasibility. 

\begin{theorem}
\label{thm:exactPR}
Let $\alpha$ be a step length and consider the update 
\[
\rho_+ \leftarrow \rho + \alpha \Delta \rho = \rho + F_\mu^{p}+\alpha \cN^\dagger (\Delta v) .
\]
\begin{enumerate}
\item
If a step length one is taken ($\alpha=1$), 
then the new primal residual is exact, i.e.,
\[
F_\mu^p=\cN^\dagger ( v_+ ) + \hat \rho -\rho_+=0.
\]
\item 
Suppose that the exact primal feasibility is achieved. Then the primal residual is $0$ throughout the iterations regardless of the step length.
\end{enumerate}
\end{theorem}

\begin{proof}
If a step length one is taken for updating
\[
\rho_+ \leftarrow \rho + \Delta \rho = \rho + F_\mu^{p}+\cN^\dagger (\Delta v),
\]
then the new primal residual 
\[
\begin{array}{rcl}
(F_\mu^p)_+
&=&
\cN^\dagger (v_+) + \hat \rho - \rho_+ 
\\&=& 
\cN^\dagger (v+\Delta v) + \hat \rho - \rho - F_\mu^{p}-\cN^\dagger (\Delta v)
\\&=&
\cN^\dagger (v) + \hat \rho - \rho - \cN^\dagger (v) -\hat \rho + \rho 
\\&=&
0.
\end{array}
\]
In other words, as for Newton's  method, a step length of one implies that the new
residuals are zero for linear equations.

We can now change the line
search to maintain $\rho_+ = \cN^\dagger (v+ \alpha \Delta v) -\hat
\rho \succ 0$ and preserve exact primal feasibility.
Assume that $F_\mu^{p}=0$.
\[
\begin{array}{rcl}
\rho_+ 
\leftarrow 
 \rho + \alpha \Delta \rho
= \rho + \alpha ( F_\mu^{p}+\cN^\dagger (\Delta v))
= \rho + \alpha \cN^\dagger (\Delta v)
\end{array}
\]
Now, we see that
\[
\GV (\rho_+) = \GV ( \rho + \alpha \cN^\dagger (\Delta v) )
= \GV ( \rho ) = \gamma,
\]
where the last equality follows from the exactly feasibility assumption.
\end{proof}

\subsection{Projected Gauss-Newton Primal-Dual Interior Point Algorithm}
\label{sect:projGNipalgor}

We now present the pseudocode for the Gauss-Newton primal-dual interior point method in~\Cref{algo:GNalgo}.

It is a series of steps that find the least squares solution
of the over-determined linear
system~\cref{eq:proj2GNeqn}, while decreasing the perturbation parameter
$\mu \downarrow 0$, and maintaining the positive definiteness of $\rho, Z$.
\Cref{algo:GNalgo} is summarized as follows: 
\begin{enumerate}
\item At each iteration, we find the projected \GN direction described in
\Cref{sect:2ndprojGNdir}; 
See \Cref{algo:GNalgo} 
lines: \ref{lst:line2}, \ref{lst:line3} and \ref{lst:line4}.
\item
We then choose a step length that maintains strict feasibility. 
Whenever the step length is one, we attain
primal feasibility for all future iterations; 
See \Cref{algo:GNalgo} line: \ref{lst:line5}.
\item
We decrease the perturbation parameter $\mu$ appropriately, and proceed
to the next iteration;
See \Cref{algo:GNalgo} line: \ref{lst:line7}.
\item
The algorithm stops when the relative duality gap  reaches a prescribed
tolerance or we reach our prescribed
maximum number of iterations; 
See~\Cref{sect:implheur} for implementation details on stopping
criteria, preconditioning, etc.
\end{enumerate}

\index{algorithm, \GN interior point for \QKDp}
\begin{algorithm}
\caption{Projected Gauss-Newton Interior Point Algorithm for \QKDp}
\begin{algorithmic}[1] 
\REQUIRE $\rho\succ 0$, $\mu \in \R_{++}$, $\eta \in (0,1)$
\WHILE{stopping criteria is not met}
\STATE solve \eqref{eq:proj2GNeqn} for $(\Delta v, \Delta y)$ \label{lst:line2}
\STATE $\Delta \rho = F_\mu^{p}+\cN^\dagger (\Delta v)$
\label{lst:line3}
\STATE $\Delta Z =  F_\mu^d+\nabla^2 f(\rho) (F_\mu^{p}+\cN^\dagger (\Delta v))
             + \Gamma_V^\dagger (\Delta y) $
\label{lst:line4}
\STATE choose step length $\alpha$
\label{lst:line5}
\STATE $(\rho,y,Z) \leftarrow (\rho,y,Z) + \alpha (\Delta \rho, \Delta y, \Delta Z)$
\label{lst:line6}
\STATE $\mu \leftarrow \langle \rho,Z\rangle/n_\rho;\, \mu \leftarrow \eta \mu$
\label{lst:line7}
\ENDWHILE
\end{algorithmic}
\label{algo:GNalgo}
\end{algorithm}

\subsection{Dual and Bounding}
\label{sect:dualbnds}
We first look at upper bounds\footnote{Our discussion about upper bounds
here is about upper bounds for the given optimization problem, which are
not necessarily key rate upper bounds of the \QKD protocol under study. This is because the constraints that one feeds into the algorithm might not use all the information available to constrain Eve's attacks.} found from feasible solutions
in~\Cref{prop:iterrefine}. Then we use 
the dual program to provide provable lower bounds for the \FR
problem \eqref{eq:keyrateopt} thus providing lower bounds for the
original problem with the accuracy of \FRp.

\subsubsection{Upper Bounds}
A trivial upper bound is obtained as soon as we have a primal feasible
solution $\hat \rho$ by evaluating the objective function. Our algorithm is a
primal-dual \emph{infeasible} interior point approach. Therefore we
typically have approximate linear feasibility $\GV (\hat \rho) \approx \gV$;
though we do maintain positive definiteness $\hat \rho\succ 0$ throughout the
iterations. Therefore, once we are close to feasibility we can project
onto the affine manifold and hopefully maintain positive definiteness,
i.e.,~we apply iterative refinement by finding the projection
\[
\min_\rho \left\{ \frac 12 \|\rho - \hat \rho\|^2  \ : \ \GV (\rho) = \gV \right\}.
\]
\begin{prop}
\label{prop:iterrefine}
Let $\hat \rho \succ 0,\,F_\mu^p =  \GV  ( \hat \rho )    - \gammaV$. Then
\[
\rho = \hat \rho - \GV^{-1} F_\mu^p 
= \argmin_\rho  \left\{ \frac 12 \|\rho - \hat \rho\|^2  \ : \ \GV (\rho) = \gV \right\},
\]
where we denote \textdef{$\GV^{-1}$, generalized inverse}.
If $\rho \succeq 0$, then $p^*  \leq f(\rho)$.
\end{prop}
\index{generalized inverse, $\GV^{-1}$}
In our numerical experiments below we see 
that we obtain valid upper bounds
starting in the early iterations and, as we use a Newton type method,
we maintain exact primal feasibility throughout the iterations resulting
in a zero primal residual, and no further need for the projection.
As discussed above, we take a step length of one as soon as
possible. This means that exact primal feasibility holds for the
remaining iterations and we keep improving the upper bound at each
iteration.

\subsubsection{Lower Bounds for \FR Problem}
\label{sect:LBforFR}
Facial reduction for the affine constraint means that the corresponding
feasible set of the original problem lies within the minimal face
$V_\rho \H^{n_\rho}_+ V_\rho^\dagger$ of the semidefinite cone.
Since we maintain positive definiteness for $\rho,Z$ during the
iterations, we can obtain a lower bound using weak duality.
Note that $\rho \succ 0$ implies that the gradient 
$\nabla f(\rho)$ exists.
\begin{cor}[lower bound for \FR~\cref{eq:finalredinrho}]
\label{cor:lowerbound}
Consider the problem~\cref{eq:finalredinrho}.
Let $\hat \rho,\hat y$ be a primal-dual iterate with $\hat
\rho\succ 0$. Let
\[
\bar Z = \nabla f(\hat \rho) +\Gamma_V^\dagger (\hat y) .
\]
If $\bar Z\succeq 0$,
then a lower bound for problem~\cref{eq:finalredinrho} is
\[
p^*  \geq f(\hat\rho) + \langle \hat y,  \GV (\hat \rho) -\gV \rangle -\langle
\hat \rho,\bar Z\rangle.
\]
\end{cor}
\begin{proof}
Consider the dual problem
\[
d^*  = \max_{y,Z\succeq 0}  
        \min_{\rho\in \H^{n_\rho}} L(\rho,y)-\langle Z,\rho\rangle.
\]
We now have dual feasibility
\[
\bar Z\succeq 0, \,
 \nabla f(\hat \rho) +\Gamma_V^\dagger (\hat y) - \bar Z = 0 \implies
\hat \rho \in \argmin_\rho L(\rho,\hat y) - \langle \bar Z,\rho\rangle.
\]
Since we have dual feasibility, weak duality
in~\Cref{thm:optcondfinalprob},~\Cref{item:strduality} as stated
in the dual problem above yields the result.
\end{proof}
\begin{remark}\label{remark:lowerbound}
We note that the lower bound in~\Cref{cor:lowerbound}
is a simplification of the approach in~\cite{winick2017reliable}, where
after a near optimal solution is found, 
a dual problem of a linearized problem is solved using CVX in MATLAB. Then a strong duality
theorem is assumed to hold and is applied along with a linearization of
the objective function.
Here we do not assume strong duality, though it holds for the facially
reduced problem. And we apply weak duality
to get a theoretically guaranteed lower bound. 

We emphasize that this holds within the margin of error of the \FRp. 
Recall that we started with the problem in~\cref{eq:absQKD}. If we only
apply the accurate \FR based on spectral decompositions, then the lower
bound from~\Cref{cor:lowerbound} is accurate and theoretically valid up
to the accuracy of the spectral decompositions.\footnote{Note that the
condition number of the spectral decomposition of Hermitian matrices is
$1$; see e.g.,~\cite{MR98m:65001}.} In fact, in our numerics, we can
obtain tiny gaps of order $10^{-13}$ when requested; and we have never
encountered a case where the lower bound is greater than the upper
bound. Thus the bound  applies to our original problem as well. 
Greater care must be taken if we had to apply \FR to the entire
constraint $\Gamma(\rho)=\gamma$. The complexity of \SDP feasibility is
still not known. Therefore, the user should be aware of the difficulties
if the full \FR is done. 

A corresponding result for a lower bound for the original problem is
given in~\Cref{cor:lowerboundpert}.
\end{remark}

\subsubsection{Lower Bounds for the Original Problem}
\label{sect:lborigprob}
We can also obtain a lower bound for the case where \FR is performed with 
some error. Recall that we assume that the original
problem~\cref{eq:absQKD} is feasible. We follow the same arguments as 
in~\Cref{sect:LBforFR} but apply it to the original problem. All that
changes is that we have to add a small perturbation to the optimum
$V_\rho \hat R V_\rho^\dagger$ from the \FR problem in order
to ensure a positive definite $\rho$ for differentiability.
The exposing vector from \FR process presents an intuitive choice for
the perturbation.
\begin{cor}
\label{cor:lowerboundpert}
Consider the original problem~\cref{eq:absQKD} and the results
from the theorem of the alternative, \Cref{lem:FRfarkas}, for fixed $y$:
\begin{equation}
\label{eq:thmalternforLB}
0 \neq W = \Gamma^\dagger (y) \succeq 0, \ \gamma^\dagger y =
\epsilon_\gamma, \quad \epsilon_\gamma \geq 0.
\end{equation}
Let the orthogonal spectral decomposition be
\[
W = \begin{bmatrix} V & N \end{bmatrix}  
\begin{bmatrix} D_\delta & 0 \cr 0 & D_>  \end{bmatrix}
\begin{bmatrix} V & N \end{bmatrix}^\dagger, \, D_> \in \mathbb{S}_{++}^r.
\]
Let $0\preceq \eta\approx W$ be the (approximate) 
exposing vector obtained as the nearest rank $r$ positive semidefinite
matrix to W,
\[
W = N D_> N^\dagger + V D_\delta V^\dagger 
= \eta + V D_\delta V^\dagger  .
\]

Let $\hat R, \hat y$ be a primal-dual iterate for the \FR problem,
with $\hat R\succ 0$. 
Add a small perturbation matrix $\Phi \succ 0$ to guarantee that the
approximate optimal solution
\[
\hat \rho_\phi = V\hat R V^\dagger +N\Phi N^\dagger \succ 0.
\]
Without loss of generality, let $\hat y$ be a dual variable
for~\cref{eq:absQKD}, adding zeros to extend the given  vector if needed. Set
\begin{equation}
\label{eq:generalZpsd}
\bar Z_\phi = \nabla f(\hat \rho_\phi ) 
	  + \Gamma^\dagger(\hat y).
\end{equation}

If $\bar Z_\phi  \succeq 0$,
then a lower bound for the original problem~\cref{eq:absQKD} is
\begin{equation}
\label{eq:generalLB}
p^*  \geq f(\hat \rho_\phi) + \langle \hat y,  \Gamma (\hat \rho_\phi)
 -\gamma \rangle 
       -\langle \hat \rho_\phi,\bar Z_\phi\rangle.
\end{equation}
\end{cor}
\begin{proof} 
By abuse of notation, we let $f,L$ be the objective function and
Lagrangian for ~\cref{eq:absQKD}. 
Consider the dual problem
\[
d^*  = \max_{y,Z\succeq 0}  
        \min_{\rho\in \H^n} \big(L(\rho,y)-\langle Z,\rho\rangle\big).
\]
We now have dual feasibility
\[
\bar Z_\phi\succeq 0, \,
 \nabla f(\hat \rho_\phi) +\Gamma^\dagger (\hat y) - \bar Z_\phi = 0 \implies
\hat \rho_\phi \in \argmin_\rho \big(L(\rho,\hat y) - \langle \bar Z_\phi,\rho\rangle\big).
\]
Since we have dual feasibility, weak duality
in~\Cref{thm:optcondfinalprob},~\Cref{item:strduality} as stated
in the dual problem above yields the result.
\end{proof}
\begin{remark}
We note that~\cref{eq:generalZpsd} with $\bar Z_\phi$ is dual feasibility
(stationarity of the Lagrangian) for an optimal $\hat \rho_\phi$.
Therefore, under continuity arguments, we expect $\bar Z_\phi\succeq 0$ to
hold as well.

In addition, for implementation we need to be able to evaluate $\nabla
f(\hat \rho_\phi)$. Therefore, we need to form the positive definite
preserving maps $\WcG, \WcZ$, but without performing \FR on the
feasible set. That we can do this accurately using a spectral
decomposition follows from~\Cref{lem:tranform}.
\end{remark}

\section{Numerical Testing}\label{sect:tests}

We compare our algorithm to other algorithms by considering six \QKD
protocols including four variants of the Bennett-Brassard 1984 (BB84)
protocol, twin-field \QKD and discrete-modulated
continuous-variable \QKD. In~\Cref{sect:testexamples} we include the descriptions of protocol examples
that we use to generate instances for the numerical tests. We note that while it is possible to simplify the optimization problem of some protocols using protocol-specific properties as discussed in \Cref{sec:intro_convex}, we have not performed those protocol-specific simplifications since we aim to demonstrate the generality of our method for a wide class of protocols and in particular the ability to handle problems with considerably large problem sizes.

We continue with the tests
in~\Cref{sect:compnumericstheory,sect:challengingtests,sect:comperform}.
This includes security analysis of some selected \QKD protocols and 
comparative performances among different algorithms. 
In particular, in~\Cref{sect:compnumericstheory}, we compare the results
obtained by our algorithm with the analytical results for selected test examples where tight analytical results can be obtained. 
In \Cref{sect:challengingtests}, we present results where it is quite
challenging for the previous method in~\cite{winick2017reliable} to
produce tight lower bounds. In particular, we consider the
discrete-modulated continuous-variable \QKD protocol and compare results obtained in~\cite{Lin2019}. In~\Cref{sect:comperform}, we compare performances among different algorithms in terms of accuracy and running time.

\subsection{Comparison between the Algorithmic Lower Bound and the Theoretical Key Rate}
\label{sect:compnumericstheory}

We compare results from instances for which there exist tight analytical key rate expressions to demonstrate that our Gauss-Newton method can achieve high accuracy with respect to the analytical key rates. There are known analytical expressions for entanglement-based BB84, prepare-and-measure BB84 as well as measurement-device-independent BB84 variants mentioned in \Cref{sect:testexamples}. We take the measurement-device-independent BB84 as an example since it involves the largest problem size among these three examples and therefore more numerically challenging. In \Cref{fig:experiment1}, we present instances with different choices of parameters for data generation. The instances are tested with a desktop computer that runs with the operating system Ubuntu 18.04.4 LTS, MATLAB version 2019a, Intel Xeon CPU E5-2630 v3 @ 2.40GHz $\times$ 32 and 125.8 Gigabyte memory. We set the tolerance $\epsilon  = 10^{-12}$ for the Gauss-Newton method. 

In \Cref{fig:experiment1}, the numerical lower bounds from the Gauss-Newton method are close to the analytical results to at least 12 decimals and in many cases they agree up to 15 decimals.

\begin{figure}[ht!]
	\centering

	\includegraphics[height=6cm]{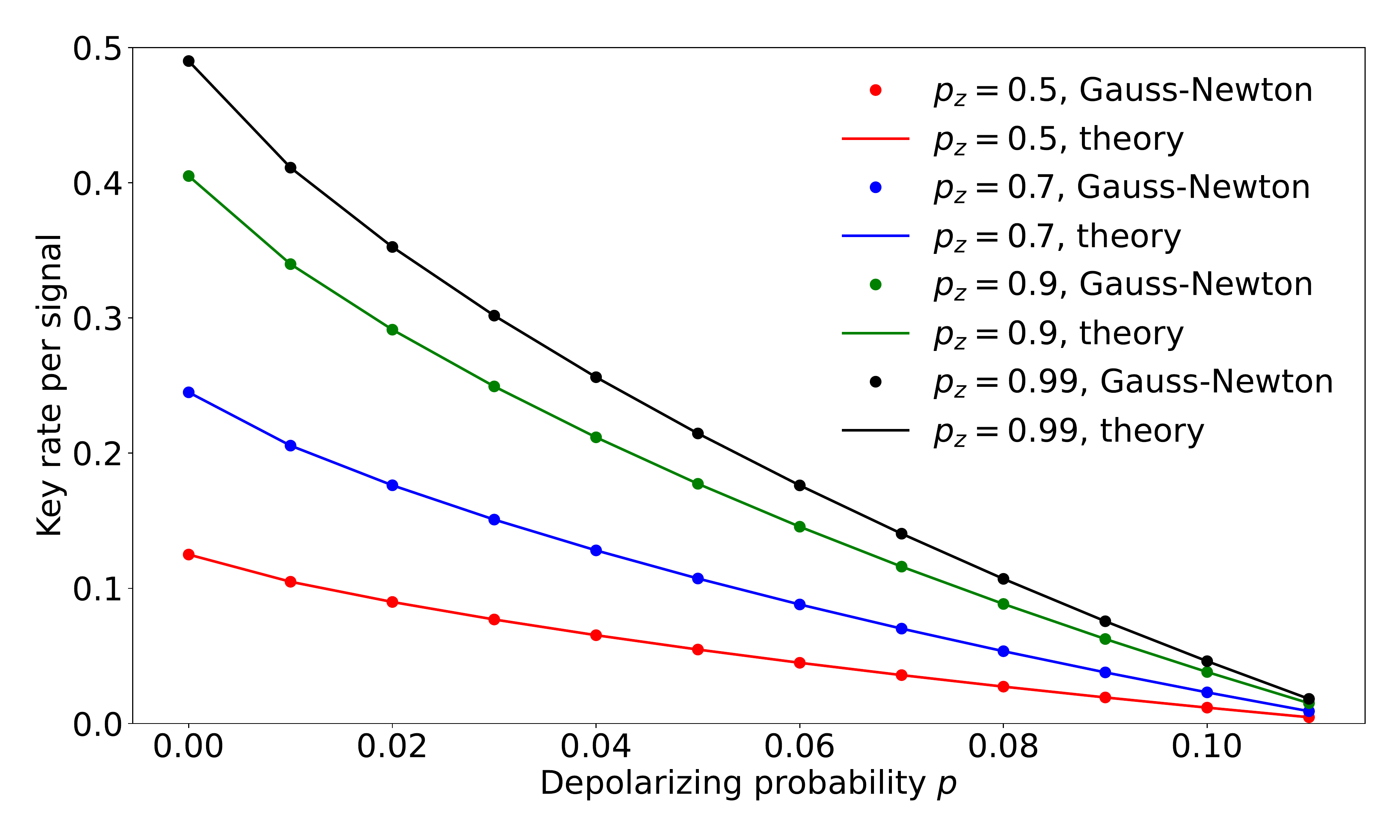}
	\caption{Comparisons of key rate for measurement-device-independent BB84 (\Cref{sect:example3}) between our Gauss-Newton method and the known analytical key rate.}
	\label{fig:experiment1}
\end{figure}

As noted in \Cref{sect:example5}, analytical results are also known when the channel noise parameter $\xi$ is set to zero since in this case, one may argue the optimal eavesdropping attack is the generalized beam splitting attack. This means the feasible set contains essentially a single $\rho$ up to unitaries. Since our objective function is unitarily invariant, one can analytically evaluate the key rate expression. In \Cref{fig:experiment2}, we compare the results from the Gauss-Newton method with the analytical key rate expressions for different choices of distances $L$ (See \Cref{sect:example5} for the description about instances of this protocol example). These instances were run in the same machine as in \Cref{fig:experiment1}. We set the tolerance $\epsilon  = 10^{-9}$ for the Gauss-Newton method.

\subsection{Solving Numerically Challenging Instances}
\label{sect:challengingtests}  

We show results where the Frank-Wolfe method without \FR has difficulties in providing tight lower bounds in certain instances. In \Cref{fig:experiment2}, we plot results obtained previously in~\cite[Figure 2(b)]{Lin2019} by the Frank-Wolfe method without \FRp. In particular, results from Frank-Wolfe method have visible differences from the analytical results starting from distance $L = 60$ km. In addition the lower bounds are quite loose once the distance reaches $150$ km. In fact, there are points like the one around $180$ km where the Frank-Wolfe method cannot produce nontrivial (nonzero) lower bounds. On the other hand, the Gauss-Newton method provides much tighter lower bounds.    

\begin{figure}[ht!]
	\centering
	\includegraphics[height=7cm]{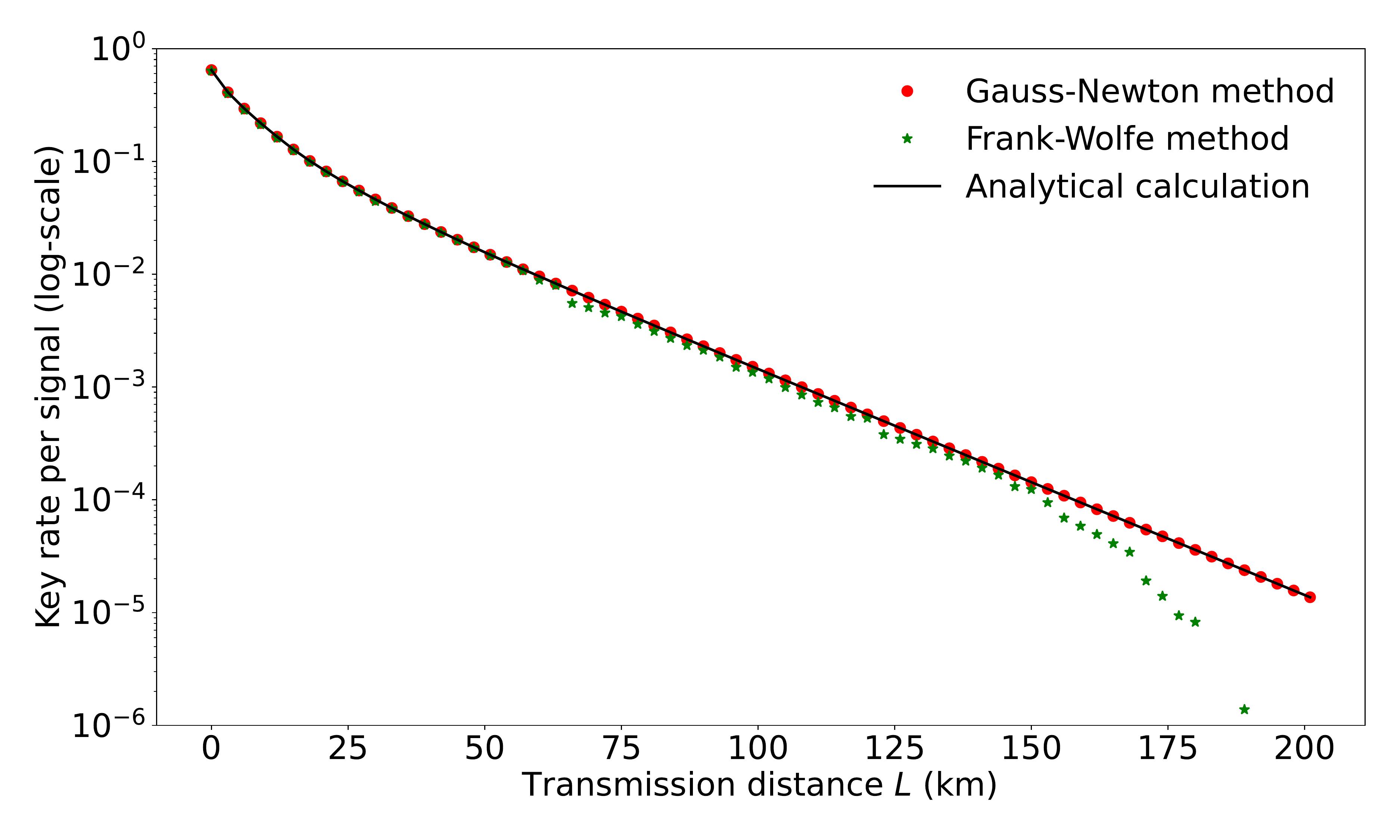}
	\caption{Comparison of key rate for discrete-modulated
continuous-variable \QKD (\Cref{sect:example5}) among our Gauss-Newton method, the Frank-Wolfe method and analytical key rate for the noise $\xi = 0$ case.}
	\label{fig:experiment2}
\end{figure}

In \Cref{fig:experiment2test6}, we show another example to demonstrate the advantages of our method. These instances were run in the same machine as in \Cref{fig:experiment2}. For this discrete-phase-randomized BB84 protocol with $5$ discrete global phases (see \Cref{sect:example6} for more descriptions), the previous Frank-Wolfe method was unable to find nontrivial lower bounds. This is because the previous method can only achieve an accuracy around $10^{-3}$ for this problem due to the problem size. This is insufficient to produce nontrivial lower bounds for many instances since the key rates are on the order of $10^{-3}$ or lower as shown in \Cref{fig:experiment2test6}. On the other hand, due to high accuracy of our method, we can obtain meaningful key rates. The advantage of high accuracy achieved by our method enables us to perform security analysis for protocols that involve previously numerically challenging problems. Like the discrete-phase-randomized BB84 protocol, these protocols involve more signal states, which lead to higher-dimensional problems.

\begin{figure}[ht!]
	\centering

	\includegraphics[height=7cm]{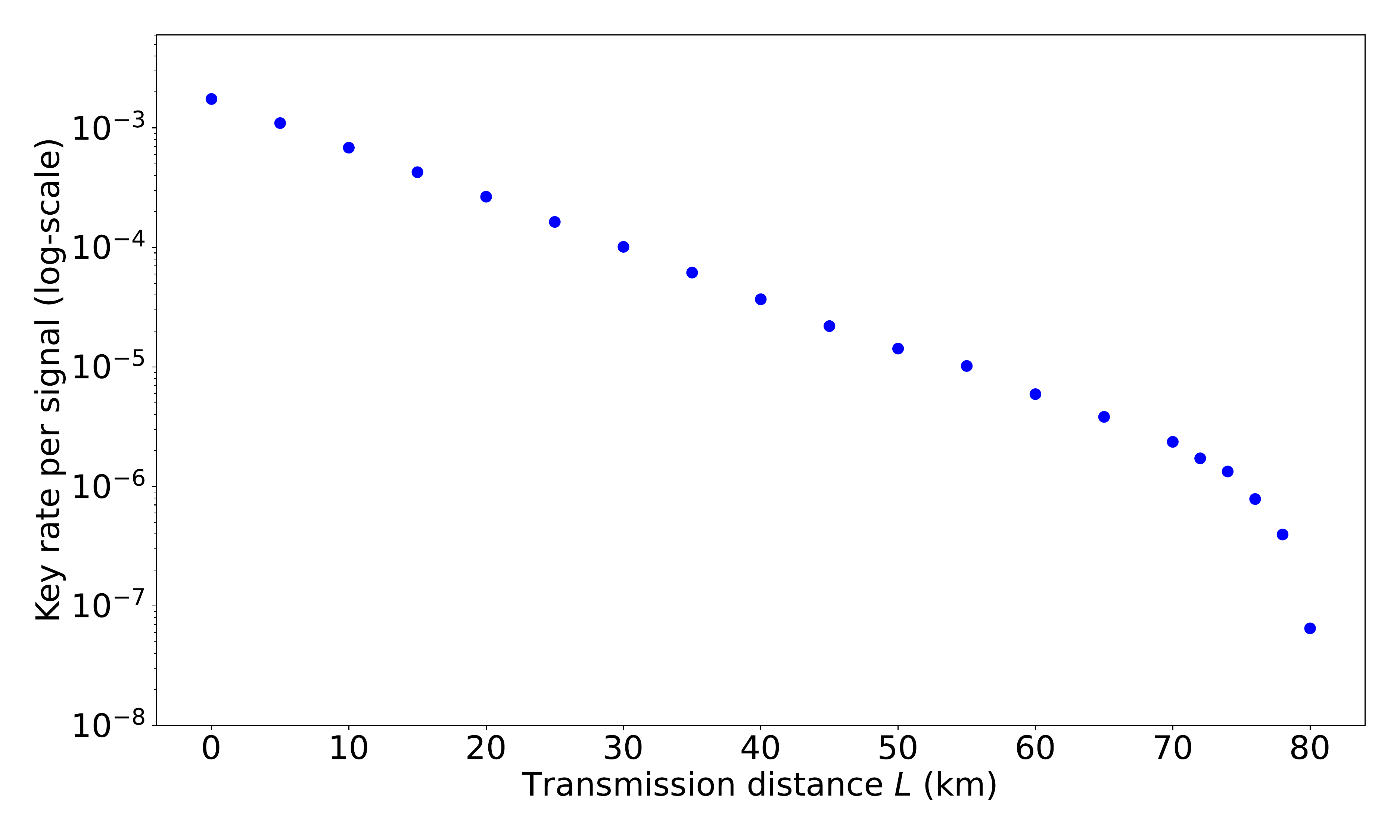}
	\caption{Key rate for discrete-phase-randomized BB84 (\Cref{sect:example6}) with the number of discrete global phases $c=5$. In this plot, the coherent state amplitude is optimized for each distance by a simple coarse-grained search over the parameter regime.}
	\label{fig:experiment2test6}
\end{figure}
\subsection{Comparative Performance}
\label{sect:comperform}

In this section we examine the comparative performance among three algorithms; the Gauss-Newton method, the Frank-Wolfe method and cvxquad. 
The Gauss-Newton method refers to the algorithm developed throughout this paper. The Frank-Wolfe method refers to the algorithm developed in~\cite{winick2017reliable} and cvxquad is developed in~\cite{cvxquad}.
We use \Cref{table:FINAL} to present detailed reports on some selected instances. More numerics are reported throughout \Cref{table:table1,table:table2,table:table3,table:table4,table:table5,table:table6} in  \Cref{sec:additionalreport}.

The instances are tested with MATLAB version 2021a using
Dell PowerEdge R640	Two Intel Xeon Gold 6244 8-core 3.6 GHz (Cascade Lake)	with 192 Gigabyte memory.
For the instances corresponds to the DMCV protocol, we used the tolerance $\epsilon = 10^{-9}$ and the tolerance $\epsilon = 10^{-12}$ was used for the remaining instances. The maximum number of iteration was set to $80$ for the Gauss-Newton method.

\begin{table}[H]
\hspace{-0.5cm}
\tiny{
\input{table/testall.tex}
}
\caption{Numerical Report from Three Algorithms}
\label{table:FINAL}
\end{table}

In \Cref{table:FINAL} \textbf{Problem Data} refers to the data used to generate the instances. 
\textbf{Gauss-Newton} refers to the Gauss-Newton method. \textbf{Frank-Wolfe} refers to the Frank-Wolfe algorithm used in~\cite{winick2017reliable} and we use `with \FR (w/o \FRp , resp)' to indicate the model is solved with \FR (without \FRp , resp). 
The header \textbf{cvxquad with \FRp} refers to the algorithm provided by~\cite{cvxquad} with \FR reformulation.
If a certain algorithm fails to give a reasonable answer within a reasonable amount of time, we give a `$\star\star$' flag in the gap followed by the time taken to obtain the error message.
We use `n/a' to indicate the instances for which cvxquad is not applicable due to the size differences in the images under $\WcG$ and $\WcZ$ due to \FRp.

The following provides details for the remaining headers in \Cref{table:FINAL}.
\begin{enumerate}
\item \textbf{protocol}: the protocol name; refer to \Cref{sect:testexamples}; 
\item \textbf{parameter}: the parameters used for testing; see \Cref{sect:example1} - \Cref{sect:example6} for the ordering of the parameters; 
\item \textbf{size}: the size $(n,k)$ of original problem; $n,k$ are defined in \eqref{eq:Gmap};
\item \textbf{gap}:
the relative gap between the bestub and bestlb;
\begin{equation}
\label{def:relgap}
	\frac
	       {\text{bestub - bestlb}}
	 {1+
	  \frac{|\text{bestub}| + |\text{bestlb}|}2} .
\end{equation}

\item \textbf{time}: time taken in seconds. 
\end{enumerate}

We make some discussions on the formula \cref{def:relgap}.
The best upper bound from Gauss-Newton algorithm is used for all instances for `bestub' in \cref{def:relgap}.
The Gauss-Newton algorithm computes the lower bounds as presented in \Cref{cor:lowerbound}.
The Frank-Wolfe algorithm presented in~\cite{winick2017reliable} obtains the lower bound by a linearization technique near the optimal. As presented in~\cite{cvxquad},
cvxquad uses the semidefinite approximations of the matrix logarithm. The lower bounds from cvxquad are often larger than the theoretical optimal values. This observation indicates that the lower bounds from cvxquad are not reliable.   
Therefore, we adopt the lower bound strategy used in~\cite{winick2017reliable} for cvxquad.

We now discuss the results in \Cref{table:FINAL}.
Comparing the two columns \textbf{gap} and \textbf{time} among
the different methods,
we see that the Gauss-Newton method outperforms other 
algorithms in both the accuracy and the running time. 
For example, comparing \textbf{Gauss-Newton} and
\textbf{Frank-Wolfe with FR}, the gaps and running times from \textbf{Gauss-Newton} are competitive.
There are three instances that \textbf{Gauss-Newton} took longer time. We
emphasize that the gap values with \textbf{Gauss-Newton} illustrate
much higher accuracy.

We now illustrate that the reformulation strategy via \FR contributes to superior algorithmic performances. 
For the columns \textbf{Frank-Wolfe with FR} and 
\textbf{Frank-Wolfe w/o FR}  in \Cref{table:FINAL},
the \FR reformulation contributes to not only giving tighter gaps but also reducing the running time significantly. 
We now consider the column corresponding to 
\textbf{cvxquad with FR}  in \Cref{table:FINAL}.
We see that the algorithm fails (marked with `$\star\star$') with some instances due to the memory shortage.
Facial reduction indeed contributes to the reduction on the problem
sizes. For example, we reduced the problem sizes in
\Cref{table:reducesize}.
\label{pg:problemsizereduction}
\begin{table}[h]
\centering
\begin{small}
\begin{tabular}{|cc|cc|}
\hline
protocol & parameter & $(n,m)$ & $(n_\rho,\mV)$ \\
\hline
pmBB84  & $(0.5,0.05)$ & $(8,21)$ &  $(4,8)$ \\
mdiBB84 & $(0.5,0.05)$ & $(48,305)$ & $(12,34)$ \\
\hline
\end{tabular}
\end{small}
\caption{Reduction in Problem Sizes}
\label{table:reducesize}
\end{table}

\index{$\mV$}
\index{$n_\rho$}

\section{Conclusion}
\label{sect:concl}
\subsection{Summary}
In this paper we have presented a robust numerical method for finding 
\emph{provable tight lower bounds} for the convex optimization problem
that finds the key rate for the \QKD problem. 
Our empirical evidence illustrates consistent 
significant improvements in solution
time and accuracy over previous methods. In particular, we solve 
many problems close to machine accuracy and provide theoretical provable 
accurate lower bounds. This includes previously unsolved problems.
(See e.g.,~\Cref{table:FINAL}.)

This paper used novel convex optimization techniques applied specifically
to the \QKD problem. This includes reformulations of the convex
optimization problem that finds the key rate. The result is a
regularized problem that avoids the need for previously used
perturbations and resulting possible instabilities.
Below, we give a summary of the contributions for the model reformulation and for the algorithm.

\subsubsection{Summary of the Model Reformulation}

We have reformulated, simplified, and stabilized the model 
for \QKD key rate calculation through the sequence 
\[
\cref{eq:keyrateopt}
\xrightarrow{\text{(1)}}
\cref{eq:keyrateoptequivsigdel}
\xrightarrow{\text{(2)}}
\cref{eq:optprob}
\xrightarrow{\text{(3)}}
\cref{eq:optprobobjconstrthree}
\xrightarrow{\text{(4)}}
\cref{eq:optprobobjintptFRreduced}
\xrightarrow{\text{(5)}}
\cref{eq:finalredinrho},
\]
via (1) variable substitutions;
(2) property of $\cZ$ from \Cref{prop:Zopprops};
(3) facial reduction on $\rho,\delta,\sigma$;
(4) rotation of the constraints;
(5) substituting the constraints back to the objective.

\subsubsection{Summary of \Cref{algo:GNalgo}}
\label{sect:summalgor}
Our algorithm \Cref{algo:GNalgo} is based
on a standard primal-dual interior-point approach applied to the \FR
stabilized model. However, it differs in several ways.
\begin{enumerate}
\item
We modify the primal feasibility to use a nullspace representation.
Therefore, both primal and dual feasibility have a similar representation. 
\item
We use a projected Gauss-Newton search direction to account for the
overdetermined least squares problem arising from the optimality
conditions. This means we project the Gauss-Newton direction after
substituting using the primal-dual linear feasibility equations.
\item
We exploit the exact feasibility of linear constraints after a step
length one for the Gauss-Newton method. Therefore, we attempt to take a
primal and/or dual step length one as soon as possible. Exact feasibility
results.
\item
We use a modified form of the dual to obtain a lower bound that is used
along with an upper bound from the objective function to stop the
algorithm when the duality gap is provably small. Our lower bound is
a provably lower bound for the original problem as it is using a
feasible dual point to evaluate the dual value. This is needed, as an
approximate primal optimal value is not exactly optimal.

\end{enumerate}

\subsection{Future Plans}
\label{sect:future}

There are still many improvements that can be made. 
Exact primal feasibility was quickly obtained and maintained throughout 
the iterations. However, accurate dual feasibility was difficult to maintain.
This is most likely due to the rounding errors in the numerical computation of the Hessian $H_{f}(\rho)$ when $\rho$ is near the boundary of the positive semidefinite cone.
This approximation can be improved by including a quasi-Newton approach, 
as we have accurate gradient evaluations.
We maintain high accuracy results even in the cases where the Jacobian 
was not full rank at the optimum. This appears to be due to the 
special data structures and more theoretical analysis at the optimum can be done.

In this work, we considered a model with the linear constraints 
$\Gamma(\rho) = \gamma$ restricted to be equalities. While this model 
covers many interesting \QKD protocols, there are scenarios where inequality
constraints are needed, e.g., when using the flag-state squasher 
\cite{zhang2021security} to reduce the dimension. Moreover, there can be 
additional matrix inequality constraints when the dimension reduction method
\cite{upadhyaya2021dimension} is applied. It is interesting and
important to address possible numerical instabilities introduced by those 
inequality constraints as was done in this paper.

\section*{Acknowledgements}
\addcontentsline{toc}{section}{Acknowledgements}
We thank Kun Fang and Hamza Fawzi for discussions and Cunlu Zhou for kindly providing us the code presented in \cite{Faybusovich2020longstep}.  
The authors H.H., J.I. and H.W. thank the support of the National Sciences and Engineering Research Council (NSERC) of Canada. 
Part of this work was done at the Institute for Quantum Computing, University of Waterloo, which is supported by Innovation, Science and Economic Development Canada. 
J.L. and N.L. are supported by NSERC under the Discovery Grants Program, Grant No. 341495, and also under the Collaborative Research and Development Program, Grant No. CRDP J 522308-17. Financial support for this work has been partially provided by Huawei Technologies Canada Co., Ltd.

\section*{Code Availability}

\addcontentsline{toc}{section}{Code Availability}

The code developed in this work is currently available at this \href{https://www.math.uwaterloo.ca/%7Ehwolkowi/henry/reports/ZGNQKDmainsolverUSEDforPUBLCNJuly31}
{link}.\footnote{\url{https://www.math.uwaterloo.ca/~hwolkowi/henry/reports/ZGNQKDmainsolverUSEDforPUBLCNJuly31/}}
 It will be integrated into the open-source \QKD security software project, which can be accessed via the \href{https://openqkdsecurity.wordpress.com}{link}.\footnote{\url{https://openqkdsecurity.wordpress.com/}}

\newpage
\appendix

\section{Background Results and Proofs}
\label{app:pfstech}

\subsection{Adjoints for Matrix Multiplication}
\label{sect:adjMM}

Adjoints are essential for our interior point algorithm when using
matrix-free methods.
We define the \textdef{symmetrization linear map, $\sym$}, as
$\sym(M) = (M+M^\dagger )/2$.
The \textdef{skew-symmetrization linear map, $\sksym$}, is
$\sksym(M) = (M-M^\dagger )/2$.
\index{$\sym$, symmetrization linear map}
\index{$\sksym$, skew-symmetrization linear map}

\begin{lemma}[adjoint of $\cW(R) := WR$]
\label{lem:WRadj}
Let $W\in \Cnn$ be a given square complex matrix, 
and define the (left matrix
multiplication) linear map $\cW : \Cnn\to\Cnn$ by
$\cW(R)=WR$. Then the adjoint $\cW^\dagger : \Cnn \to \Cnn$ is defined by
\begin{equation}
\label{eq:generalWRadj}
\cW^\dagger (M) = 
  \Re(W)^\dagger \Re(M) +\imag(W)^\dagger \imag(M)
 + i\left( \Re(W)^\dagger \imag(M) -\imag(W)^\dagger \Re(M)\right).
\end{equation}
If $W\in \Hn$ and $\cW:\Hn \to \Cnn$, then the adjoint $\cW^\dagger :\Cnn\to
\Hn$ is defined by
\begin{equation}
\label{eq:adjWMHerm}
\cW^\dagger (M) = 
\sym\left[\Re(W)\Re(M) -\imag(W)\imag(M)\right] +
 i \sksym\left[\imag(W) \Re(M) + \Re(W) \imag(M)\right].
\end{equation}
\end{lemma}

\begin{proof}
We have
\[
\begin{array}{rcl}
WR 
&=& (\Re(W) +i\imag(W))(\Re(R) +i\imag(R)) \\
&=& \Re(W)\Re(R) - \imag(W)\imag(R) 
+i\Re(W)\imag(R) +i \imag(W)\Re(R) . \\
\end{array}
\]
Hence,
\[
\Re(WR) = \Re(W)\Re(R) - \imag(W)\imag(R), \quad
\imag(WR) = \Re(W)\imag(R) + \imag(W)\Re(R). 
\]
Then the inner product \cref{def:complex_innerprod} yields
\[
\begin{array}{rcl}
\<\cW(R),M \>  
&=&
\<WR,M \>  
\\&=&
 \< \Re(WR) ,\Re(M) \> + \<\imag(WR),\imag(M)\>.
\end{array}
\]
We first focus on the first term.
\[
\begin{array}{rcl}
\< \Re(WR) ,\Re(M) \>
&=& \trace( \Re(WR)^\dagger \Re(M)   )\\
&=& \trace\left( [ \Re(W)\Re(R) - \imag(W)\imag(R) ]^\dagger \Re(M) \right)\\
&=& \trace\left( [ \Re(W)\Re(R)]^\dagger \Re(M) \right)
    -\trace\left( [ \imag(W)\imag(R) ]^\dagger \Re(M) \right) \\
&=& \trace\left( \Re(R)^\dagger \Re(W)^\dagger \Re(M) \right)
    -\trace\left( \imag(R)^\dagger \imag(W)^\dagger \Re(M) \right) \\
&=& \left\langle \Re(R), \Re(W)^\dagger \Re(M) \right\rangle   
    -\left\langle \imag(R),\imag(W)^\dagger \Re(M) \right\rangle \\
\end{array}
\]
We now focus on the second term.
\[
\begin{array}{rcl}
\< \imag(WR) ,\imag(M) \>
&=& \trace( \imag(WR)^\dagger \imag(M)   )\\
&=& \trace\left( [ \Re(W)\imag(R) + \imag(W)\Re(R) ]^\dagger \imag(M) \right)\\
&=& \trace\left( [ \Re(W)\imag(R)]^\dagger \imag(M) \right)
    + \trace\left( [\imag(W)\Re(R) ]^\dagger \imag(M) \right) \\
&=& \trace\left( \imag(R)^\dagger \Re(W)^\dagger \imag(M) \right)
    +\trace\left( \Re(R)^\dagger\imag(W)^\dagger \imag(M) \right) \\
&=& \left\langle \imag(R), \Re(W)^\dagger \imag(M) \right\rangle   
    +\left\langle \Re(R),\imag(W)^\dagger \imag(M) \right\rangle \\
\end{array}
\]
Therefore,
\begin{equation}
\label{eq:IWRIM}
\begin{array}{rcl}
\<\cW(R),M\>  
&=&
 \left\langle \Re(R), \Re(W)^\dagger \Re(M) \right\rangle   
    -\left\langle \imag(R),\imag(W)^\dagger \Re(M) \right\rangle 
\\&& \quad + \left\langle \imag(R), \Re(W)^\dagger \imag(M) \right\rangle   
    +\left\langle \Re(R),\imag(W)^\dagger \imag(M) \right\rangle 
\\&=&
 \left\langle \Re(R), \Re(W)^\dagger \Re(M) +\imag(W)^\dagger \imag(M) \right\rangle 
\\&& \quad + \left\langle \imag(R), \Re(W)^\dagger \imag(M) 
                           -\imag(W)^\dagger \Re(M) \right\rangle 
\\&=&
 \langle R, \cW^\dagger (M)\rangle.
\end{array}
\end{equation}
This proves the first general adjoint expression \eqref{eq:generalWRadj}.

Now, suppose that $W$ Hermitian is given, and consider $\cW: \Hn\to
\Cnn$, i.e.,~a mapping from $\Hn$.
Then \cref{eq:IWRIM} becomes
\begin{equation}
\label{eq:IWRIM2}
\begin{array}{rcl}
\<\cW(R),M\>  
&=&
 \left\langle \Re(R), \Re(W)^\dagger \Re(M) +\imag(W)^\dagger \imag(M) \right\rangle 
\\&& \quad + \left\langle \imag(R), \Re(W)^\dagger \imag(M) 
                           -\imag(W)^\dagger \Re(M) \right\rangle 
\\&=&
 \left\langle \Re(R), \Re(W) \Re(M) -\imag(W) \imag(M) \right\rangle 
\\&& \quad + \left\langle \imag(R), \Re(W) \imag(M) 
                           +\imag(W) \Re(M) \right\rangle 
\\&=&
 \left\langle \Re(R), \sym(\Re(W) \Re(M) -\imag(W) \imag(M)) \right\rangle 
\\&& \quad + \left\langle \imag(R), \sksym(\Re(W) \imag(M) 
                           +\imag(W) \Re(M)) \right\rangle.
\end{array}
\end{equation}
This yields the second term in~\cref{eq:adjWMHerm}.
\end{proof}

\begin{lemma}[adjoint of $\rho(S)=S\rho$]
\label{lem:Srho}
Let $\rho\in \Cnn$ be a given square complex matrix, 
and define the (right matrix
multiplication) linear map $\rho : \Cnn\to\Cnn$ by
$\rho(S)=S \rho$. Then the adjoint $\rho^\dagger : \Cnn \to \Cnn$ is defined by
\begin{equation}
\label{eq:generalSrhoadj}
\rho^\dagger (M) = \Re(M)\Re(\rho)^\dagger +\imag(M) \imag(\rho)^\dagger 
+ i  \left( -\Re(M)\imag(\rho)^\dagger +  \imag(M) \Re(\rho)^\dagger  \right).
\end{equation}
If $\rho\in \Hn$ and $\rho:\Hn \to \Cnn$, then the adjoint $\rho^\dagger :\Cnn\to
\Hn$ is defined by
\begin{equation}
\label{eq:adjSrhoHerm}
\rho^\dagger (M) = \cS\left[ \Re(M)\Re(\rho) -\imag(M) \imag(\rho) \right]
+ i \sksym \left[ \Re(M)\imag(\rho) +  \imag(M) \Re(\rho)  \right] .
\end{equation}
\end{lemma}

\begin{proof}
The proof is similar to the proof of~\Cref{lem:WRadj}.
\end{proof}

\subsection{Proof of \Cref{lemma:rangeZrelation} }
\label{sec:proof_rangelem}

\begin{proof}
Recall that
\begin{equation}
\label{eq:ABpsdrange}
A,B\succeq 0 \implies \range(A+B) = \range(A) + \range(B).
\end{equation}
Let $X$ be a positive semidefinite matrix with rank $r$ and spectral
decomposition
\[
X = \sum_{i=1}^r \lambda_i u_i u_i^\dagger .
\]
We only focus on the first term $\lambda_1 u_1 u_1^\dagger$.
Then 
\[
\cZ(\lambda_1 u_1 u_1^\dagger) = 
\sum_{j=1}^n Z_j (\lambda_1 u_1 u_1^\dagger )Z_j = 
\sum_{j=1}^n \lambda_1 ( Z_j u_1) (Z_ju_1)^\dagger . 
\]
We note, from \cref{eq:ABpsdrange}, that  
\[
\begin{array}{rl}
\range( \cZ(\lambda_1 u_1 u_1^\dagger)  ) 
& = \range( \lambda_1 ( Z_1 u_1) (Z_1u_1)^\dagger
+ \lambda_1 ( Z_2 u_1) (Z_2u_1)^\dagger + \cdots
+\lambda_1 ( Z_n u_1) (Z_n u_1)^\dagger    ) \\ 
&= \range( Z_1 u_1 ) + \cdots + \range( Z_n u_1 ) .
\end{array}
\]
We also note that 
\[
u_1 = I u_1 = \left( \sum_{j=1}^n Z_j \right) u_1
=\sum_{j=1}^n Z_j  u_1 
\in \range(Z_1u_1) + \cdots + \range(Z_n u_1) .
\]
Hence, 
\[
\range(\lambda_1 u_1 u_1^\dagger) =
\range(u_1) 
\subseteq 
\range(Z_1u_1) + \cdots + \range(Z_n u_1) 
= \range( \cZ(\lambda_1 u_1 u_1^\dagger )  ) .
\]

We now consider the first two terms in $X$, $\lambda_1 u_1 u_1^\dagger + \lambda_2 u_2 u_2^\dagger $.
Similarly, 
\begin{equation}
\label{eq:rangeinclusions}
\range(\lambda_1 u_1 u_1^\dagger)
\subseteq 
\range( \cZ(\lambda_1 u_1 u_1^\dagger)  ) 
\ \ \text{ and }  \ \ 
\range(\lambda_2 u_2 u_2^\dagger)
\subseteq 
\range( \cZ(\lambda_2 u_2 u_2^\dagger)  ) .
\end{equation}
Then 
\[
\begin{array}{rll}
\range( \lambda_1 u_1 u_1^\dagger + \lambda_2 u_2 u_2^\dagger )
&= \range(\lambda_1 u_1 u_1^\dagger) + \range(\lambda_2 u_2 u_2^\dagger) & \text{by \cref{eq:ABpsdrange} }\\
& \subseteq  \range( \cZ(\lambda_1 u_1 u_1^\dagger)  ) + \range( \cZ(\lambda_2 u_2 u_2^\dagger)  )  & \text{by \cref{eq:rangeinclusions} } \\
& = \range( \cZ(\lambda_1 u_1 u_1^\dagger) + \cZ(\lambda_2 u_2 u_2^\dagger)  ) & \text{by \cref{eq:ABpsdrange} } \\
& = \range( \cZ(\lambda_1 u_1 u_1^\dagger + \lambda_2 u_2 u_2^\dagger ))  &  \text{by linearity of } \cZ . \\
\end{array}
\]
This completes the proof (The induction steps are clear.).
\end{proof}

\subsection{Derivatives for Quantum Relative Entropy under Positive Definite Assumptions}
\label{sect:derivsquant}

We can reformulate the quantum relative entropy function defined in the key rate optimization \cref{eq:keyrateopt} as
\begin{equation}
\label{eq:f_in_trace}
\begin{array}{rcl}
f(\rho) 
&=&
D(\cG(\rho)\|\cZ(\cG(\rho)))
\\
&=&
\Tr\left(\cG(\rho) \log \cG(\rho)\right) -
\Tr(\cG(\rho) \log
\cZ(\cG(\rho)))\\
&=&
\Tr\left(\cG(\rho) \log \cG(\rho)\right) -
\Tr(\cZ(\cG(\rho)) \log
\cZ(\cG(\rho)))\\
\end{array}
\end{equation}
Here, the linear map $\cZ$ is added to the second term in~\cref{eq:f_in_trace} above, and the equality follows from~\Cref{prop:Zopprops}.

In this section, we review the gradient (Fr\'echet derivative), and the image of the Hessian, for the reformulated relative entropy function $f$ defined in~\cref{eq:f_in_trace}. We obtain the derivatives of $f$ under the
assumption that the matrix-log is acting on positive
definite matrices. This assumption is needed for differentiability.
Note that the difficulty arising from the singularity is handled by
using perturbations in~\cite{winick2017reliable,Faybusovich2020longstep}. 
This emphasizes the need for the regularization below as otherwise $f$ 
in~\cref{eq:f_in_trace} is \emph{never} differentiable. 
We avoid using perturbations in this paper by applying \FR in the  sections below. 

We now use the chain rule and
derive the first and the second order derivatives of the composition
of a linear and entropy function.
\begin{lemma}
	\label{lem:gradfromchainrule}
	Let $\cH:\H^n\to \H^k$ be a linear map that
	preserves positive semidefiniteness. Assume that $\cH(\rho) \in  \H_{++}^k$. Define the composite function $g:\H^k_+ \to \R$ by
	\[
	g(\rho) = \trace \left(\cH (\rho)\log(\cH(\rho)) \right).
	\]
	Then the gradient of $g$ at $\rho$ is 
	\[
	\nabla g(\rho) =   \cH^\dagger  (\log[\cH(\rho)]) +  \cH^\dagger  (I),
	\]
	and the Hessian of $g$ at $\rho$ acting on $\Delta \rho$ is
	\[
	\nabla^2 g(\rho)(\Delta \rho) =\cH^\dagger  \left( 
	\log^\prime\cH(\rho)(\cH(\Delta \rho))
	\right), 
	\]
where $\log^\prime$ denotes the  Fr\'echet  derivative.
\end{lemma}

\begin{proof}
	We first work on the first-order derivative.
	\begin{equation}
	\label{eq:gradient_for_1sttermnew}
	\begin{aligned}
	\langle \nabla g(\rho),\Delta \rho \rangle
	&= 
	\Big\langle 
	\frac{d}{d \rho}
	\trace \left( \cH(\rho)  \log(\cH(\rho) )
	\right),\Delta \rho \Big\rangle \\
	&= 
	\trace \Big( \frac{d}{d \rho}
	\left( \cH(\rho)  \log(\cH(\rho) )
	\right)(\Delta \rho) \Big) \\
	&= 
	\trace \left(
	\frac{d}{d \rho}
	\Big(\cH(\rho)\Big) (\Delta \rho)
	\log(\cH(\rho)) +  \cH(\rho)
	\frac{d}{d \rho}\Big(\log(\cH(\rho))\Big)(\Delta \rho) \right)\\
	&= 
	\Big\langle 
	\frac{d}{d \rho}
	\Big(\cH(\rho)\Big) \Delta \rho
	,  \log(\cH(\rho))\Big\rangle + \Big\langle \cH(\rho),
	\frac{d}{d \rho}\Big(\log(\cH(\rho))\Big)\Delta \rho \Big\rangle \\
	&= \left\langle  \Delta \rho, \cH^\dagger 
	\Big(\log(\cH(\rho))\Big) \right\rangle + 
	\left\langle 
	\left(\frac{d}{d \rho}\log(\cH(\rho))\right)^\dagger 
	\cH(\rho),
	\Delta \rho\right\rangle\\
	&=\left\langle  \Delta \rho, \cH^\dagger  \Big(\log[\cH(\rho)]\Big)
	\right\rangle + \left\langle  \frac{d \cH(\rho)}{d \rho}^\dagger  (I), 
	\Delta \rho\right\rangle\\
	&=\left\langle  \cH^\dagger  \Big(\log[\cH(\rho)]\Big), \Delta \rho \right\rangle + \left\langle  \cH^\dagger  (I), 
	\Delta \rho\right\rangle. \\
	\end{aligned}
	\end{equation}
	Note that we used the fact that the directional derivative of matrix-log at $\rho$ in the direction  $\rho$ is:
	\[
	\log^\prime(\delta)(\delta)=\log^\prime(\delta;\delta) = I.
	\]
	Similarly, the Hessian $g$ at $\rho$ acting on $\Delta \rho$ can be obtained as follows.
	\begin{equation}
	\label{eq:hessianqkdnew}
	\begin{array}{c}
	\nabla^2 g(\rho)(\Delta \rho) 
	=
	\frac {\partial}{\partial \rho}
	\cH^\dagger  \Big( 
	[\log\cH(\rho)]  \Big)
	=
	\cH^\dagger 
	\frac {\partial}{\partial \rho}
	\Big( 
	[\log\cH(\rho)]  \Big) 
	=
	\cH^\dagger  \Big( 
	[\log^\prime\cH(\rho)(\cH\Delta \rho)]
	\Big).
	\end{array}
	\end{equation}
	
\end{proof}
	\index{$\log^\prime$, Fr\'echet  derivative of $\log$}
	\index{Fr\'echet  derivative of $\log$, $\log^\prime$}

Under the assumption that $\cG(\rho) \succ 0$, 
we can use~\Cref{lemma:rangeZrelation} and show
that $\cZ(\cG(\rho)) \succ 0$.
Using \Cref{lem:gradfromchainrule} and \eqref{eq:deltaZdeltalog}, 
we obtain the first and the 
second order derivatives of the objective function $f$ in \eqref{eq:f_in_trace}.

\begin{corollary}
	\label{thm:gradf} 
	Suppose that $\rho\in \Hnp$ and $\cG(\rho) \succ 0$. 
	Then the \textdef{gradient of $f$} at $\rho$ is
	\begin{equation}
	\label{eq:gradient}   
	\nabla f(\rho) =
	\cG^\dagger  \Big( \log[\cG(\rho)]\Big) -
	(\cZ\circ \cG)^\dagger  \Big( \log[(\cZ \circ \cG)(\rho)] \Big).
	\end{equation}
	The
	\textdef{Hessian at $\rho\in \Hn_{+}$ acting on the direction 
		$\Delta \rho \in \Hn$} is
	\begin{equation}
	\label{eq:hessianqkd}
	\begin{array}{rcl}
	\nabla^2 f(\rho)(\Delta \rho) 
	&=&
	\cG^\dagger  \Big( 
	[\log^\prime\cG(\rho)(\cG\Delta \rho)]\Big)
	- 
	(\cZ\circ\cG)^\dagger  \Big([\log^\prime(\cZ \circ
	\cG)(\rho)((\cZ\circ\cG)(\Delta \rho))\Big).
	\end{array}
	\end{equation}
\end{corollary}

\subsection{Proof of \Cref{thm:FRonObserv}}
\label{sec:proofFRonObserv}

\begin{proof}
We provide an alternative, self-contained proof. We note that
the key is finding an \textdef{exposing vector} for $S_R$,
i.e.,~$Z_\Gamma\succeq 0$ such that $\langle Z_\Gamma,\rho \rangle = 0$, $\forall \rho \in S_R$.
See e.g.,~\cite{DrusWolk:16}. The standard theorem of the alternative 
for strict feasibility, \Cref{lem:FRfarkas},
yields the following equalities for $Z_\Gamma$:
\[
0 \neq Z_\Gamma = 
\sum_j y_{j} \Gamma_{j} 
=\sum_j y_j \left(\Theta_j \otimes \one_B \right) 
=\left(\sum_j y_j \Theta_j\right) \otimes \one_B  
\succeq 0; \quad 
y^\dagger\theta=0.
\]
It is equivalent to look at the smaller problem and find $y$ so that
\[
0 \neq Z_\Theta = 
\sum_j y_j \Theta_j
\succeq 0; \quad 
y^\dagger\theta=0.
\]
Since the reduced density operator constraint requires that
$\theta_{j} = \Tr \rho_A \Theta_j$, we get
\[
0=\sum_j y_{j}\theta{j} = \Tr \left( \rho_A \sum_j y_{j}\Theta_j \right) \iff
\rho_A \left(\sum_j y_{j}\Theta_j \right) =
\rho_A Z_\Theta = 0,
\]
i.e.,~the exposing vector $Z_\Theta = QR_\Theta Q^\dagger$, for some $R_\Theta$.
Conversely, we can set $Z_\Theta = QQ^\dagger, R_\Theta=I$, by the 
basis property of the
$\Theta_i$, i.e.,~the basis property means we can always find an
appropriate $y$ so that 
$\sum_j y_{j}\Theta_j  = QQ^\dagger$. We get that $\rank Z_\Theta= n_A-r$.
Therefore, $Z_\Gamma = Z_\Theta \otimes \one_B$, with $\rank Z_\Theta =n_B(n_A-r)$, is an exposing vector as desired, i.e.,~we have
\[
S_R \subset \left\{\rho \succeq 0 \,:\, 
\langle Z_\Theta\otimes \one_B,\rho \rangle = 0 \right\}.
\]
Thus we get the conclusion that $\rho = VRV^\dagger$, as
desired.
\end{proof}

\subsection{Proof of \Cref{thm:optcondfinalprob}}
\label{sec:proofoptcondfinalprob}
\begin{proof}
The dual in \Cref{item:strduality} is obtained from from standard min-max 
argument; See~\cite[Chapter 5]{MR2061575}.
\[
\begin{array}{rcl} 
d^*  = \max\limits_y \min\limits_{\rho\in \Hp^{n_\rho}} L(\rho,y)
&=& \max\limits_y  \left\{ L(\rho,y) : Z \in 
\partial f(\rho) + \Gamma_V^\dagger (y) , \, Z \in (\H_+^{n_\rho}-\rho)^\dagger \right\}
\\&=& \max\limits_{y,Z\succeq 0}  
        \min\limits_{\rho\in \H_+^{n_\rho}} L(\rho,y)-\langle Z,\rho\rangle.
\end{array}
\]
That strong duality holds comes from our regularization process, i.e., 
the existence of a Slater point; see~\cite[Chapter 8]{Lu:69}.

\Cref{item:optcharac} is the standard optimality conditions for convex
	programming, where the dual feasibility $0\in \partial f(\rho) +
	\Gamma_V^\dagger (y)  -Z $ holds from \Cref{cor:fsubdiff}.
\end{proof}

\section{Implementation Details}
\label{app:compasp}

In this section we look at simplifications for
evaluations of the objective function and its derivatives.

\subsection{Matrix Representations of Derivatives}
\label{sec:matrix_repre}

We now include a matrix representation for the derivatives.

Let $A,B,C$ be given compatible matrices. If $X$ is Hermitian, 
then the linear system $AXB=C$ can be written as 
$$ 
\left((B^\dagger )^T \otimes A\right) \text{vec}(X) = \text{vec}(C).
$$
Note that $M^T$ is the transpose of $M$, i.e.,~without conjugation.

Let $g: \R \rightarrow \R$ be a continuously differentiable function. The first divided difference $h^{[1]}(\lambda,\mu)$ of $g$ at $\lambda,\mu \in \R$  is defined as
\begin{equation}\label{fdd}
		h^{[1]}(\lambda,\mu) = \begin{cases}
\frac{g(\lambda)-g(\mu)}{\lambda-\mu}  &\text{ if } \lambda \neq \mu \\
g'(\lambda) &\text{ if } \lambda = \mu \; . \\
		\end{cases}
\end{equation}
If $D$ is a diagonal matrix with diagonal entries $\lambda_{1},\ldots,\lambda_{n}$, then we define $h^{[1]}(D)$ to be the symmetric $n \times n$ matrix given by $h^{[1]}(\diag(D))$.

\begin{lemma}
Let $\A : \H^s\to \H^t$ be a linear map,
$\rho, \Delta \rho \in \H^s, \A (\rho) \in \H_{++}^t$, and 
$f(\rho ) = \tr ( \A(\rho )\log \A(\rho ) )$.
Let $\A(\rho ) = UDU^\dagger$ be the spectral decomposition of $f$ at $\rho$, and  the Hessian of $f$ at $\rho$ in the 
direction $\Delta \rho$ are given by 
$$
\nabla f(\rho ) = \A^\dagger  ( \log \A(\rho) ) + \A^\dagger (I),
$$
 and 
$$
\left(H_{f}(\rho )\right)(\triangle \rho) = \A^\dagger \left(U (h^{[1]}(D)\circ U^\dagger \A(\triangle \rho)U) U^\dagger \right),
$$
where $h^{[1]}(D)$ is the first divided difference of the logarithm function $g(x) = \ln x$, see \cref{fdd} and the paragraph below.
\end{lemma}

In the actual computation, it is more convenient to express the gradient and Hessian in matrix form. Let $A$ be the matrix representation of $\A$. The Hessian in matrix form is
$$\begin{array}{rrl}
	H_{f}(\rho )&=& 
	A^\dagger (U^\dagger\otimes U)\Diag(h^{[1]}(D))((U^\dagger )^\dagger\otimes U^\dagger )A.
\end{array}$$

\subsection{Matrix Representation of the Second Projected Gauss-Newton System }
\label{app:matrix_repre_sys}

We present the matrix representation of  \cref{eq:proj2GNeqn}.
Let $N_i$ be a basis element of $\nul(\Gamma_V)$.
Then $\cN^\dagger (w)$ has the representation $\sum_i w_i{ N_i } $.
Then the LHS of \cref{eq:proj2GNeqn} becomes 
\begin{equation}
\begin{array}{rl}
&\left[Z  \cN^\dagger (\Delta v) +\nabla^2 f(\rho) 
[ \cN^\dagger (\Delta v) ] \ \rho \right]  + \left[ \Gamma_V^\dagger (\Delta y) \rho \right] \\
=& 
Z \sum_i  N_i \Delta v_i
+\nabla^2 f(\rho) [ \sum_i  N_i \Delta v_i ] \ \rho  +  \Gamma_V^\dagger (\Delta y) \rho  \\
=& 
\sum_i Z N_i \Delta v_i
+ \sum_i  \nabla^2 f(\rho) N_i \rho \Delta v_i + \sum_i \Gamma_i \rho  \Delta y_i  .
\end{array} 
\end{equation}
Applying $\Cvec$\footnote{The operator \text{vec} maps a real matrix to a column vector by stacking the columns on top of one another. $\Cvec$ is a generalization of \text{vec} for complex matrices. It maps a complex matrix $M$ to the column vector $\begin{bmatrix}
		\text{vec}(\real(M)) \\
		\text{vec}(\im(M))
\end{bmatrix}$.} to the terms related to $\Delta v$, we have the following matrix representation:
\[
\begin{bmatrix}
\Cvec\left( ZN_1 +  \nabla^2 f(\rho) N_1 \rho \right)  & \cdots 
& \Cvec\left( ZN_{n_\rho^2-\mV} +  \nabla^2 f(\rho) N_{n_\rho^2-\mV} \rho \right)
\end{bmatrix}
\begin{bmatrix}
\Delta v_1  \\ \vdots \\ \Delta v_{n_\rho^2-\mV}
\end{bmatrix} .
\]
Similarly, applying $\Cvec$ on the terms related to $\Delta y$, we have the following matrix representation:
\[
\begin{bmatrix}
\Cvec \left( \Gamma_1 \rho \right) & \cdots & \Cvec \left( \Gamma_m \rho \right)
\end{bmatrix}
\begin{bmatrix}
\Delta y_1  \\ \vdots \\ \Delta y_{m} 
\end{bmatrix}.
\]
The RHS of \cref{eq:proj2GNeqn} becomes  
$\Cvec \left(  F_\mu^c + Z F_\mu^{p} + \left( F_\mu^d +\nabla^2 f(\rho)F_\mu^{p}  \right) \rho \right) $.
Thus, $\dGN$ is obtained be solving the system
\[
\begin{array}{rl}
&
\begin{bmatrix}
\left[
\Cvec\left( ZN_i +  \nabla^2 f(\rho) N_i \rho \right)
\right]_{i=1,\ldots, n_\rho^2-\mV}
& 
\left[
\Cvec \left( \Gamma_j \rho \right) 
\right]_{j=1,\ldots, \mV}
\end{bmatrix}
\begin{bmatrix}
\Delta v \\ \Delta y
\end{bmatrix}
\\
=&\Cvec \left(  F_\mu^c + Z F_\mu^{p} + \left( F_\mu^d +\nabla^2 f(\rho)F_\mu^{p}  \right) \rho \right) .
\end{array} 
\]

\subsection{Implementation Heuristics}
\label{sect:implheur}
We now discuss the implementation details. This involves preprocessing
for a nullspace representation and preconditioning. The details follow.

\subsubsection{Stopping Criteria}

We terminate the algorithm when the optimality condition 
\eqref{eq:optcondpert} is approximately satisfied. 
Denote the residual in~\Cref{thm:2ndprojGN} by
\[
\text{RHS} = - F_\mu^c -Z F_\mu^{p} - \left( F_\mu^d +\nabla^2 f(\rho)F_\mu^{p}  \right) \rho,
\]
and the denominator term by
\[
  \text{denom} =
       1 + \frac 12\min\left\{ \  \|\rho\|+\|Z\| \  ,  \
       |\text{bestub}| + | \text{bestlb}| \ \right\}.
\]
Then
\begin{equation}
\label{def:stopgap}
       \text{relstopgap} = \frac 1{\text{denom}}
	\max \left\{ \ \text{bestub} - 	\text{bestlb} \ , \
\| \text{RHS} \| \ \right\}.
\end{equation}
In other words, for a pre-defined tolerance $\epsilon$, we terminate the
algorithm when the relstopgap$ < \epsilon$.
If the algorithm computes lower and upper bounds of the optimal value throughout its execution, we may terminate the algorithm when the gap between lower and upper bounds is within $\epsilon$.
Finally, a common way to terminate an algorithm is to impose restrictions on the running time, e.g., setting an upper bound on the number of iterations or the physical running time. 

\subsubsection{\GN Direction using Sparse Nullspace Representation}
\label{sect:sparsenullrep}

We let $r = \Hvec (\rho) $, and construct
a matrix representation $H$ for the Hessian, and a matrix representation
$M$ for the linear constraints that
includes a permutation of rows and columns $rp,cp$ with inverse column
permutation $icp$, so that 
\[
r=\Hvec(\rho): \quad \textdef{$r(cp) = P_{cp}r$},\,\textdef{$r =
P_{icp}r(cp)$},\quad P_{cp}P_{icp}=P_{icp}P_{cp}=I,\,
P_{icp}=P_{cp}^\dagger.
\]
We can ignore the row permutations. We have
\[
\begin{array}{rcl}
\GV (\rho )
&=&
(\GV \HMat) \Hvec(\rho) 
\\&=&
(\GV \HMat) P_{icp}P_{cp}\Hvec(\rho) 
\\&=&
(P_{cp}(\GV \HMat)^\dagger)^\dagger P_{cp}\Hvec(\rho) 
\\&=&
 M r(cp)\\
& =& M P_{cp}\Hvec(\rho).
\end{array}
\]
We now get the nullspace representation:
\[
\hat r = \Hvec (\hat \rho) ;\,
\GV ( \hat \rho ) = M \hat r(cp) = \gV, \, M = \begin{bmatrix} B & E
\end{bmatrix}, \,
N^\dagger  = \begin{bmatrix}B^{-1} E \cr -I\end{bmatrix};
\]
\begin{equation}
\label{eq:nullspacerep}
r=\Hvec (\rho) :\quad 
\GV( \rho) = \gV \iff
\MGV P_{cp} r = \gV \, \iff\,  r = \hat r + 
P_{icp} N^\dagger ( w ), \text{ for some } w.
\end{equation}
The permutation of rows and columns are done in order to obtain a simple,
near triangular, well conditioned $B$ so that $B^{-1}E$ can be done
simply and maintain sparsity if possible.
The permutation of the rows does not affect the problem and we can
ignore it. However the permutation of the columns cannot be ignored.
We get the following
\[
\cN^\dagger (v) =\HMat \left(P_{icp} N^\dagger (v)\right) , \, \Gamma_V^\dagger (\Delta y) = P_{icp}M (\Delta y) , \,
\nabla^2f(\rho)\cN^\dagger  (\Delta v) =  \HMat \left(HP_{icp}N^\dagger (\Delta v) \right).
\]
By abuse of notation, the Gauss-Newton direction $\dGN\in \R^{n_\rho^2}$ 
can now be found from:
\begin{equation}
\label{eq:2ndprojelim}
\begin{array}{rcl}
F_\mu^{c\prime}\dGN 
&=&
 Z\HMat \left(P_{icp}N^\dagger (\Delta v)\right) + 
      \left(\HMat \left(\nabla^2f(\rho)(P_{icp}N^\dagger (\Delta v))\right)+\Gamma_V^\dagger (\Delta y) \right) \rho
\\&=&
\begin{bmatrix}
\cMZ\left(\HMat \left(P_{icp}N^\dagger \bigcdot\right) \right)    + \cMrho\left(\HMat \left(\nabla^2 f(\rho)
	P_{icp}N^\dagger \bigcdot\right)  \right) &
\cMrho \Gamma_V^\dagger  \bigcdot 
\end{bmatrix}
\begin{pmatrix}
\Delta v \cr \Delta y
\end{pmatrix}
\\&=&
-(F^c_\mu +F^d_\mu \rho + ZF^p_\mu) -  
\left(\nabla^2 f(\rho)(F_\mu^{p})\right) \rho.
\end{array}
\end{equation}

\subsubsection{Preconditioning}
\label{sect:precond}

The overdetermined linear system in~\cref{eq:2ndprojelim} can be
ill-conditioned. We use diagonal preconditioning, i.e.,~we let
$d_i=\|F_\mu^{c\prime}(e_i)\|$, for unit vectors $e_i$ and then
column precondition using 
\[
F_\mu^{c\prime} \leftarrow F_\mu^{c\prime} \Diag(d)^{-1}.\footnote{The
MATLAB command is:
$d_{GN} = ((F_\mu^{c\prime}/ \Diag(d))\backslash \text{RHS})./d$.}
\]
This diagonal preconditioning has been shown to be the optimal diagonal
preconditioning for the so-called $\Omega$-condition
number,~\cite{DeWo:90}. It performs exceptionally well in our tests
below.

\subsubsection{Step Lengths}
The \GN method is based on a linearization that suggests a step length of
one. However, long step methods are known to be more efficient in practice for
interior point methods for linear \SDPp s. Typically step lengths are
found using backtracking to ensure primal-dual
positive definiteness of $\rho, Z$.

In our case we do not have a
linear objective and we typically experience Maratos type situations,
i.e.,~we get fast convergence for primal feasibility but slow and no
convergence for dual feasibility. However, we do have the gradient and
Hessian of the objective function and therefore can minimize the
quadratic model for the objective function in the search direction
$\Delta \rho$
\[
\min_\alpha f(\rho) + \alpha \langle \nabla f(\rho),\Delta \rho \rangle
+ \frac 12 \alpha^2 \langle \Delta \rho \nabla^2 f(\rho),\Delta \rho
\rangle, \quad \alpha^*  = -\langle \nabla f(\rho),\Delta \rho\rangle/
 \langle \Delta \rho , \nabla^2 f(\rho)\Delta \rho\rangle.
\]
Therefore, we begin the backtracking with this step.

Moreover, we take a step length of one as soon as possible, and only
after this do we allow step lengths larger than one. This means
that exact primal feasibility holds for all further steps. This happens
relatively early for our numerical tests.

\section{Descriptions and Further Numerics of the Protocols}
\label{sect:testexamples}

We briefly describe six \QKD protocols where
we compare our algorithm to other algorithms. 
We also describe how the data ($\gamma$ in
\cref{eq:keyrateopt}) is generated. In addition,
we remark on the level of numerical difficulty for each example. We
consider four variants of the Bennett-Brassard 1984 (BB84) protocol
\cite{Bennett1984} including single-photon based variants:
entanglement-based (\Cref{sect:example1}), prepare-and-measure
(\Cref{sect:example2}), measurement-device-independent~\cite{Lo2012}
(\Cref{sect:example3}) and a coherent-state based variant with discrete
global phase randomization~\cite{cao2015discrete}
(\Cref{sect:example6}).  We also consider the single-photon version of
the twin-field \QKD~\cite{Lucamarini2018} (\Cref{sect:example4}).
Another interesting protocol in our numerical tests is the
quadrature phase-shift keying scheme of discrete-modulated
continuous-variable \QKD with heterodyne detection
(\Cref{sect:example5}), see
\cite[Protocol 2]{Lin2019}. These protocols correspond to numerical problems with
the level of difficulty ranging from easy to difficult. In the
descriptions below, we use the Dirac notation for quantum states which
are vectors in the underlying Hilbert space. We skip the description
about some common classical postprocessing steps in a \QKD protocol like
error correction and privacy amplification since they are unimportant
for our discussions here. We note that the description of linear maps
$\cG$ and $\cZ$ directly follow from the protocol description by
following the simplification procedure explained in~\cite[Appendix
A]{Lin2019}. We omit those detailed descriptions here and note that the explicit expressions for some of protocols can also be found in~\cite[Appendix D]{george2020finite}.

\subsection{Entanglement-Based BB84}\label{sect:example1}

We consider this protocol with a single-photon source and restrict our
discussions to the qubit space. In the quantum communication phase,
Alice and Bob each receive one half of a bipartite state. This is
supposed to be a two-qubit maximally entangled state before Eve's
tampering. And the measure in the $Z$ basis is with probability $p_z$,
or in $X$ basis with a probability $1-p_z$. In the classical
communication phase, they announce their basis choices for each round
and perform sifting to keep those rounds where they both chose the same
basis. In the end, they generate keys from both $Z$ and $X$ bases.

In the simulation, we assume both bases have the same error rate $e_z =
e_x = Q$. In particular, $\Gamma$ (in \cref{eq:keyrateopt}) contains the
$Z$-basis error rate, $X$-basis error rate constraints as well as one
coarse-grained constraint for each mismatched basis choice scenario.
This is to
ensure that Alice and Bob get completely uncorrelated outcomes in that
case. In other words, the data $\gamma$ are determined by $Q$. This test
example is supposed to be numerically easy, since it involves the smallest
possible size of $\rho$ for \QKDp, i.e., four. Moreover, there is no reduced density operator $\rho_A$ constraint for this example. In Table \ref{table:FINAL}, instances of this test example are labeled as ebBB84($p_z, Q$) for different values of $p_z$ and $Q$.

\subsection{Prepare-and-Measure BB84}\label{sect:example2}

Another protocol example in our numerical tests is the
prepare-and-measure version of BB84 with a single-photon source. In the
quantum communication phase, Alice chooses the $Z$ basis with a probability
$p_z$ or the $X$ basis with a probability $1-p_z$. When she chooses the 
$Z$ basis,
she sends either $\ket{0}$ or $\ket{1}$ at random, where $\ket{0}$ and
$\ket{1}$ are eigenstates of the Pauli $\sigma_Z$ operator. When she chooses
the $X$ basis, she sends either $\ket{+}$ or $\ket{-}$ at random, where
$\ket{\pm} = \frac{1}{\sqrt{2}} (\ket{0}\pm \ket{1})$. After Alice sends
the state of her choice to Bob, Bob chooses to measure in the $Z$ basis with
a probability $p_z$ or the $X$ basis with a probability $1-p_z$. The
rest of the protocol is exactly the same as the entanglement-based BB84
protocol described in~\Cref{sect:example1}. We call
$\{\ket{0},\ket{1},\ket{+},\ket{-}\}$ as stated in BB84.

For the security analysis, we use the source-replacement
scheme~\cite{Ferenczi2012} to convert it to its equivalent
entanglement-based scheme. Therefore, the main differences between this
example and the one in \Cref{sect:example1} are: (1) the dimension of
Alice's system for this example is four due to the source-replacement
scheme, while it is two for the entanglement-based BB84; (2) there is
the reduced density operator constraint $\rho_A$ which is of size $4$
and translated to $16$ linear constraints. In this test example, the
size of $\rho$ is $8$ and the size of $\cG(\rho)$ is $32$ before \FRp. 

The data simulation is done in a similar way as that in the
entanglement-based BB84 protocol, i.e.,~$e_x = e_z = Q$. In \Cref{table:FINAL}, instances of this test example are labeled as pmBB84($p_z, Q$).

\subsection{Measurement-Device-Independent BB84}\label{sect:example3}

In the measurement-device-independent variant of BB84 with single-photon
sources, Alice and Bob each prepare one of four BB84 states (with the
probability of choosing the $Z$ basis as $p_z$). Then they both
send this to an untrusted
third-party Charlie for measurements. He ideally then performs the
Bell-state measurements and announces the outcomes. We
consider a setup where Charlie only uses linear optics, and thus can only
measure two out of four Bell states. In this protocol, Charlie announces
either a successful Bell-state measurement or a failure. If a successful
measurement, Charlie then also announces the Bell state. Therefore, 
there are three possible announcement outcomes. 
After the announcement, Alice and Bob perform the basis sifting, as well
as discard rounds that are linked to unsuccessful events. They then
generate keys from rounds where they both chose the $Z$ basis and Charlie's announcement is one of the successful events. 

We now consider the measurement-device-independent type of protocols. 
As described in~\cite{coles2016numerical}, the optimization variable 
$\rho$ involves three parties as $\rho_{ABC}$. Here, registers $AB$
together serve the role of $A$ in the reduced density operator constraint
set~\cref{eq:constrreduced}. The dimension of Alice's system is $4$ and
so is Bob's dimension. The register $C$ is a classical register that
stores the announcement outcome. Thus it is three-dimensional with three
possible announcement outcomes. In the data simulation, we assume that
each qubit sent to Charlie goes through a depolarizing channel, with the
depolarizing probability $p$.  

In the numerical tests, we label instances of this protocol example as
mdiBB84($p_z, p$). The size of $\rho$ is $48$ and that of $\cG(\rho)$ is
$96$ before \FRp.

\subsection{Twin-Field \QKDp}\label{sect:example4}

As above, this protocol also uses the measurement-device-independent 
setup. The
exact protocol description can be found in~\cite[Protocol 1]{Curty2019}.
In this protocol, Alice and Bob each prepare a state $\ket{\phi_q}_{Aa}
= \sqrt{q}\ket{0}_A\ket{0}_a + \sqrt{1-q}\ket{1}_A\ket{1}_a$
($\ket{\phi_q}_{Bb}$) with $0 \leq q \leq 1$, where the register $A$ is a
qubit system and the register $a$ is an optical mode with the vacuum
state $\ket{0}_a$ and the single-photon state $\ket{1}_a$. After they
send states to the intermediate station, Charlie at the intermediate
station is supposed to perform the single-photon interference of these
two signal pulses and then announces the measurement outcome for each of
two detectors: click or no-click. Then Alice and Bob each perform
the $X$-basis measurement on their local qubits with a probability $p_x$ or
the $Z$-basis measurement with a probability $1-p_x$. They generate keys
from rounds where they both choose the $X$ basis and where Charlie announces a successful measurement outcome, that is, having exactly one of two detectors click.

In the simulation, we consider a lossy channel, with the transmittance
$10^{-0.02 L}$, for the distance $L$ in kilometers between Alice and
Bob. We consider the symmetric scenario where Charlie is at an equal
distance away from Alice and Bob. We also consider detector
imperfections: each detector at Charlie's side has detector efficiency
$\eta_d = 14.5\%$ and dark count probability $p_d = 10^{-8}$. 
In instances of this protocol, data is generated as a function of:
$q$ that appears in the states $\ket{\phi_q}_{Aa}$ and
$\ket{\phi_q}_{Bb}$; the total distance $L$ in kilometers between Alice
and Bob; and the probability of choosing $X$ basis $p_x$. 
The instances of this test example are labeled as TFQKD($q, L, p_x$).   

\subsection{Discrete-Modulated Continuous-Variable \QKDp}\label{sect:example5}

The exact protocol description can be found in~\cite[Protocol
2]{Lin2019}. We use the same simulation method described
in~\cite[Equation (30)]{Lin2019} to generate the data $\gamma$. In this
protocol, Alice sends Bob one of four coherent states $\ket{\alpha
e^{i \theta_j}}$, where $\theta_j = \frac{j \pi}{2}$ for $j=0,1,2,3$. And
Bob performs the heterodyne measurement, i.e.,~measuring both $X$-
and $P$-quadratures after splitting the signal into two halves by a
$50/50$ beamsplitter. The first and second moments of $X$- and
$P$-quadratures are used to constrain $\rho$. The data simulation uses a
phase-invariant Gaussian channel with transmittance $\eta_t$ and excess
noise $\xi$ to generate those values. We use the same photon-number
cutoff assumption used there to truncate the infinite-dimensional
Hilbert space. For the calculation, it is typically sufficient to choose
$N_c \geq 10$ to minimize the effects of errors due to the truncation. 
For simplicity, we assume the detector at Bob's side is an ideal
detector. The channel transmittance, $\eta_t$, is related to the
transmission distance $L$ between Alice and Bob, by $\eta_t = 10^{-0.02 L}$.

Let $N_c$ be an integer that represents the cutoff photon number. Before
\FRp, the sizes of $\rho$ and $\cG(\rho)$ are
 $4(N_c+1)$ and $16(N_c+1)$, respectively. 
In \Cref{table:FINAL}, we label instances of this example as DMCV($N_c, L, \xi, \alpha$). 

When the noise $\xi = 0$, this problem can be solved analytically via
physical arguments. The detailed instructions for analytical calculation can be found in~\cite[Appendix C]{Lin2019}. We use this special case to demonstrate that our interior-point method can reproduce the analytical results to high precision. 

\subsection{Discrete-Phase-Randomized BB84}\label{sect:example6}

We consider the phase-encoding BB84 protocol with $c$ (a parameter) 
discrete global phases evenly spaced between [0, $2\pi$]
\cite{cao2015discrete}. A detailed protocol description can also be
found in~\cite[Sec. IV D]{george2020finite}. In particular, each of four
BB84 states is realized by a two-mode coherent state $\ket{\alpha
e^{i\theta}}_r\ket{\alpha e^{i(\theta+\phi_A)}}_s$, where the first mode
is the phase reference mode and the second mode encodes the private
information. In particular, the $\theta$ is a global phase that involves
discrete phase randomization, i.e.,~$\theta \in \{\frac{2\pi \ell
}{c}: \ell=0,\dots, c-1\}$. The relative phase for encoding is $\phi_A
\in \{\frac{j\pi}{2}: j=0,1,2,3\}$, where $\{0,\pi\}$ correspond to the $Z$
basis and $\{\frac{\pi}{2},\frac{3\pi}{2}\}$ correspond to the $X$ basis.

Data simulation is done in exactly the same way as in~\cite[Section IV
D]{george2020finite}, where we consider detector imperfections and a lossy channel with a misalignment error due to phase drift. 

We remark that the instances of this test example become more
challenging as one increases the number of discrete phases $c$, since
the size of $\rho$ is $12c$ and the size of $\cG(\rho)$ is $48c$ before
\FRp. In all instances of this protocol, we choose $p_z = 0.5$ and the data are simulated with detector efficiency $\eta_d = 0.045$, dark count probability $p_d = 8.5 \times 10^{-7}$ and relative phase drift of $11^{\circ}$. The final key rate values are presented by taking the error correction efficiency as 1.16. In the numerical tests, we label instances of this protocol as dprBB84($c, \alpha, L$).

\subsection{Additional Numerical Results}
\label{sec:additionalreport}

\begin{table}[h]
\centering
\tiny{
\input{table/testallONE.tex}
}
\caption{Numerical Report for ebBB84 Instances}
\label{table:table1}
\end{table}

\begin{table}
\centering
\tiny{
\input{table/testallTWO.tex}
}
\caption{Numerical Report for pmBB84 Instances}
\label{table:table2}
\end{table}

\begin{table}
\centering
\tiny{
\input{table/testallTHREE.tex}
}
\caption{Numerical Report for mdiBB84 Instances}
\label{table:table3}
\end{table}

\begin{table}
\centering
\tiny{
\input{table/testallFOUR.tex}
}
\caption{Numerical Report for TFQKD Instances}
\label{table:table4}
\end{table}

\begin{table}
\centering
\tiny{
\input{table/testallFIVE.tex}
}
\caption{Numerical Report for DMCV Instances}
\label{table:table5}
\end{table}

\begin{table}
\centering
\tiny{
\input{table/testallSIX.tex}
}
\caption{Numerical Report for dprBB84 Instances}
\label{table:table6}
\end{table}

\FloatBarrier

	We remark on the behavior of the Frank-Wolfe method when used 
	without \FRp.  \Cref{table:FINAL} shows 
	that two instances of the TFQKD protocol have vastly different running times.
	Similar situations without \FR also occur in the tables
	presented here.
	The Frank-Wolfe method that is used and described in~\cite{winick2017reliable} adopts a two-step procedure.
	The running time of this method is mainly dependent
	on the number of iterations, since each iteration solves
	a linear \SDP problem using CVX. There are two stopping conditions:
	the first depends on the value of the
	gradient information; the second is the
	maximum number of iterations, currently set at $300$. 
	For the instances with
	short running time, termination typically occurred due to the
	gradient condition;
	the other instances stopped after reaching the maximum number of
	iterations. A likely reason for early termination is
	the perturbation approach used to preserve 
	positive definiteness of
	the matrices involved in the objective function. We emphasize that early 
	termination still leads to reliable (but pessimistically lower) key rates. More explanations can be found in ~\cite{winick2017reliable}. 
	As discussed in this paper, a general  motivation for \FR
	is to avoid numerical instability associated with such perturbation 
	approaches, i.e.,~without \FRp ,
	we typically have the erratic behavior due to ill-posedness of the problems.

\newpage

\cleardoublepage
\phantomsection
\addcontentsline{toc}{section}{Index}
\printindex

\label{ind:index}

\newpage

\cleardoublepage
\phantomsection
\addcontentsline{toc}{section}{Bibliography}
\bibliographystyle{unsrtnat}
\bibliography{QKD_report}

\end{document}

%% file: table/testall.tex
\begin{tabular}{|ccc||cc||cc||cc||cc|} \hline
\multicolumn{3}{|c||}{\textbf{Problem Data}} & \multicolumn{2}{|c||}{\textbf{Gauss-Newton}} & \multicolumn{2}{|c||}{\textbf{Frank-Wolfe with FR}} & \multicolumn{2}{|c||}{\textbf{Frank-Wolfe w/o FR}} & \multicolumn{2}{|c|}{\textbf{cvxquad with FR}}   \cr\hline
\multicolumn{1}{|c|}{\textbf{protocol}} & \multicolumn{1}{|c|}{\textbf{parameter}} & \multicolumn{1}{|c||}{\textbf{size}} & \multicolumn{1}{|c|}{\textbf{gap}} & \multicolumn{1}{|c||}{\textbf{time}} & \multicolumn{1}{|c|}{\textbf{gap}} & \multicolumn{1}{|c||}{\textbf{time}} & \multicolumn{1}{|c|}{\textbf{gap}} & \multicolumn{1}{|c||}{\textbf{time}} & \multicolumn{1}{|c|}{\textbf{gap}} & \multicolumn{1}{|c|}{\textbf{time}}   \cr\hline
ebBB84 & (0.50,0.05) & (4,16) &5.98e-13 & 0.40 & 1.01e-04 & 92.49 & 1.17e-04 & 93.05 & 5.46e-01 & 214.02  \cr 
ebBB84 & (0.90,0.07) & (4,16) &1.42e-12 & 0.20 & 2.71e-04 & 91.26 & 2.75e-04 & 94.49 & 7.39e-01 & 177.64  \cr 
pmBB84 & (0.50,0.05) & (8,32) &5.51e-13 & 0.23 & 1.12e-04 & 1.38 & 6.47e-04 & 1.91 & 5.26e-01 & 158.64  \cr 
pmBB84 & (0.90,0.07) & (8,32) &5.13e-13 & 0.17 & 7.31e-05 & 1.29 & 6.25e-04 & 38.65 & 6.84e-01 & 233.43  \cr 
mdiBB84 & (0.50,0.05) & (48,96) &1.14e-12 & 1.09 & 4.99e-05 & 104.31 & 5.22e-04 & 134.05 & 1.82e-01 & 557.08  \cr 
mdiBB84 & (0.90,0.07) & (48,96) &2.96e-13 & 0.96 & 2.04e-04 & 106.61 & 2.85e-03 & 126.62 & 4.57e-01 & 537.52  \cr 
TFQKD & (0.80,100.00,0.70) & (12,24) &1.15e-12 & 0.79 & 2.60e-09 & 1.21 & 1.57e-03 & 124.48 & n/a & 0.01  \cr 
TFQKD & (0.90,200.00,0.70) & (12,24) &1.04e-12 & 0.44 & 3.98e-09 & 1.13 & 1.68e-04 & 2.25 & n/a & 0.00  \cr 
DMCV & (10.00,60.00,0.05,0.35) & (44,176) &2.71e-09 & 507.83 & 4.35e-06 & 467.41 & 3.57e-06 & 657.08 & n/a & 0.01  \cr 
DMCV & (11.00,120.00,0.05,0.35) & (48,192) &3.24e-09 & 700.46 & 2.35e-06 & 194.62 & 2.15e-06 & 283.06 & n/a & 0.01  \cr 
dprBB84 & (1.00,0.08,30.00) & (12,48) &4.92e-13 & 1.19 & 3.85e-06 & 96.74 & 9.43e-05 & 141.38 & $\star\star$ & 118.81  \cr 
dprBB84 & (2.00,0.14,30.00) & (24,96) &1.04e-12 & 11.76 & 5.71e-06 & 17.66 & 5.38e-06 & 34.60 & $\star\star$ & 106.24  \cr 
dprBB84 & (3.00,0.10,30.00) & (36,144) &4.96e-13 & 63.26 & 6.48e-04 & 7.38 & 2.08e-02 & 29.00 & $\star\star$ & 582.64  \cr 
dprBB84 & (4.00,0.12,30.00) & (48,192) &3.80e-13 & 330.39 & 4.42e-05 & 13.78 & 9.79e-04 & 175.39 & $\star\star$ & 3303.23  \cr 
 \hline
\end{tabular}

%% file: table/testallONE.tex
\begin{tabular}{|ccc||cc||cc||cc||cc|} \hline
\multicolumn{3}{|c||}{\textbf{Problem Data}} & \multicolumn{2}{|c||}{\textbf{Gauss-Newton}} & \multicolumn{2}{|c||}{\textbf{Frank-Wolfe with FR}} & \multicolumn{2}{|c||}{\textbf{Frank-Wolfe w/o FR}} & \multicolumn{2}{|c|}{\textbf{cvxquad with FR}}   \cr\hline
\multicolumn{1}{|c|}{\textbf{protocol}} & \multicolumn{1}{|c|}{\textbf{parameter}} & \multicolumn{1}{|c||}{\textbf{size}} & \multicolumn{1}{|c|}{\textbf{gap}} & \multicolumn{1}{|c||}{\textbf{time}} & \multicolumn{1}{|c|}{\textbf{gap}} & \multicolumn{1}{|c||}{\textbf{time}} & \multicolumn{1}{|c|}{\textbf{gap}} & \multicolumn{1}{|c||}{\textbf{time}} & \multicolumn{1}{|c|}{\textbf{gap}} & \multicolumn{1}{|c|}{\textbf{time}}   \cr\hline
ebBB84 & (0.50,0.01) & (4,16) &1.14e-12 & 0.42 & 5.96e-05 & 95.32 & 5.88e-05 & 99.09 & 6.37e-01 & 216.75  \cr 
ebBB84 & (0.50,0.03) & (4,16) &8.35e-13 & 0.20 & 6.37e-05 & 93.35 & 6.24e-05 & 99.47 & 5.88e-01 & 258.98  \cr 
ebBB84 & (0.50,0.05) & (4,16) &5.98e-13 & 0.18 & 1.01e-04 & 95.31 & 1.17e-04 & 101.36 & 5.46e-01 & 213.04  \cr 
ebBB84 & (0.50,0.07) & (4,16) &1.05e-12 & 0.19 & 1.66e-04 & 96.92 & 1.65e-04 & 100.46 & 5.07e-01 & 179.38  \cr 
ebBB84 & (0.50,0.09) & (4,16) &1.35e-12 & 0.18 & 1.43e-04 & 96.35 & 2.55e-04 & 100.64 & 4.70e-01 & 170.14  \cr 
ebBB84 & (0.70,0.01) & (4,16) &1.21e-13 & 0.21 & 7.60e-05 & 96.33 & 7.62e-05 & 99.14 & 7.06e-01 & 172.09  \cr 
ebBB84 & (0.70,0.03) & (4,16) &6.36e-13 & 0.18 & 9.16e-05 & 96.77 & 9.15e-05 & 99.78 & 6.59e-01 & 160.88  \cr 
ebBB84 & (0.70,0.05) & (4,16) &5.34e-13 & 0.18 & 1.67e-04 & 97.03 & 1.03e-04 & 100.73 & 6.14e-01 & 173.77  \cr 
ebBB84 & (0.70,0.07) & (4,16) &1.26e-12 & 0.20 & 1.80e-04 & 96.22 & 1.74e-04 & 100.71 & 5.70e-01 & 261.77  \cr 
ebBB84 & (0.70,0.09) & (4,16) &2.06e-13 & 0.18 & 3.85e-04 & 97.08 & 2.61e-04 & 101.55 & 5.26e-01 & 225.94  \cr 
ebBB84 & (0.90,0.01) & (4,16) &5.40e-13 & 0.17 & 1.02e-04 & 97.68 & 1.03e-04 & 98.45 & 8.73e-01 & 141.93  \cr 
ebBB84 & (0.90,0.03) & (4,16) &7.05e-13 & 0.21 & 1.25e-04 & 97.97 & 1.43e-04 & 99.50 & 8.27e-01 & 164.26  \cr 
ebBB84 & (0.90,0.05) & (4,16) &3.48e-13 & 0.18 & 1.48e-04 & 96.91 & 9.81e-05 & 100.99 & 7.83e-01 & 186.87  \cr 
ebBB84 & (0.90,0.07) & (4,16) &1.42e-12 & 0.17 & 2.71e-04 & 96.55 & 2.75e-04 & 101.79 & 7.39e-01 & 179.55  \cr 
ebBB84 & (0.90,0.09) & (4,16) &1.09e-12 & 0.21 & 3.51e-04 & 96.40 & 3.42e-04 & 101.99 & 6.94e-01 & 228.63  \cr 
 \hline
\end{tabular}

%% file: table/testallTWO.tex
\begin{tabular}{|ccc||cc||cc||cc||cc|} \hline
\multicolumn{3}{|c||}{\textbf{Problem Data}} & \multicolumn{2}{|c||}{\textbf{Gauss-Newton}} & \multicolumn{2}{|c||}{\textbf{Frank-Wolfe with FR}} & \multicolumn{2}{|c||}{\textbf{Frank-Wolfe w/o FR}} & \multicolumn{2}{|c|}{\textbf{cvxquad with FR}}   \cr\hline
\multicolumn{1}{|c|}{\textbf{protocol}} & \multicolumn{1}{|c|}{\textbf{parameter}} & \multicolumn{1}{|c||}{\textbf{size}} & \multicolumn{1}{|c|}{\textbf{gap}} & \multicolumn{1}{|c||}{\textbf{time}} & \multicolumn{1}{|c|}{\textbf{gap}} & \multicolumn{1}{|c||}{\textbf{time}} & \multicolumn{1}{|c|}{\textbf{gap}} & \multicolumn{1}{|c||}{\textbf{time}} & \multicolumn{1}{|c|}{\textbf{gap}} & \multicolumn{1}{|c|}{\textbf{time}}   \cr\hline
pmBB84 & (0.50,0.01) & (8,32) &5.96e-13 & 0.38 & 6.19e-06 & 1.95 & 4.64e-04 & 25.18 & 6.30e-01 & 190.28  \cr 
pmBB84 & (0.50,0.03) & (8,32) &1.01e-12 & 0.18 & 1.71e-05 & 1.37 & 6.54e-04 & 90.80 & 5.74e-01 & 181.83  \cr 
pmBB84 & (0.50,0.05) & (8,32) &5.51e-13 & 0.19 & 1.12e-04 & 1.31 & 6.47e-04 & 1.94 & 5.26e-01 & 159.28  \cr 
pmBB84 & (0.50,0.07) & (8,32) &8.88e-14 & 0.16 & 5.89e-05 & 1.30 & 8.77e-04 & 1.98 & 4.81e-01 & 161.49  \cr 
pmBB84 & (0.50,0.09) & (8,32) &9.38e-13 & 0.19 & 6.71e-05 & 1.42 & 9.04e-04 & 2.03 & 4.40e-01 & 179.01  \cr 
pmBB84 & (0.70,0.01) & (8,32) &7.69e-13 & 0.25 & 7.62e-06 & 1.33 & 2.39e-04 & 130.73 & 7.03e-01 & 213.43  \cr 
pmBB84 & (0.70,0.03) & (8,32) &4.75e-13 & 0.18 & 2.38e-05 & 1.32 & 2.68e-04 & 133.52 & 6.51e-01 & 246.50  \cr 
pmBB84 & (0.70,0.05) & (8,32) &6.16e-13 & 0.17 & 3.52e-05 & 1.37 & 3.37e-04 & 134.59 & 6.04e-01 & 281.84  \cr 
pmBB84 & (0.70,0.07) & (8,32) &6.30e-13 & 0.19 & 7.43e-05 & 1.26 & 3.04e-04 & 141.97 & 5.60e-01 & 259.33  \cr 
pmBB84 & (0.70,0.09) & (8,32) &8.47e-13 & 0.16 & 9.25e-05 & 1.29 & 3.55e-04 & 7.07 & 5.18e-01 & 297.32  \cr 
pmBB84 & (0.90,0.01) & (8,32) &3.68e-13 & 0.18 & 7.07e-06 & 1.33 & 3.27e-04 & 4.98 & 8.60e-01 & 247.82  \cr 
pmBB84 & (0.90,0.03) & (8,32) &1.27e-12 & 0.18 & 2.29e-05 & 1.35 & 5.43e-04 & 137.74 & 7.96e-01 & 230.35  \cr 
pmBB84 & (0.90,0.05) & (8,32) &1.36e-12 & 0.16 & 4.62e-05 & 1.35 & 5.96e-04 & 72.04 & 7.38e-01 & 291.83  \cr 
pmBB84 & (0.90,0.07) & (8,32) &5.13e-13 & 0.17 & 7.31e-05 & 1.35 & 6.25e-04 & 40.11 & 6.84e-01 & 235.27  \cr 
pmBB84 & (0.90,0.09) & (8,32) &7.84e-13 & 0.19 & 1.06e-04 & 1.32 & 7.39e-04 & 142.32 & 6.32e-01 & 244.81  \cr 
 \hline
\end{tabular}

%% file: table/testallTHREE.tex
\begin{tabular}{|ccc||cc||cc||cc||cc|} \hline
\multicolumn{3}{|c||}{\textbf{Problem Data}} & \multicolumn{2}{|c||}{\textbf{Gauss-Newton}} & \multicolumn{2}{|c||}{\textbf{Frank-Wolfe with FR}} & \multicolumn{2}{|c||}{\textbf{Frank-Wolfe w/o FR}} & \multicolumn{2}{|c|}{\textbf{cvxquad with FR}}   \cr\hline
\multicolumn{1}{|c|}{\textbf{protocol}} & \multicolumn{1}{|c|}{\textbf{parameter}} & \multicolumn{1}{|c||}{\textbf{size}} & \multicolumn{1}{|c|}{\textbf{gap}} & \multicolumn{1}{|c||}{\textbf{time}} & \multicolumn{1}{|c|}{\textbf{gap}} & \multicolumn{1}{|c||}{\textbf{time}} & \multicolumn{1}{|c|}{\textbf{gap}} & \multicolumn{1}{|c||}{\textbf{time}} & \multicolumn{1}{|c|}{\textbf{gap}} & \multicolumn{1}{|c|}{\textbf{time}}   \cr\hline
mdiBB84 & (0.50,0.01) & (48,96) &1.25e-12 & 1.27 & 1.64e-05 & 109.60 & 3.75e-04 & 498.82 & 2.11e-01 & 719.99  \cr 
mdiBB84 & (0.50,0.03) & (48,96) &8.37e-13 & 0.81 & 3.53e-05 & 107.90 & 5.31e-04 & 2811.15 & 1.95e-01 & 630.89  \cr 
mdiBB84 & (0.50,0.05) & (48,96) &1.14e-12 & 0.83 & 4.99e-05 & 111.08 & 5.22e-04 & 464.06 & 1.82e-01 & 586.22  \cr 
mdiBB84 & (0.50,0.07) & (48,96) &1.35e-12 & 1.08 & 4.87e-05 & 335.57 & 4.60e-04 & 2169.93 & 1.71e-01 & 569.14  \cr 
mdiBB84 & (0.50,0.09) & (48,96) &1.25e-12 & 0.91 & 8.27e-05 & 342.04 & 4.37e-04 & 829.61 & 1.60e-01 & 568.97  \cr 
mdiBB84 & (0.70,0.01) & (48,96) &1.20e-12 & 0.68 & 2.20e-05 & 5.68 & 5.53e-04 & 901.29 & 3.79e-01 & 719.27  \cr 
mdiBB84 & (0.70,0.03) & (48,96) &4.24e-13 & 1.05 & 1.11e-04 & 5.83 & 1.21e-03 & 256.79 & 3.55e-01 & 670.49  \cr 
mdiBB84 & (0.70,0.05) & (48,96) &1.06e-12 & 1.10 & 6.92e-05 & 340.21 & 1.75e-03 & 865.68 & 3.31e-01 & 651.44  \cr 
mdiBB84 & (0.70,0.07) & (48,96) &5.71e-13 & 0.95 & 1.37e-04 & 337.17 & 1.55e-03 & 864.25 & 3.09e-01 & 604.94  \cr 
mdiBB84 & (0.70,0.09) & (48,96) &1.57e-13 & 0.92 & 1.61e-04 & 347.21 & 2.24e-03 & 872.37 & 2.88e-01 & 604.71  \cr 
mdiBB84 & (0.90,0.01) & (48,96) &8.44e-13 & 0.66 & 4.42e-05 & 343.38 & 3.21e-03 & 175.20 & 5.53e-01 & 710.65  \cr 
mdiBB84 & (0.90,0.03) & (48,96) &1.39e-12 & 0.90 & 9.15e-05 & 343.08 & 3.66e-03 & 301.29 & 5.16e-01 & 671.95  \cr 
mdiBB84 & (0.90,0.05) & (48,96) &9.88e-13 & 0.84 & 1.73e-04 & 345.79 & 4.64e-03 & 519.80 & 4.85e-01 & 644.23  \cr 
mdiBB84 & (0.90,0.07) & (48,96) &2.96e-13 & 1.05 & 2.04e-04 & 338.15 & 2.85e-03 & 148.24 & 4.57e-01 & 580.20  \cr 
mdiBB84 & (0.90,0.09) & (48,96) &5.21e-13 & 1.02 & 2.52e-04 & 343.95 & 3.26e-03 & 193.04 & 4.31e-01 & 595.07  \cr 
 \hline
\end{tabular}

%% file: table/testallFOUR.tex
\begin{tabular}{|ccc||cc||cc||cc|} \hline
\multicolumn{3}{|c||}{\textbf{Problem Data}} & \multicolumn{2}{|c||}{\textbf{Gauss-Newton}} & \multicolumn{2}{|c||}{\textbf{Frank-Wolfe with FR}} & \multicolumn{2}{|c|}{\textbf{Frank-Wolfe without FR}}   \cr\hline
\multicolumn{1}{|c|}{\textbf{protocol}} & \multicolumn{1}{|c|}{\textbf{parameter}} & \multicolumn{1}{|c||}{\textbf{size}} & \multicolumn{1}{|c|}{\textbf{gap}} & \multicolumn{1}{|c||}{\textbf{time}} & \multicolumn{1}{|c|}{\textbf{gap}} & \multicolumn{1}{|c||}{\textbf{time}} & \multicolumn{1}{|c|}{\textbf{gap}} & \multicolumn{1}{|c|}{\textbf{time}}   \cr\hline
TFQKD & (0.75,50.00,0.70) & (12,24) &8.45e-13 & 1.33 & 2.72e-09 & 2.04 & 1.83e-03 & 155.08  \cr 
TFQKD & (0.75,100.00,0.70) & (12,24) &1.42e-12 & 0.86 & 3.75e-09 & 1.35 & 1.53e-03 & 153.57  \cr 
TFQKD & (0.75,150.00,0.70) & (12,24) &9.94e-13 & 0.78 & 2.82e-09 & 1.39 & 7.82e-04 & 162.13  \cr 
TFQKD & (0.75,200.00,0.70) & (12,24) &1.15e-12 & 1.02 & 3.98e-09 & 1.37 & 7.19e-04 & 155.11  \cr 
TFQKD & (0.75,250.00,0.70) & (12,24) &6.79e-13 & 0.79 & 8.21e-09 & 1.32 & 3.14e-04 & 172.53  \cr 
TFQKD & (0.80,50.00,0.70) & (12,24) &1.15e-12 & 0.77 & 2.82e-09 & 1.39 & 1.52e-03 & 156.17  \cr 
TFQKD & (0.80,100.00,0.70) & (12,24) &1.15e-12 & 1.08 & 2.60e-09 & 1.33 & 1.57e-03 & 156.11  \cr 
TFQKD & (0.80,150.00,0.70) & (12,24) &1.25e-12 & 0.94 & 2.97e-09 & 1.29 & 8.30e-04 & 158.12  \cr 
TFQKD & (0.80,200.00,0.70) & (12,24) &9.23e-13 & 0.73 & 4.23e-09 & 1.32 & 5.60e-04 & 155.45  \cr 
TFQKD & (0.80,250.00,0.70) & (12,24) &3.91e-13 & 0.55 & 2.22e-09 & 1.35 & 1.97e-04 & 164.83  \cr 
TFQKD & (0.90,50.00,0.70) & (12,24) &8.08e-13 & 0.79 & 4.30e-09 & 1.30 & 1.55e-03 & 156.38  \cr 
TFQKD & (0.90,100.00,0.70) & (12,24) &1.38e-12 & 0.37 & 3.62e-09 & 1.31 & 1.23e-03 & 154.08  \cr 
TFQKD & (0.90,150.00,0.70) & (12,24) &1.14e-12 & 0.61 & 2.87e-09 & 1.28 & 6.02e-04 & 161.73  \cr 
TFQKD & (0.90,200.00,0.70) & (12,24) &1.04e-12 & 0.34 & 3.98e-09 & 1.36 & 1.68e-04 & 2.63  \cr 
TFQKD & (0.90,250.00,0.70) & (12,24) &5.77e-13 & 0.43 & 2.77e-09 & 1.37 & 1.08e-05 & 2.07  \cr 
TFQKD & (0.95,50.00,0.70) & (12,24) &9.84e-13 & 0.77 & 4.00e-09 & 1.43 & 1.38e-03 & 156.09  \cr 
TFQKD & (0.95,100.00,0.70) & (12,24) &1.06e-12 & 0.70 & 4.36e-09 & 1.41 & 7.47e-04 & 153.19  \cr 
TFQKD & (0.95,150.00,0.70) & (12,24) &1.36e-12 & 0.71 & 3.92e-09 & 1.35 & 8.38e-04 & 149.27  \cr 
TFQKD & (0.95,200.00,0.70) & (12,24) &7.63e-13 & 0.65 & 2.60e-09 & 1.36 & 3.13e-04 & 150.75  \cr 
TFQKD & (0.95,250.00,0.70) & (12,24) &1.02e-12 & 0.64 & 3.61e-09 & 1.42 & 5.56e-06 & 1.81  \cr 
 \hline
\end{tabular}

%% file: table/testallFIVE.tex
\begin{tabular}{|ccc||cc||cc||cc|} \hline
\multicolumn{3}{|c||}{\textbf{Problem Data}} & \multicolumn{2}{|c||}{\textbf{Gauss-Newton}} & \multicolumn{2}{|c||}{\textbf{Frank-Wolfe with FR}} & \multicolumn{2}{|c|}{\textbf{Frank-Wolfe without FR}}   \cr\hline
\multicolumn{1}{|c|}{\textbf{protocol}} & \multicolumn{1}{|c|}{\textbf{parameter}} & \multicolumn{1}{|c||}{\textbf{size}} & \multicolumn{1}{|c|}{\textbf{gap}} & \multicolumn{1}{|c||}{\textbf{time}} & \multicolumn{1}{|c|}{\textbf{gap}} & \multicolumn{1}{|c||}{\textbf{time}} & \multicolumn{1}{|c|}{\textbf{gap}} & \multicolumn{1}{|c|}{\textbf{time}}   \cr\hline
DMCV & (10.00,60.00,0.05,0.35) & (44,176) &2.71e-09 & 1016.19 & 4.35e-06 & 612.62 & 3.57e-06 & 919.84  \cr 
DMCV & (10.00,120.00,0.05,0.35) & (44,176) &2.70e-09 & 1090.41 & 2.27e-06 & 216.61 & 2.16e-06 & 277.47  \cr 
DMCV & (10.00,180.00,0.05,0.35) & (44,176) &2.87e-09 & 1173.14 & 1.90e-07 & 23.64 & 1.36e-07 & 32.50  \cr 
DMCV & (11.00,60.00,0.05,0.35) & (48,192) &3.05e-09 & 1419.71 & 2.50e-06 & 768.52 & 1.43e-06 & 1095.78  \cr 
DMCV & (11.00,120.00,0.05,0.35) & (48,192) &3.24e-09 & 1484.12 & 2.35e-06 & 261.89 & 2.15e-06 & 380.63  \cr 
DMCV & (11.00,180.00,0.05,0.35) & (48,192) &3.40e-09 & 1796.16 & 2.53e-07 & 26.98 & 1.59e-07 & 38.62  \cr 
DMCV & (10.00,150.00,0.02,0.70) & (44,176) &2.07e-09 & 1143.23 & 1.67e-06 & 83.32 & 1.82e-06 & 99.96  \cr 
DMCV & (10.00,200.00,0.02,0.70) & (44,176) &1.96e-09 & 1214.62 & 7.29e-07 & 21.26 & 7.06e-07 & 28.04  \cr 
DMCV & (10.00,150.00,0.02,0.80) & (44,176) &3.27e-09 & 1113.17 & 1.10e-06 & 105.03 & 1.85e-06 & 106.68  \cr 
DMCV & (10.00,200.00,0.02,0.80) & (44,176) &1.63e-09 & 1106.38 & 3.18e-07 & 19.61 & 2.90e-07 & 26.11  \cr 
DMCV & (11.00,150.00,0.02,0.70) & (48,192) &3.31e-09 & 1607.78 & 1.72e-06 & 103.49 & 1.36e-06 & 150.76  \cr 
DMCV & (11.00,200.00,0.02,0.70) & (48,192) &3.05e-09 & 1573.12 & 7.12e-07 & 27.40 & 6.68e-07 & 35.62  \cr 
DMCV & (11.00,150.00,0.02,0.80) & (48,192) &3.37e-09 & 1597.15 & 1.36e-06 & 93.86 & 1.57e-06 & 115.21  \cr 
DMCV & (11.00,200.00,0.02,0.80) & (48,192) &3.38e-09 & 1541.27 & 3.38e-07 & 25.34 & 2.99e-07 & 34.36  \cr 
 \hline
\end{tabular}

%% file: table/testallSIX.tex
\begin{tabular}{|ccc||cc||cc||cc|} \hline
\multicolumn{3}{|c||}{\textbf{Problem Data}} & \multicolumn{2}{|c||}{\textbf{Gauss-Newton}} & \multicolumn{2}{|c||}{\textbf{Frank-Wolfe with FR}} & \multicolumn{2}{|c|}{\textbf{Frank-Wolfe without FR}}   \cr\hline
\multicolumn{1}{|c|}{\textbf{protocol}} & \multicolumn{1}{|c|}{\textbf{parameter}} & \multicolumn{1}{|c||}{\textbf{size}} & \multicolumn{1}{|c|}{\textbf{gap}} & \multicolumn{1}{|c||}{\textbf{time}} & \multicolumn{1}{|c|}{\textbf{gap}} & \multicolumn{1}{|c||}{\textbf{time}} & \multicolumn{1}{|c|}{\textbf{gap}} & \multicolumn{1}{|c|}{\textbf{time}}   \cr\hline
dprBB84 & (1.00,0.08,15.00) & (12,48) &9.42e-13 & 1.66 & 3.85e-06 & 117.56 & 1.03e-04 & 189.23  \cr 
dprBB84 & (1.00,0.08,30.00) & (12,48) &4.92e-13 & 1.70 & 3.85e-06 & 114.97 & 9.43e-05 & 178.82  \cr 
dprBB84 & (1.00,0.14,15.00) & (12,48) &2.96e-13 & 0.91 & 3.63e-04 & 115.55 & 1.16e-02 & 178.91  \cr 
dprBB84 & (1.00,0.14,30.00) & (12,48) &5.21e-13 & 1.43 & 2.60e-04 & 1.89 & 8.31e-03 & 2.51  \cr 
dprBB84 & (2.00,0.08,15.00) & (24,96) &1.10e-12 & 22.06 & 9.58e-05 & 41.11 & 2.11e-03 & 82.96  \cr 
dprBB84 & (2.00,0.08,30.00) & (24,96) &9.58e-13 & 22.26 & 1.17e-04 & 26.10 & 7.32e-06 & 114.66  \cr 
dprBB84 & (2.00,0.14,15.00) & (24,96) &1.35e-12 & 24.93 & 1.89e-05 & 7.56 & 5.11e-04 & 16.59  \cr 
dprBB84 & (2.00,0.14,30.00) & (24,96) &1.04e-12 & 24.63 & 5.71e-06 & 20.79 & 5.38e-06 & 42.06  \cr 
dprBB84 & (2.00,0.14,30.00) & (24,96) &1.04e-12 & 24.82 & 5.71e-06 & 20.19 & 5.38e-06 & 44.72  \cr 
dprBB84 & (3.00,0.08,15.00) & (36,144) &1.38e-12 & 129.39 & 2.36e-04 & 12.94 & 7.41e-03 & 33.59  \cr 
dprBB84 & (3.00,0.08,30.00) & (36,144) &6.33e-13 & 139.34 & 2.26e-04 & 12.54 & 7.04e-03 & 36.53  \cr 
dprBB84 & (3.00,0.14,15.00) & (36,144) &1.30e-12 & 118.84 & 4.32e-05 & 41.49 & 1.43e-04 & 50.07  \cr 
dprBB84 & (3.00,0.14,30.00) & (36,144) &3.32e-13 & 127.94 & 5.80e-06 & 11.32 & 5.74e-06 & 37.63  \cr 
dprBB84 & (4.00,0.08,15.00) & (48,192) &6.98e-09 & 766.17 & 2.88e-04 & 61.19 & 8.10e-03 & 235.02  \cr 
dprBB84 & (4.00,0.08,30.00) & (48,192) &2.13e-09 & 786.01 & 2.97e-04 & 21.23 & 8.45e-03 & 232.35  \cr 
dprBB84 & (4.00,0.14,15.00) & (48,192) &2.85e-12 & 539.08 & 1.29e-04 & 18.50 & 3.85e-03 & 201.96  \cr 
dprBB84 & (4.00,0.14,30.00) & (48,192) &1.17e-12 & 545.04 & 1.26e-04 & 25.67 & 3.73e-03 & 209.05  \cr 
 \hline
\end{tabular}